\documentclass[journal,12pt,onecolumn,draftclsnofoot]{IEEEtran}
\usepackage{mathrsfs}
\usepackage[noadjust]{cite}
\IEEEoverridecommandlockouts
\usepackage{epsfig}
\usepackage{amssymb,latexsym,amsfonts,amsmath}
\usepackage{graphicx}
\usepackage{subfigure}
\usepackage{color}
\usepackage{url}
\usepackage{amsmath}
\usepackage{amsthm}
\usepackage{amssymb}
\usepackage{epstopdf}
\newtheorem{theorem}{Theorem}
\newtheorem{lemma}{Lemma}
\newtheorem{proposition}{Proposition}

\newtheorem{corollary}{Corollary}
\newtheorem{definition}{Definition}

\newtheorem{remark}{Remark}

\newtheorem{example}{Example}

\newcommand\numberthis{\addtocounter{equation}{1}\tag{\theequation}}

\begin{document}
\title
{Resilient Learning-Based Control for Synchronization of Passive Multi-Agent Systems under Attack}
\author{Arash Rahnama and Panos J. Antsaklis,~\IEEEmembership{Fellow,~IEEE}%
\thanks{Arash Rahnama and Panos J. Antsaklis are with the Department of Electrical Engineering, University of Notre Dame, Notre Dame, IN 46556, USA (e-mail:~\{arahnama,~antsaklis.1\} @nd.edu). The support of the National Science Foundation under Grant No. CNS-1035655 and CNS-1446288 is gratefully acknowledged.}}
\IEEEpubid{0000--0000/00\$00.00 ~\copyright ~2016 IEEE example}
\maketitle
\begin{abstract}
	In this paper, we show synchronization for a group of output passive agents that communicate with each other according to an underlying communication graph to achieve a common goal. We propose a distributed event-triggered control framework that will guarantee synchronization and considerably decrease the required communication load on the band-limited network. We define a general Byzantine attack on the event-triggered multi-agent network system and characterize its negative effects on synchronization. The Byzantine agents are capable of intelligently falsifying their data and manipulating the underlying communication graph by altering their respective control feedback weights. We introduce a decentralized detection framework and analyze its steady-state and transient performances. We propose a way of identifying individual Byzantine neighbors and a learning-based method of estimating the attack parameters. Lastly, we propose learning-based control approaches to mitigate the negative effects of the adversarial attack.
\end{abstract}
\section{Introduction}\label{sec:intro}
Distributed coordination of multi-agent systems has been discussed extensively in control, communication and computer science literature. The wide range of applications in this area includes multiple robot coordination \cite{schneider1998territorial}, cooperative control of vehicle formations \cite{fax2004information}, flocking \cite{vicsek1995novel} and spacecraft formation flying \cite{beard2001coordination}. A strong body of literature exists on the state synchronization of homogeneous multi-agent systems with identical dynamics. In many practical applications of multi-agent systems, however, individual systems may have different dynamics with different state-space dimensions. This has instigated the need for the design of output-based control frameworks which do not require the full knowledge of dynamic states and the focus on synchronization of multi-agent systems with different dynamics based on their output information. The problem of synchronization naturally arises when a group of networked agents are seeking output-based agreement according to a certain quantity of interest that depends on the overall goal of the multi-agent system. More specifically, synchronization for a multi-agent system is defined as the agents following a desired output behavior that is achieved thorough local cooperation of neighboring agents. This cooperation is based on a feedback mechanism consisting of a weighted sum of the differences of the outputs of the neighboring agents. Some examples of systems under cooperative control resulting in sophisticated dynamic patterns which cannot be achieved by individual members are migration (or flocking), swarming, and torus. 

There exists a large body of valuable works in the area of synchronization and control. The problem of synchronization for multi-agent systems with dynamic communication edges has been explored in \cite{xiang2017cooperative}. Adaptive synchronization of diffusively coupled systems is discussed in \cite{shafi2015adaptive}. Synchronization of multi-agent systems that are physically coupled is discussed in \cite{wang2017coordination}. Another interesting sub-field of synchronization in multi-agent system consists of leader-follower synchronization problems, such works include \cite{xiang2017cooperative, ni2010leader,zheng2011consensus,klotz2015asymptotic}. The relationship amongst dissipativity, passivity and output synchronization has been explored in the literature as well \cite{liu2015output,yu2014output,scardovi2010synchronization}. Synchronization under switching topologies is discussed in \cite{wang2017cooperative}. Cluster-based synchronization in which only the synchronizations of separate clusters are achieved is discussed in \cite{liu2017cluster}. Some of the other recent notable works in the area of synchronization in multi-agent systems are given in \cite{liu2012synchronization,liu2017consensus,zhang2014synchronization,lu2014synchronization, chen2014pattern,su2016distributed,lu2015synchronization,wen2016event,niu2017adaptive}. 

In none of the works above, the problem of security and the negative effects of malicious attacks on synchronization have been discussed. In this work, we consider the effects of a Byzantine attack on the multi-agent network. Byzantine attacks were first proposed by \cite{lamport1982byzantine} and may cover different types of malicious behaviors \cite{zhang2015byzantine}. In our work, Byzantine agents intelligently falsify their data \textemdash Similar to the adversaries defined in \cite{fragkiadakis2013survey,vempaty2011adaptive,rawat2011collaborative}. The Byzantine agents are assumed to be powerful in the sense that they have the complete knowledge of the whole system and can update their information in an arbitrary way and send different data to distinct neighbors at the same time. Additionally, the Byzantine nodes are capable of disturbing the structure of the underlying communication graph by manipulating their feedback weights \textemdash The communication graph is usually required to meet certain conditions for synchronization to happen \cite{liu2015output,yu2014output,scardovi2010synchronization}. Lastly, we propose a distributed method of detection and mitigation as opposed to the more common centralized methods where a fusion center takes upon itself the responsibility of detecting and mitigating the attacks. There is obviously always a limitation to this approach as the central fusion unit may be compromised as well. Our proposed distributed detection and mitigation framework will eliminate this possibility. In the consensus literature, the decentralized method of detection has been proposed in works such as \cite{marano2006distributed,abdelhakim2011reliable,chen2008robust,marano2009distributed}. In \cite{marano2009distributed} for example, it is assumed that through collaboration, the
Byzantine agents are aware of the true hypothesis, which is similar to the assumption we make in the present work. As another example, in \cite{chen2008robust}, the authors rely on a sequential decentralize probability ratio test that is modified via a reputation-based
mechanism in order to filter out the false data and only accept reliable messages. Lastly, most detection and mitigation frameworks in the literature rely on exclusion of Byzantine agents from the synchronization algorithm \cite{yu2009defense,liu2012adaptive}. For example, in \cite{yan2012vulnerability}, the authors propose an adaptive outlier detection framework, based on which, the outside of the bound received information are excluded from the consensus process. In our work, we propose a mitigation scheme that takes advantage of the falsified information received from the Byzantine agents and mitigates the effects of the attack without excluding the Byzantine neighbors. This is due to the fact that excluding the Byzantine agents usually is not the best practice as most synchronization algorithms \cite{liu2015output,yu2014output,scardovi2010synchronization}, rely on balancedness and connectedness of the underlying communication graph and exclusion of Byzantine agents may contradict these conditions.  

Our framework is based on each individual agent locally deciding, based on its local test statistics that contain the information received by the agent from its neighbors, whether the entire multi-agent system has reached synchronization. We also show synchronization for an event-triggered control framework. This is motivated by the fact that event-triggered control frameworks can considerably reduce communication and computation load on the band-limited communication network \cite{wang2011event}. Additionally, it has been shown that event-based control methods can maintain the same performance index as their continuous and periodic based control counter-parts \cite{dimarogonas2012distributed,seyboth2013event}. First, we show that, under no attack, the entire event-triggered multi-agent network system is capable of reaching synchronization and that each agent may decide correctly on synchronization based on their local summary statistics, if our proposed triggering-based control framework and the underlying communication graph meet certain conditions. Next, we propose a method of identifying Byzantine agents based on the statistical distribution of Byzantine agents' outputs. We characterize and analyze the performance of the detection unit. Lastly, we propose a method of mitigation for the attacks in order to maintain the synchronization of the entire event-triggered multi-agent network system. In this vein, the contributions of our work are listed below,
\begin{itemize}
	\item We show synchronization for an event-triggered multi-agent network system with output passive agents. We introduce a local decision making process based on which each individual agent decides whether the entire system has reached synchronization or not.
	\item We propose a simple design-oriented event-triggering control framework based on simple output-based triggering conditions which guarantees synchronization and positive lower-bounds for the inner-event time-instances (lack of Zeno behavior).
	\item We define a rather general Byzantine attack framework, and characterize the effects of the attack on passive qualities of the multi-agent system in particular and synchronization of the entire system in general.
	\item We introduce a decentralized detection framework for detecting the Byzantine attack.
	\item We analyze the performance of the proposed detection framework. We characterize both the steady-state and transient performance of the detection framework.
	\item We propose a specific method of identifying individual Byzantine neighbors and learning their attack parameters.
	\item Lastly, we introduce two different learning-based mitigation processes; one based on the passive properties of the agents, and one based on the statistical distribution of the data received from the neighboring agents. Based on which, we propose a learning-based control framework that can considerably mitigate the negative effects of the attack. 	
\end{itemize}
\section{Mathematical and Statistical Preliminaries} \label{sec:prel}
Consider the dynamical system $G$, 
\begin{align*}
G:\begin{cases}
\dot{x}(t)= f(x(t),u(t)) 
&\\y(t)= h(x(t),u(t)),
\end{cases}
\end{align*}
where $f$ and $h$ are Lipschitz functions, $x(t) \in X \subset R^n$, and $u(t) \in U \subset R^m$, and $y(t) \in Y \subset R^m$  are respectively the state, input and output of the system, and $X$, $U$ and $Y$ are the state, input and output spaces. 
\begin{definition}\label{supply}
(\cite{willems1972dissipative}) The supply rate $\omega(u(t),y(t))$ is a well-defined supply rate, if for all $t_0$, $t_1$ where $t_1\geq t_0$, and all solutions $x(t)\in X$, $u(t)\in U$, and $y(t)\in Y$ of the dynamical system, we have,
\begin{align*}
\int_{t_0}^{t_1}|\omega(u(t),y(t))| dt < \infty.
\end{align*}
\end{definition}

Dissipativity and passivity are energy-based notions that characterize a dynamical system by its input/output behavior. A system is dissipative if the change in the system's stored energy is upper-bounded by the energy supplied to the system. The energy supplied to the system is mathematically modeled by the supply function, and the energy stored in the system is mathematically modeled by the storage function.
\begin{definition} \label{diss}
(\cite{willems1972dissipative}) System $G$ is dissipative with respect to the well-defined supply rate $\omega(u(t),y(t))$, if there exists a nonnegative storage function $V(x): X \to R^+$ such that for all $t_0$, $t_1$ where $t_1\geq t_0$, and all solutions $x(t)\in X$, $u(t)\in U$, and $y(t)\in Y$ of the dynamical system,
\begin{align*}
V(t_1)-V(t_0) \leq \int_{t_0}^{t_1}\omega(u(t),y(t)) dt, 
\end{align*}
is satisfied. If the storage function is differentiable, we have,
\begin{align*}
\dot{V}(t) \leq \omega(u(t),y(t)),~ \forall t\geq0.
\end{align*} 
\end{definition}
\begin{definition} \label{pass}
(\cite{bao2007process}) As a special case of dissipativity, system $G$ is called passive, if there exists a nonnegative storage function $V(x): X \to R^+$ such that,
\begin{align*}
V(t_1)-V(t_0) \leq \int_{t_0}^{t_1} u^T(t)y(t) dt
\end{align*}
is satisfied for all $t_0$, $t_1$ where $t_1\geq t_0$, and all solutions $x(t)\in X$, $u(t)\in U$, and $y(t)\in Y$ of the dynamical system. 
\end{definition}

\begin{definition}	\label{OPD}
(\cite{khalil2002nonlinear}) System $G$ is considered to be Output Feedback Passive (OFP), if it is dissipative with respect to the well-defined supply rate,
\begin{align*}
\omega(u,y)=u^Ty-\rho y^Ty,
\end{align*}
for some $\rho \in R$. Additionally, if the storage function is differentiable, we may have,
\begin{align*}
\dot{V}(t) \leq u^Ty-\rho y^Ty.
\end{align*}
\end{definition}
The above definition presents a more general form for the concept of passivity. Based on Definition \ref{OPD}, we can denote an output passive system based on its output passivity index. If $\rho<0$ then the system has a shortage of passivity. A positive value for the passivity index $\rho$ indicates an excess in passivity. If $\rho>0$, then the system is called \emph{output strictly passive} (OSP). 
\begin{definition}(\cite{khalil2002nonlinear})
System $G$ is called finite-gain $L_2$-stable, if for the smallest possible positive gain $\gamma$, and $\forall u(t)\in U$, a $\beta$ exists such that over the time interval $[0,\tau]$ and for any positive $\tau$, we have,
\begin{align*}
&||y_{\tau}||_{L2} \leq \gamma||u_{\tau}||_{L2} + \beta.
\end{align*}
Here, $||y_{\tau}||_{L2}$ and $||u_{\tau}||_{L2}$ represent the $L_2$-norm of truncated signals over the time interval $[0,\tau]$. For instance,
\begin{align*}
&||y_{\tau}||_{L2} = \sqrt{\int_{0}^{\tau}y^T(t)y(t)dt}.
\end{align*}
\end{definition}
In probability theory, the expected value $(E[X])$ of a random variable $X$, intuitively, is the long-run average value of repetitions of the experiment it represents, in the continuous sense, this is defined as, 
\begin{align*}
E[X]=\int_{-\infty}^{+\infty} x f_{PDF}(x) dx.
\end{align*}
The notation $f_{PDF}(.)$ represents the probability density function (PDF) of a distribution. Expectation of the random variable $X$ conditioned on the hypothesis (or random distribution) $H$ is represented as $E[X|H]$. The complementary distribution function of the standard normal Gaussian distribution with zero mean $(\mu=0)$ and standard deviation $\sigma=1$ is denoted as $Q(z)=\frac{\int_{z}^{\infty} e^{\frac{-t^2}{2}}dt}{\sqrt{2\pi}}$. The Gaussian distribution with mean $\mu$ and variance $\sigma^2$ is denoted as $\phi(x|\mu,\sigma^2)=\frac{ e^{\frac{-(x-\mu)^2}{2\sigma^2}}}{\sigma\sqrt{2\pi}}$. Null and alternative hypotheses are represented as $H_0$ and $H_1$. Probability of an event is represented as $P$. Probability of false alarm (type 1 error) or accepting the alternative hypothesis and rejecting the null hypothesis mistakenly is shown as $P_{FA}=Pr(D=H_1|H_0)$ and probability of detection is $P_D=Pr(D=H_1|H_1)$. 
\section{The Communication Graph Model} \label{sec:gra}
The communication flow between agents may be represented as a weighted directed graph \cite{godsil2013algebraic}. A graph is directed, if its edges have direction. We consider a finite positively weighted directed graph $G:=(V,E)$ with no loops and with the adjacency matrix $A$, where the entry $a_{i,j} \neq 0$, if there is a directed edge from vertex $i$ to vertex $j$, otherwise $a_{i,j} = 0$. The adjacency matrix $A$ represents both the link weights and the topology of the graph. $V$ is the vertex set including all vertices (all $N$ agents), $V=\{1,2,...,N\}$. $E$ is the edge set including all edges (communication links), $E \subset V \times V$. The agent $G_i$ can send information to agent $G_j$, if $(i,j) \in E$ and $a_{i,j} \neq 0$. The in-degree of a vertex $j$ is given by $d_{in}(j)= \sum_{j} a_{kj}$ and the out-degree of a vertex $j$ is given by $d_{out}(j)= \sum_{j} a_{jk}$ where $k$ respectively represents the in-neighbor ($V_{in}(j)=\{k\in V_{in}(j)|(k,j) \in E\}$) and out-neighbor ($V_{out}(j)=\{k\in V_{out}(j)|(j,k) \in E\}$) agents that have a communication link in common with agent $j$. We introduce the diagonal degree matrix $D^{N \times N}$ with $d_{j,j}=d_{out}(j), \forall j \in V$. The weighted Laplacian matrix $L$ of the graph is defined as $L=D-A$. For the graph presented in Fig. \ref{fig:Graph}, we have,
\begin{figure}[!t]
\centering
\includegraphics[scale = 0.5]{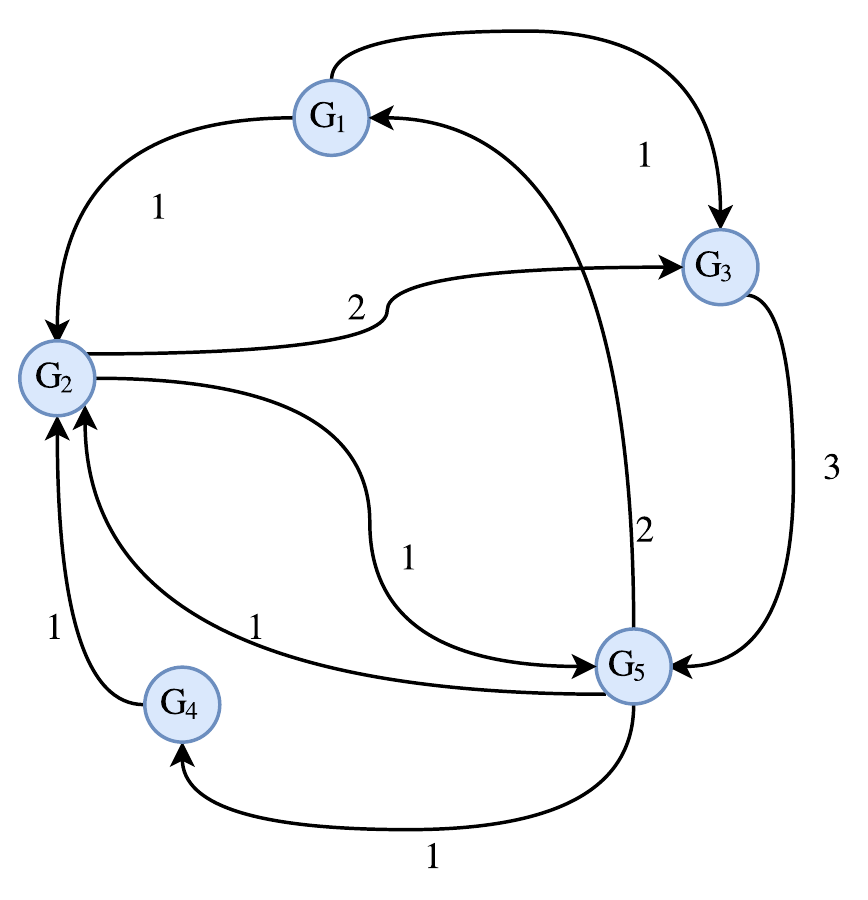}
\caption{Graph - An Example.}
\label{fig:Graph}
\end{figure}
\begin{align*} 
A=\begin{bmatrix}
0    &  1   &  1  & 0 &  0 \\
0    &  0   &  2  & 0 &  1 \\
0    &  0   &  0  & 0 &  3 \\
0    &  1   &  0  & 0 &  0 \\
2    &  1   &  0  & 1 &  0 \\
\end{bmatrix},~
D=\begin{bmatrix} 
2    &  0   &  0  & 0 &  0 \\
0    &  3   &  0  & 0 &  0 \\
0    &  0   &  3  & 0 &  0 \\
0    &  0   &  0  & 1 &  0 \\
0    &  0   &  0  & 0 &  4 \\
\end{bmatrix},~
L=\begin{bmatrix} 
2    &  -1  &  -1 & 0 &  0 \\
0    &  3   &  -2 & 0  &  -1\\
0    &  0   &  3  & 0  & -3 \\
0    &  -1  &  0  & 1  &  0 \\
-2   &  -1  &  0  & -1 &  4 
\end{bmatrix}.
\end{align*}
\begin{definition}
\cite{godsil2013algebraic} A vertex is balanced, if its in-degree is equal to its out-degree. A directed wighted graph is balanced, if all of its vertices are balanced.
\end{definition}
It is important to note that the followings hold for a balanced directed graph, $1_N^TL=0$ and $L^T1_N=0$, where $1_N=[1,...,1]^T$ is a vector of size $N$. The graph presented in Fig. \ref{fig:Graph} is balanced.
\begin{definition}
\cite{godsil2013algebraic} A path of length $r$ in a directed graph is a sequence of $r+1$ distinct vertices $\{v_0,v_1,...,v_r\}$ such that for every $i \in \{0,...,r+1\}$, $(v_i,v_{i+1})$ is an edge. A weak path is a sequence of $r+1$ distinct vertices $\{v_0,v_1,...,v_r\}$ such that for every $i \in \{0,...,r+1\}$, either $(v_i,v_{i+1})$ or $(v_{i+1},v_i)$ is an edge. A directed graph is weakly connected if any two vertices can be joined by a weak path.
\end{definition}
\begin{definition}
\cite{godsil2013algebraic} A directed graph is  connected, if for any pair of distinct vertices $v_i$ and $v_j$, there is a weak path from $v_i$ to $v_j$. A directed graph is strongly connected, if for any pair of distinct vertices $v_i$ and $v_j$, there is a directed path from $v_i$ to $v_j$. 
\end{definition}
The connectivity measures of directed graphs are related to the algebraic properties of their Laplacian matrices \cite{wu2005algebraic}.
\begin{definition}
\cite{wu2005algebraic} For a directed graph $G$ with the Laplacian matrix $L$, the algebraic connectivity is a real number defined as  
\begin{align*}
\lambda(G) := min_{z \in P}  ~z^TLz,
\end{align*}
where $P=\{z \in R^N:z\bot1_N,||z||=1\}$.
\end{definition}
For a balanced connected graph G with nonnegative weights and Laplacian matrix $L$, we have $\lambda(G)=\gamma_2(\frac{L+L^T}{2})>0$, where $\gamma_2$ is the second smallest eigenvalue of the matrix $\frac{L+L^T}{2}$  ($\gamma_1=0$) \cite{wu2005algebraic}. Lastly, we define $\mathcal{N}_j^{in}$ and $\mathcal{N}_j^{out}$. $\mathcal{N}_j^{in}$ denotes the set of all neighboring nodes that send information to agent $G_j$ including the weights associates with their communication graph topology. $\mathcal{N}_j^{out}$ denotes the set of all neighboring nodes that receive information from agent $G_j$ including the weights associates with their communication graph topology. For a balanced graph, the cardinality of these two are equal  $|\mathcal{N}_j^{out}|=|\mathcal{N}_j^{in}|$. For instance, for the graph presented in Fig. \ref{fig:Graph}, we have: $\mathcal{N}_5^{in}=\{1G_2,3G_3\}$ and $|\mathcal{N}_5^{in}|=|\mathcal{N}_5^{out}|=4$. 
\section{Problem Statement} \label{sec:prob}
We consider the problem of synchronization for a multi-agent system consisting of $N$ agents under an event-triggered network control framework. We assume that agents are output passive,
\begin{align*}
\dot{V}_j(t) \leq u_j^T(t)y_j(t)-\rho_jy^T_j(t)y_j(t),~ \forall t>0~for~ j=1,...N.
\end{align*}
We consider an efficient event-based framework where agents communicate with each other only when necessary. In other words, agent $G_j$ sends new information to its neighboring agents when the last information sent to other agents is outdated and requires a new modification based on $G_j$'s current dynamics and the event-triggering condition. This considerably decreases the communication load on the shared network. Consequently, it is assumed that the agents that will receive the new information from $G_j$ will update their control inputs accordingly. Each agent establishes a new communication attempt with its neighboring agents over a band-limited networks when its triggering condition is met. The triggering conditions are output-based and simple to design,
\begin{align} \label{eq:trig}
&||e_j(t)||_2^2 > \delta_j||y_j(t)||_2^2.
\end{align}

The event-detector is located on the output of each agent to monitor the behavior of its output. An updated measure of $y_j$ is sent to the communication network when the error between the last information sent ($y_{t_k})$) and the current one, $e_j(t)=y_j(t)-y_j(t_{k})$ (for $t \in [t_{k}, t_{k+1})$) exceeds a predetermined threshold established by the designer based on the relation given in Eq. \ref{eq:trig} and the design parameter $\delta_j$. At instances for which the triggering condition is met, and new information is successfully exchanged and the error is set back to zero, $e_j(t_{k+1})=0$. These simple triggering conditions will facilitate the design process by making it easier for the designer to understand and analyze the trade-offs amongst synchronization, performance and communication load. Each agent has its own respective sampler condition which is designed based on its passivity properties and its location in the underlying communication graph. This will be analytically presented in Section \ref{sec:main}. Theorem \ref{thm:synch} outlines the design condition for each $\delta_j$. The control input for each agent is represented by the summation of the differences between the agent's output and the output of its neighboring agents multiplied by respective positive control gains,
\begin{align} \label{inp}
u_j=\sum_{k \in  \mathcal{N}_j^{in}}^{} a_k (y_k(t^n_k)-y_j(t^n_j)).
\end{align}
More specifically, the input $u_j$ for agent $G_j$ consists of the summation of $a_k(y_k(t^n_k)-y_j(t^n_j))$, where $y_j(t^n_j)$ represents agent $G_j$'s last output sent to its neighbors, and $y_k(t^n_k)$ represents the last received output from the neighboring agent $k$ where  $k \in  \mathcal{N}_j^{in}$. $a_k>0$ represents a control gain established by agent $G_j$ for each neighboring agent,
\[\begin{cases}
a_k   & \quad \text{if agent $G_j$ receives information from agent $G_k$ } \\
0  & \quad \text{otherwise.}
\end{cases}\]

One can represent the underlying communication graph according to Section \ref{sec:gra}, in which case the control gains $a_k$ represent the arc weights in the graph. The  assumption made here is that during the initialization and design of the gains and  communication links for the entire multi-agent, the underlying communication graph is connected and balanced. We denote the outputs of $N$ agents by the vector $Y=[y_1,y_2,...,y_N]^T$. We define the matrix $\Phi \in R^{(N-1)\times N}$ as follows,
\begin{align} \label{phi}
\Phi=\begin{bmatrix}
-1+(N-1)\nu    &  1-\nu    &   -\nu    & ... &   -\nu \\
-1+(N-1)\nu    & -\nu      & 1-\nu     &\ddots& \vdots\\
\vdots         &\vdots     &\ddots     &\ddots&   -\nu\\
-1+(N-1)\nu    & -\nu      & ...       & -\nu &  1-\nu
\end{bmatrix}
\end{align}
where $\nu=\frac{N-\sqrt{N}}{N(N-1)} \in R$. Matrix $\Phi$ exhibits the following properties: $\Phi1_N=0$, $\Phi\Phi^T=I_{N-1}$, and,
\begin{align*} 
\Phi^T \Phi=\begin{bmatrix}
\frac{N-1}{N} & \frac{-1}{N}  &  ...         & \frac{-1}{N} \\
\frac{-1}{N}  & \frac{N-1}{N} &\ddots        & \vdots       \\
\vdots        &\ddots         &\ddots        &\frac{-1}{N}  \\
\frac{-1}{N}  &...            & \frac{-1}{N} &\frac{N-1}{N}
\end{bmatrix} = I_N - \frac{1}{N}1_N1_N^T.
\end{align*}
To measure synchronization mathematically, we define,
\begin{align}\label{ave}
\bar{Y}=\frac{1}{N}1_N^TY=\frac{1}{N}\sum_{i=1}^{N}y_i,
\end{align}
and,
\begin{align}\label{dev}
Y_\Delta=(y_1-\bar{Y},y_2-\bar{Y},...,y_N-\bar{Y})^T.
\end{align}
$Y_\Delta$ represents a measure for synchronization of agents. $Y_\Delta=0$ only happens when all agents reach the same synchronized state $y_1=y_2=...=y_N=\bar{Y}$.  We have $\Phi^T \Phi Y=(I_N - \frac{1}{N}1_N1_N^T)Y=Y_\Delta$. Further,
\begin{align}\label{Ydelta}
Y^T\Phi^T \Phi \Phi^T \Phi Y=Y_\Delta^TY_\Delta.
\end{align}
Lastly, we can show that, 
\begin{align*}
&Y^TL^TY = (Y_\Delta + \frac{1}{N}1_N1_N^TY) L^TY \\&~~~~~~~~~ = Y_\Delta L^TY = Y_\Delta L^T (Y_\Delta + \frac{1}{N}1_N1_N^TY)\\&~~~~~~~~~=Y_\Delta^TL^TY_\Delta \geq \lambda(G)Y^T\Phi^T \Phi Y=\lambda(G)Y^TY-\frac{\lambda(G)}{N}Y^T1_N1_NY, \numberthis \label{Y}
\end{align*}
where $\lambda(G)$ represents the algebraic connectivity of the underlying communication graph and $L$ is the Laplacian matrix. In Section \ref{sec:main}, we represent the results for synchronization of the entire event-triggered multi-agent system and the design conditions for each event-detector based on the passivity properties of agents and algebraic properties of the communication graph. 
\section{Sensing, Detection and Fusion Frameworks} \label{sec:Det}
The three most popular signal detection approaches for spectrum sensing are matched filtering detection method, feature detection method, and energy detection method \cite{zhang2015byzantine}. Here, we adopt an energy-based detection approach for the detection center on each agent \cite{digham2007energy,urkowitz1967energy}. The energy detector measures the energy in the input wave over a specific time interval. This means that our framework is based on detecting a deterministic signal over a noisy communication channel. The energy detection method, however, cannot differentiate between noise and signal, but at the same time does not need any prior knowledge about the signal's distribution. It is assumed that the detection center makes decisions under a Neyman-Pearson (NP) set-up, and that the adversary is aware of it \cite{bickel2015mathematical}. The local summary statistic of each agent is calculated from the received signal energy from the neighboring agents. As mentioned, at each triggering instance, each agent communicates with its neighbors. In our detection framework, this means that each communication attempt will update the summary statistic of neighboring agents. This process continues until the whole multi-agent system synchronizes to a steady-state. This steady-state represents the global test statistic at which the entire multi-agent system has reached synchronization. At each updating instance, each agent makes a decision whether the entire system has reached synchronization or not. As later defined, this process also decides if a neighboring agent is Byzantine or not. In order to fulfill the premise behind this framework, each agent is equipped with a detection unit that has access to the network topology in order to gain information \cite{urkowitz1967energy}. We explain this in more details in this section.

The signals received by each agent's detection unit are assumed to be unknown in details but deterministic. The band-limited communication environment in which signals travel is known. The noise is assumed to be Gaussian and additive with zero mean. Based on the assumption of a deterministic signal, we know that the input with signal present is Gaussian with a nonzero mean. For agent $G_j$, at time instant $\tau$, the sensed signal received from the neighboring agent $G_k$, $y_k^\tau$ is given as,
\[ y_k^\tau=\begin{cases}
n_k^\tau   & \quad \text{under $H_0$ } \\
\tilde{h}_k s^\tau+n_k^\tau  & \quad \text{under $H_1$,}
\end{cases}\]
where $\tilde{h}_k$ represents the channel gain and $n_k^n$ represents the noise for the communication link from agent $G_k$ to agent $G_j$ ($H_1$ and $H_0$ here represent the hypotheses under which, the signal is present or not). The channel gain in the communication link between each two agents, models the effects of channel shadowing, channel loss and fading. $n_k^\tau$ is additive Gaussian noise with zero mean and variance $\sigma_k^2$ $(\mathcal{N}(0,\sigma_k^2))$. It is assumed that the noise $n_k^\tau$ and signal $s^\tau_k$ are statistically independent. The channel gains $\tilde{h}_k$ and noise variances $\sigma_k^2$ for channels are readily available for each agent. These assumptions are justified by the fact that each detection unit can perform simple noise power estimation and channel gain estimation (by averaging the signal-to-noise ratio over a certain time interval) between consecutive sensing intervals to accurately obtain these values \cite{shen2010deflection}. Additionally, we assume that $\tilde{h}_k$ is considered larger than the estimate value to compensate for any overhead \cite{shen2010deflection}.

It has been shown that control gain designs that compensate for the negative effects of the communication channel $\tilde{h}_k$ comparatively perform  better \cite{visser2008multinode}. As a result, one can design the optimal control gains $a_k$ (explained in details in Section \ref{sec:prob}) according to $a_k=\frac{K_k}{\tilde{h}_k}$ to compensate for channel effects. This is not a necessary rule to follow for the results presented in this paper. This weight design, however, will efficiently assign higher weights to channels with higher Signal-to-Noise ratio (more confidence in the received data) and vice-versa \cite{visser2008multinode}. Lastly, the channel gains are assumed independent of each other, known and constant over each sensing period. This is justified by the slow-changing nature of the communication links where the delay requirement is short compared to the channel coherence time \cite{arka2013selective}. Each agent $G_j$ calculates a local summary statistic $T_k$ over a detection interval of $L$ samples, from the information received from its neighboring agent $G_k$,
\begin{align} \label{ts}
T_k=\sum_{i=1}^{L}|y_k^i-y_j^i|^2.
\end{align}
It can be assumed that $L=2TW$ where $TW$ is an integer representing the time-bandwidth product of the energy detector with $T$ standing for the effective spectrum sensing time-interval and $W$ standing for the bandwidth of the sensing spectrum \cite{zhang2011distributed}. $y_j^i$ represents the last output sent from agent $G_j$ to its neighboring agents at instance $i$, which is also utilized in calculating the local summary statistic $T_k$ over the detection interval of $L$. The energy in a finite number of samples for the local summary statistic can be approximated by the sum of squares of statistically independent Gaussian random variables having certain means ($|y_k^i-y_j^i|$) and equal variances. This sum has a Chi-Square distribution with $L$ degrees of freedom $(\mathcal{X}^2_L)$ in the absence of signal. In the presence of a deterministic signal ($H_1$ hypothesis), the sampling plan yields an approximation to the energy consisting of the sum of squares of random variables, where the sum has a non-central Chi-Square distribution with $L$ degrees of freedom with the non-centrality parameter $\eta_k$,
\[\frac{T_k}{\sigma_k^2} \simeq \begin{cases}
\mathcal{X}^2_L   & \quad \text{under $H_0$ } \\
\mathcal{X}^2_L(\eta_k)  & \quad \text{under $H_1$,}
\end{cases}\]
where $\eta_k=\frac{\sum_{i=1}^{L}|\tilde{h}_k y^i_k-y_j^i|^2}{\sigma_k^2}$.
\subsection{Decision Making Step} \label{sub:decmakste}
Each agent $G_j$ makes local decisions as to whether the entire multi-agent system has reached synchronization or not. The summary statistic for synchronization, given the entire system, may be represented as $T^\star = \sum_{j \in  \mathcal{N}}^{}T_j^*$, where $T_j^*=\sum_{k \in  \mathcal{N}_j^{in}}^{}T^j_{k}$ for $j=1,...,N$. This then can be compared against a threshold $\gamma$ in order to decide if the system has synchronized. If this holds for the entire event-triggered multi-agent system then the entire multi-agent network has synchronized. Each agent $G_j$, however, makes its own decision on the synchronization hypothesis using the predefined threshold $\gamma_j$,
\[ Decision_{syn}=\begin{cases}
H_0   & \quad \text{if $T_j^*<\gamma_j$ } \numberthis \label{epsilon}\\
H_1  & \quad \text{otherwise.}\end{cases}\]
Where, $\gamma_j=\sum_{k \in  \mathcal{N}_j^{in}}^{}L\sigma_k^2+\lambda$ (see (\ref{D1})). $\lambda$ is a positive constant representing the allowed margin of error (or our confidence in the process). The exact choice of $\lambda$ depends on the desired detection and false alarm rates and is beyond the scope of work presented here. This is explained in more details in Section \ref{sec:analdet}. We assume the threshold $\gamma_j$ has already been selected based on performance, detection and false alarm criteria. The relation in (\ref{epsilon}) means that if sums of differences between an honest agent's output and the outputs of all its neighboring agents is small enough, then the honest agent may decide that the multi-agent system has reached synchronization. In other words, the entire event-triggered multi-agent system has reached synchronization, if $T_j^*<\gamma_j$ for $j=1,...,N$. 
\section{Byzantine Attack} \label{sec:Byz}
Multi-agent systems are vulnerable to attacks due their strong reliance on secure communication links and legitimate exchange of information. One of the most common type of such attacks is named Byzantine. Originally, proposed in \cite{lamport1982byzantine}, a Byzantine attack may take different forms \cite{dolev1982byzantine,ho2004byzantine}, our focus in this paper remains with intelligent data-falsification and weight manipulation attacks \cite{vempaty2011adaptive,fragkiadakis2013survey}. The main goals of Byzantine attackers is to first decrease the detection probability and increase the probability of false alarms, and then to degrade the multi-agent system's performance. This makes the problem of the Byzantine attack and defending against it very challenging and complicated. For the Byzantine agents, we adopt an approach that leaves the attacker with more power than usually allowed in practice. This leads to a conservative assessment of security risks but helps with analytical tractability. In this vein, we assume that Byzantine agents in fact know the true hypothesis and they use this knowledge to construct the most effective fictitious data in order to confuse the synchronization goal. This assumption obviously is difficult (but not impossible) to satisfy in practice. For this to be possible, the attackers should have a separate network for the cooperation amongst themselves.

As we will show in Section \ref{sec:main}, for the entire event-triggered multi-agent network system to reach synchronization, a connected balanced communication graph is required. In Section \ref{sec:main}, we also quantify the negative effects of weight manipulation resulting in an unbalanced underlying communication graph. We assume that Byzantine agents attack the multi-agent system from two different angles. First, the Byzantine agents disturb the underlying premise behind the convergence of the multi-agent system by introducing new weights that will undermine the balanced property of the underlying communication graph. Second, the Byzantine agents falsify their own information sent to other honest agents in order to conceal their identity and also to coerce the entire multi-agent system into following their desired behavior. The attack model, we consider is extremely general and covers several different Byzantine plots. To be more specific, if the event-triggered multi-agent network system is designed and initialized according to Theorem \ref{thm:synch} by the designer to reach synchronization, then we assume that at the initialization instance, the Byzantine agents $(\mathcal{N}_B)$ introduce the following fictitious weights $(a_k^\prime)$ into the underlying communication graph,
\begin{align*}
a_k^\prime = a_k + \omega_j~~~ \forall G_j \in \mathcal{N}_B,~and~ \forall G_k \in \mathcal{N}_j^{in}.
\end{align*}  
Additionally, at each communication instance, we assume that the Byzantine agents falsify their information according to,
\begin{align*}
\tilde{y}_j = y_j \pm \Delta_j~~~ \forall G_j \in \mathcal{N}_B.
\end{align*}  
Where $\Delta_j$ may represent the power of the attack inflicted by the Byzantine agent $G_j$. The model presented above allows the Byzantine nodes to manipulate their weights and falsify their information in a completely arbitrary manner based on their desire. As a result, the Byzantine agents are able to conceal themselves while degrading the performance of the entire system.
\subsection{Modeling of the Data Falsification Attack}
The main goal of the Byzantine agents is to manipulate the sensing results in a stealthy way and to reverse the synchronization status. In the presence of a synchronized state, the goal is to "vandalize" and move the multi-agent's state back to the state of lack of synchronization ($H_0 \rightarrow H_1$), and in the absence of synchronization, the goal is to "exploit" and to move the current state to  the state of the presence of synchronization at the desired value set by the Byzantine agents ($H_1 \rightarrow H_0$). This type of data injection attack is adaptive and extremely general. Each Byzantine agent may perform a stealthy manipulation of sensing data independently. The attack is "adaptive", in the sense that the data-falsification is based on the neighbors' states, and with the assumption that the adversary has prior knowledge on the detection algorithm. The attack is "covert", in the sense that the adversary manipulates the sensing data without being detected. Outsider attackers can be effectively expelled from the network with an authentication mechanism. In this work, we focus on insider attackers that reside in legitimate nodes.

Based on the assumption that Byzantine agents are intelligent and know the true hypothesis, we analyze the worst case detection performance of data-falsifications and define the attack devised by the agent $G_i$ as follows,
\[\tilde{y}_i =\begin{cases}
y_i+\Delta_i   &\quad \text{with propabilty $P_i$ }  \quad \text{under $H_0$ } \\
y_i  & \quad \text{with propabilty $1-P_i$ } \quad \text{under $H_0$,}
\end{cases}\]
and, 
\[\tilde{y}_i = \begin{cases}
y_i-\Delta_i   & \quad \text{with propabilty $P_i$ } \quad \text{under $H_1$ } \\
y_i  & \quad \text{with propabilty $1-P_i$ }  \quad \text{under $H_1$,} 
\end{cases} \]
where $P_i$ is the attack probability and $y_i$ is the Byzantine agent's true time-variant output. $\Delta_i$ is a constant value that represents the strength of the attack. $\Delta_i$ is set by the Byzantine agent based on the information it receives from its neighbors and may be positive or negative to fulfill the "exploitation" and "vandalism" objectives. For example, under the hypothesis $H_0$, we may define the test statistics $\eta_i=\frac{\sum_{k=1}^{L}|\tilde{h}_i y^k_i-y_j^k|^2}{\sigma_i^2}\approx 0$. The Byzantine agent by utilizing the attack parameter $\Delta_i>0$ or $\Delta_i<0$ may commit vandalism ($\eta^\prime_i=\frac{\sum_{k=1}^{L}|\tilde{h}_i (y^k_i+\Delta_i)-y_j^k|^2}{\sigma_i^2}\approx L\tilde{h}^2_i \Delta^2_i$, $H_0 \rightarrow H_1$). Under the hypothesis $H_1$, we may define the mean values $\mu_j=\frac{1}{L}\sum_{k=1}^{L}y_j^k$, $\mu_i=\frac{1}{L}\sum_{k=1}^{L}\tilde{h}_iy_i^k$, and $\eta_i=\frac{\sum_{k=1}^{L}|\tilde{h}_i y^i_i-y_j^k|^2}{\sigma_i^2}$ for an honest communication from agent $G_i$ to the host agent $G_j$ and $\eta^\prime_i=\frac{\sum_{k=1}^{L}|\tilde{h}_i (y^k_i-\Delta_i)-y_j^k|^2}{\sigma_i^2}$ for a Byzantine communication from agent $G_i$ to the host agent $G_j$ over the detection interval $L$. One can see that, $\eta^\prime_i=\eta_i+\frac{\sum_{k=1}^{L}(\tilde{h}^2_i \Delta^2_i+2\tilde{h}_i\Delta_iy_j^k-2\tilde{h}^2_i\Delta_iy_i^k)}{\sigma_i^2}$, hence the Byzantine agent with the selection of $\Delta_i>\frac{2(\mu_i-\mu_j)}{\tilde{h}_i}$ may commit an exploitative attack ($H_1 \rightarrow H_0$). Lastly, the Byzantine agent can adaptively estimate the relationship between its true output and its neighboring outputs based on the information it receives and accordingly set the value of $\Delta_i$. 

This modeling of Byzantine attacks is quite common in literature and covers a vast domain of adversary models \cite{zhang2015byzantine}. The above inequalities show the basic principle in terms of the amount of changes an attacker has to inject in order to fulfill "exploitation" and "vandalism" objectives, respectively. Lastly, as shown later, Byzantine agents will use large values for $\Delta_i$'s so that the magnitude of the local test statistics are dominated by the Byzantine agents' outputs and the degradation of the detection performance and the overall system's performance is maximized. This is, however, in odds with the Byzantine agents' other objective to conceal themselves. As a result, the Byzantine agents will have to choose their parameters wisely in order to fulfill both concealment and performance degradation objectives.
\section{Main Results} \label{sec:main}
\subsection{Synchronization Results}\label{sub:syn}
\begin{theorem}\label{thm:synch}
Consider the event-triggered multi-agent system described in Section \ref{sec:prob}, where each sub-system $G_j$ is output passive with the output passivity index $\rho_j$ and is controlled by the input mechanism given in (\ref{inp}). If the underlying connected communication graph is balanced, the communication time-delays and disturbances are negligible, and the communication attempts amongst all agents $G_j$ where $j=1,..., N$, are governed by the triggering conditions,
\begin{align*}
&||e_j(t)||_2^2 > \delta_j||y_j(t)||_2^2,
\end{align*}
where the design parameters $\delta_j$ are chosen such that,
\begin{align*}
0 < \delta_j \leq  \frac{\frac{2}{|\mathcal{N}_j^{in}|}(\lambda(G)+\rho_j)-\frac{1}{ \alpha}-\frac{1}{ \beta}}{\alpha+\beta},
\end{align*}
where $\alpha>0$ and $\beta>0$ are design variables and $\lambda(G)$ is the connectivity of the underlying communication graph, then the entire event-triggered multi-agent system achieves output synchronization asymptotically.
\end{theorem}
\begin{proof}
Each agent $G_j$ is output passive with the storage function (Lyapunov function) $V_j$ where,
\begin{align*}
\dot{V}_j(t) \leq u_j^T(t)y_j(t)-\rho_jy^T_j(t)y_j(t),~ \forall t>0,
\end{align*}
where the output passivity level is indicated by $\rho_j\in R$. $u_j, y_j \in R^m$ are the inputs and outputs of appropriate dimensions for the agent $G_j$. The error of the triggering condition for agent $j$ is defined as $e_j(t)=y_j(t)-y_j(t^n_i)$ for triggering instances $n=0,1,2,...$. Accordingly, for each agent, we have $e^T_j(t)e_j(t)\leq \delta_j y_j^T(t)y_j(t)$ between each two triggering instances.
Given the control input in (\ref{inp}), and the framework described in Section \ref{sec:prob}, the input to the agent $G_j$ is defined as,
\begin{align*} \label{input}
u_j=\sum_{k \in  \mathcal{N}_j^{in}}^{} a_k (y_k(t^n_k)-y_j(t^n_j)) 
   =\sum_{k \in  \mathcal{N}_j^{in}}^{} a_k [(y_k(t)-e_k(t))-(y_j(t)-e_j(t))],
\end{align*}
where $n=0,1,2,...$ are the triggering instances. The relationship for the storage function of agent $G_j$ becomes,
\begin{align*}
&\dot{V}_j \leq \sum_{k \in  \mathcal{N}_j^{in}}^{} a_k [(y_k(t)-e_k(t))-(y_j(t)-e_j(t))]^Ty_j(t)-\rho_jy^T_j(t)y_j(t)\\&~~= \sum_{k \in  \mathcal{N}_i^{in}}^{} a_k [(y_k(t)-y_j(t))-(e_k(t)-e_j(t))]^Ty_j(t)-\rho_jy^T_j(t)y_j(t).
\end{align*}
In order to show synchronization for all $N$ agents, we consider the following storage function for the entire multi-agent system,
\begin{align*}
&\dot{S}= \sum_{j=1}^{N} \dot{V}_j \leq \sum_{j=1}^{N} \sum_{k \in  \mathcal{N}_j^{in}}^{} a_k [(y_k(t)-y_j(t))-(e_k(t)-e_j(t))]^Ty_j(t)-\sum_{j=1}^{N}\rho_jy^T_j(t)y_j(t).
\end{align*}
As we explained in Section \ref{sec:gra} and Section \ref{sec:prob}, the flow of information amongst agents may be represented by the Laplacian of the underlying communication graph $L$. Moreover, if we define the matrix $E=[e_1^T,e_2^T,...,e_N^T]^T$, then we have,
\begin{align} 
&\dot{S}= \sum_{j=1}^{N} \dot{V}_j \leq -Y^TL^TY+Y^TL^TE-\sum_{j=1}^{N}\rho_jy^T_j(t)y_j(t)\\&
\leq -\lambda(G)Y^TY+Y^TL^TE-\sum_{j=1}^{N}\rho_jy^T_j(t)y_j(t),  \label{S}
\end{align}
where $\lambda(G)>0$ represents the algebraic connectivity of the underlying connected communication graph. Next, we may show the following,
\begin{align}
&Y^TL^TE=E^TLY=\sum_{j=1}^{N} \sum_{k \in  \mathcal{N}_j^{in}}^{} a_k (y_j(t)-y_k(t))^Te_j(t) \nonumber\\& 
~~~~~~~~~=\sum_{j=1}^{N}\sum_{k \in  \mathcal{N}_j^{in}}^{} a_k y_j^T(t)e_j(t)-\sum_{j=1}^{N}\sum_{k \in  \mathcal{N}_j^{in}}^{} a_k y_k^T(t)e_j(t).  \label{EY}
\end{align}
For all $j$ and $k$, we can have: $y_j^T(t)e_j(t) \leq \frac{\alpha e_j^T(t)e_j(t)}{2} + \frac{y_j^T(t)y_j(t)}{2 \alpha}$ and $y_j^T(t)e_j(t) \leq \frac{\beta e_j^T(t)e_j(t)}{2} + \frac{y_k^T(t)y_k(t)}{2 \beta}$ where $\alpha,\beta>0$. Utilizing these relationships in (\ref{EY}), we have,
\begin{align*}
&Y^TL^TE \leq \sum_{j=1}^{N}\sum_{k \in  \mathcal{N}_j^{in}}^{} a_k [\frac{\alpha e_j^T(t)e_j(t)}{2} + \frac{y_j^T(t)y_j(t)}{2 \alpha}]\\&~~~~~~~~~~+\sum_{j=1}^{N}\sum_{k \in  \mathcal{N}_j^{in}}^{} a_k [\frac{\beta e_j^T(t)e_j(t)}{2} + \frac{y_k^T(t)y_k(t)}{2 \beta}].
\end{align*}
This can be further simplified to have,
\begin{align*}
&Y^TL^TE \leq \sum_{j=1}^{N} |\mathcal{N}_j^{in}| [\frac{(\alpha+\beta) e_j^T(t)e_j(t)}{2} + \frac{y_j^T(t)y_j(t)}{2 \alpha}]\\&~~~~~~~~~~+\sum_{j=1}^{N}\sum_{k \in  \mathcal{N}_j^{in}}^{} a_k [\frac{y_k^T(t)y_k(t)}{2 \beta}].
\end{align*}
Further, we know that between any two triggering instances, one can show $e^T_j(t)e_j(t)\leq \delta_j y_j^T(t)y_j(t)$. This further gives us,
\begin{align*}
&Y^TL^TE \leq \sum_{j=1}^{N} |\mathcal{N}_j^{in}| [\frac{(\alpha+\beta) \delta_j}{2} + \frac{1}{2 \alpha}]y_j^T(t)y_j(t)\\&~~~~~~~~~~+\sum_{j=1}^{N}\sum_{k \in  \mathcal{N}_j^{in}}^{} a_k [\frac{y_k^T(t)y_k(t)}{2 \beta}].
\end{align*}
We have assumed that the underlying communication graph is balanced. This property implies that $\sum_{j=1}^{N}\sum_{k \in  \mathcal{N}_j^{in}}^{} a_k [\frac{y_k^T(t)y_k(t)}{2 \beta}]=\sum_{j=1}^{N}|\mathcal{N}_j^{in}| [\frac{y_j^T(t)y_j(t)}{2 \beta}]$. This leads to,
\begin{align}\label{error}
&Y^TL^TE \leq \sum_{j=1}^{N} |\mathcal{N}_j^{in}| [\frac{(\alpha+\beta) \delta_j}{2} + \frac{1}{2 \alpha}+\frac{1}{2 \beta}]y_j^T(t)y_j(t).
\end{align}
Utilizing (\ref{error}) in (\ref{S}), we have,
\begin{align*} \label{LS}
&\dot{S}= \sum_{j=1}^{N} \dot{V}_j \leq -\lambda(G)Y^TY-\sum_{j=1}^{N}\rho_jy^T_j(t)y_j(t)\\&~~+\sum_{j=1}^{N} |\mathcal{N}_j^{in}| [\frac{(\alpha+\beta) \delta_j}{2} + \frac{1}{2 \alpha}+\frac{1}{2 \beta}]y_j^T(t)y_j(t).  \numberthis
\end{align*}
We introduce the square diagonal matrix $\Theta \in R^{N\times N}$, where
\[ [\Theta]_{j,i} = \begin{cases}
+\lambda(G)+\rho_j-|\mathcal{N}_j^{in}| [\frac{(\alpha+\beta) \delta_j}{2} + \frac{1}{2 \alpha}+\frac{1}{2 \beta}]     & \quad \text{if } j=i \\
0  & \quad \text{otherwise.}
\end{cases}\]
Given (\ref{Ydelta}) and $\Theta$, (\ref{LS}) becomes,
\begin{align*}
&\dot{S} \leq -Y_\Delta^T \Theta Y_\Delta.
\end{align*}

If the  event-triggered multi-agent system is designed according to the theorem such that for each node $G_j$, we have: $\delta_j\leq \frac{\frac{2}{|\mathcal{N}_j^{in}|}(\lambda(G)+\rho_j)-\frac{1}{ \alpha}-\frac{1}{ \beta}}{\alpha+\beta}$ then matrix $\Theta$ is semi-positive. Moreover, for the storage function $S$ we have: $S \geq 0$ and $\dot{S}\leq 0$ for $\forall y \in R^m$ and $\forall t \geq0$. This implies $\dot{S} \rightarrow 0$ as $t\rightarrow \infty$ according to Barbalat's Lemma \cite{khalil2002nonlinear}. Consequently, $Y_\Delta$ converges to the limit set $ D=\{x|Y_\Delta=0, x \in R^{mN}\}$ for all states of all agents, 
\begin{align*} 
&0 \leq Y_\Delta^T \Theta Y_\Delta \leq -\dot{S} 
\end{align*}
This also means that the entire multi-agent system synchronizes asymptotically. 
\end{proof}
\begin{remark}
The triggering conditions show that agents that are more passive with higher output passivity indices can have larger triggering intervals and will be required to send their information to the network less frequently.  
\end{remark}
\begin{remark}
Graph connectivity has a relation with the communication rate amongst agents as well. The higher the connectivity of the underlying communication graph for the multi-agent system is, the larger the triggering intervals may be (less frequent communication attempts). 
\end{remark}
\begin{remark}
The result presented in Theorem \ref{thm:synch} also shows that agents with a higher number of neighbors will be required to send their information to the network more frequently compared to others. In other words, agents with a high number of neighbors play a more crucial part in the synchronization process of the entire multi-agent system. This is due to the fact that the triggering conditions show a reciprocal relationship between triggering intervals and number of neighbors. If an agent is responsible for sending its information to a higher number of neighbors (a higher number of neighboring agents rely on its information), then the agent will have to update its neighbors more frequently. 
\end{remark}
\begin{remark}
It is important to note that we did not consider the effects of external disturbances and time-delays in Theorem \ref{thm:synch}. It is assumed that these effects are negligible. However, if the delays are large enough, or external disturbances are strong enough, then they may affect the performance of the entire system.
\end{remark}
\begin{remark}
The results in Theorem \ref{thm:synch} are quite lenient. More specifically, they may hold for non-passive systems as well. For agent $G_j$, the triggering instance $\delta_j$ should be chosen such that  $+\lambda(G)+\rho_j-|\mathcal{N}_j^{in}| [\frac{(\alpha+\beta) \delta_j}{2} + \frac{1}{2 \alpha}+\frac{1}{2 \beta}] >0$. It is clear that for a non-passive system with a shortage of output passivity $\rho_j$, one can still design a multi-agent system that will synchronize as long as $\delta_j$ is chosen such that, $0 < \delta_j\leq \frac{\frac{2}{|\mathcal{N}_j^{in}|}(\lambda(G)+\rho_j)-\frac{1}{ \alpha}-\frac{1}{ \beta}}{\alpha+\beta}$.
\end{remark}
\subsection{Zeno-Behavior Analysis} \label{sub:zeno}
In practical settings, it may be necessary to guarantee a lower-bound on the time-intervals between triggering instances. The main motivation behind this problem is to avoid Zeno-behavior for the triggering conditions. Zeno-behavior happens when an infinite number of triggering conditions are met in a finite time-interval defeating the purpose of the event-triggered control framework. In order to avoid this behavior, we introduce a small positive constant $c$ to the triggering conditions to guarantee a positive lower-bound. We have shown before that the triggering conditions given in Theorem \ref{thm:synch} do gauarntee a positive lower-bound for inner-event time instances \cite{Rahnama2017}. Here, we show that our synchronization results does hold for the triggering condition $||e_j(t)||_2^2>\delta_j||y_j(t)||_2^2+c$ for $j=1,...,N$ as well. As a result, in practical applications one can use this triggering condition to secure a positive lower-bound, if necessary.
\begin{corollary} \label{cor:Zsynch}
Consider the event-triggered multi-agent system described in Section \ref{sec:prob}, where each sub-system $G_j$ is output passive with the output passivity index $\rho_j$ and is controlled by the input given in \ref{inp}. If the underlying connected communication graph is balanced, the communication time-delays and disturbances are negligible, and the communication attempts amongst all agents $G_j$ where $j=1,..., N$, are governed by the triggering conditions,
\begin{align*}
\label{eq:trigcon}
&||e_j(t)||_2^2 > \delta_j||y_j(t)||_2^2+c,
\end{align*}
where the design parameters $\delta_j$ are chosen such that,
\begin{align*}
0 < \delta_j \leq  \frac{\frac{2}{|\mathcal{N}_j^{in}|}(\lambda(G)+\rho_j)-\frac{1}{ \alpha}-\frac{1}{ \beta}}{\alpha+\beta},
\end{align*}
then the entire event-triggered network system achieves output synchronization asymptotically.
\end{corollary}
\begin{IEEEproof}
The proof follows the same line of reasoning as the proof given for Theorem \ref{thm:synch}. Following the same steps, one can show that,
\begin{align*}
&Y^TL^TE \leq \sum_{j=1}^{N} |\mathcal{N}_j^{in}| [\frac{(\alpha+\beta) \delta_j}{2} + \frac{1}{2 \alpha}]y_j^T(t)y_j(t)\\&~~~~~~~~~~+\sum_{j=1}^{N}\sum_{k \in  \mathcal{N}_j^{in}}^{} a_k [\frac{y_k^T(t)y_k(t)}{2 \beta}]+\sum_{j=1}^{N} |\mathcal{N}_j^{in}| [\frac{(\alpha+\beta) c}{2}].
\end{align*}
Further, one sees,
\begin{align*} 
&\dot{S}= \sum_{j=1}^{N} \dot{V}_j \leq -\lambda(G)Y^TY-\sum_{j=1}^{N}\rho_jy^T_j(t)y_j(t)\\&~~+\sum_{j=1}^{N} |\mathcal{N}_j^{in}| [\frac{(\alpha+\beta) \delta_j}{2} + \frac{1}{2 \alpha}+\frac{1}{2 \beta}]y_j^T(t)y_j(t)+\sum_{j=1}^{N} |\mathcal{N}_j^{in}| [\frac{(\alpha+\beta) c}{2}].
\end{align*}
By introducing the same matrix given in Theorem \ref{thm:synch}, $\Theta$, one has,
\begin{align*}
&\dot{S} \leq -Y_\Delta^T \Theta Y_\Delta +\sum_{j=1}^{N} |\mathcal{N}_j^{in}| [\frac{(\alpha+\beta) c}{2}].
\end{align*}
For small values of $c$, and if the event-triggered multi-agent system is design according to the corollary such that for each node $G_j$, we have: $\delta_j\leq \frac{\frac{2}{|\mathcal{N}_j^{in}|}(\lambda(G)+\rho_j)-\frac{1}{ \alpha}-\frac{1}{ \beta}}{\alpha+\beta}$ then matrix $\Theta$ is semi-positive. Moreover, for the storage function $S$ we have: $S \geq 0$ and $\dot{S}\leq 0$ for $\forall y \in R^m$ and $\forall t \geq0$. This implies $\dot{S} \rightarrow 0$ as $t\rightarrow \infty$ according to Barbalat's Lemma \cite{khalil2002nonlinear}. Consequently, $Y_\Delta$ converges to the limit set $ D=\{x|Y_\Delta=0, x \in R^{mN}\}$ for all states of all agents, which proves the corollary. 
\end{IEEEproof}
\begin{remark}
Corollary \ref{cor:Zsynch}, shows a trade-off between communication rate and performance. It is clear that for very large values of $c$ (very low communication rate), the synchronization state degrades quickly. In other words, synchronization is upper-bounded according to the relation, $Y_\Delta^T \Theta Y_\Delta < \sum_{j=1}^{N} |\mathcal{N}_j^{in}| [\frac{(\alpha+\beta) c}{2}]$. As a result, the designer should consider this trade-off before selecting the design parameter $c$. However, synchronization is possible based on the assumption that $c$ is chosen to be a very small positive number, and as a result the selection of $c$ is not consequential for the synchronization of the overall system.  
\end{remark}		
\subsection{Effects of Byzantine Agents on Synchronization}\label{sub:eff}
We assume that amongst $N$ agents, there are $N_B$ Byzantine nodes with the attack model described in Section \ref{sec:Byz} and $N_H$ honest nodes ($N_H+N_B=N$). $\mathcal{N}_H^{}$ and $\mathcal{N}_B^{}$ represent the set of honest and Byzantine agents, respectively. We represent the honest and Byzantine neighboring agents for $G_j$ by $\mathcal{N}_j^{in_H}$ and $\mathcal{N}_j^{in_B}$ ($\mathcal{N}_j^{in_H} \cap \mathcal{N}_j^{in_B}=\emptyset$, $\mathcal{N}_j^{in_H} \cup \mathcal{N}_j^{in_B}=\mathcal{N}_j^{in}$). $|\mathcal{N}_j^{in}|$ represents the same cardinality definition as given in Theorem \ref{thm:synch}. We define the cardinality of $\mathcal{N}_{j_B}^{in}$, $|\mathcal{N}_{j_B}^{in}|$ as only the number of neighbors for the Byzantine agents excluding their communication weights. $|\mathcal{N}_{j_B}^{in}|$ is zero for honest agents. The set of all Byzantine agents is represented by $\mathcal{N}^{B}$ and the set of all honest agents is represented by $\mathcal{N}^{H}$. For the honest agent $G_j^H$, the input under both hypotheses may be presented as,
\begin{align*} 
u_j^H=\sum_{k \in  \mathcal{N}_j^{in}}^{} a_k (y_k(t^n_k)-y_j(t^n_j)) 
= \sum_{k \in  \mathcal{N}_j^{in}}^{} a_k [(y_k(t)-y_j(t))-(e_k^\prime(t)-e_j(t))],
\end{align*}
where,
\[ e_k^\prime(t) = \begin{cases}
e_k(t)  \pm \Delta_k    & \quad \text{if } G_k \in \mathcal{N}_j^{in_B} \\
e_k(t)  & \quad \text{otherwise.}
\end{cases} \]
For the Byzantine agent $G_j^B$, the input may be presented as,
\begin{align*}
u_j^B=\sum_{k \in  \mathcal{N}_j^{in}}^{} a_k^B (y_k(t^n_k)-y_j(t^n_j)) 
= \sum_{k \in  \mathcal{N}_j^{in}}^{} a_k^B [(y_k(t)-y_j(t))-(e_k^\prime(t)-e_j(t))],
\end{align*}
where, $a_k^B=a_k + \omega_j$. The Lyapunov storage function for the entire multi-agent event-triggered network system becomes,
\begin{align*}
&\dot{S}= \sum_{j=1}^{N} \dot{V}_j \leq \sum_{j=1}^{N} \sum_{k \in  \mathcal{N}_j^{in}}^{} a_k^\prime [(y_k(t)-y_j(t))-(e_k^\prime(t)-e_j(t))]^Ty_j(t)-\sum_{j=1}^{N}\rho_jy^T_j(t)y_j(t).
\end{align*}
where,
\[ a_k^\prime = \begin{cases}
a_k + \omega_j    & \quad \text{if } G_j \in \mathcal{N}^{B} \\
a_k,& \quad \text{otherwise.}
\end{cases} \]
It is important to note that $\omega_j=0$ for honest agents. First, it can be shown that,
\begin{align*}
&\sum_{j=1}^{N} \sum_{k \in  \mathcal{N}_j^{in}}^{} a_k^\prime (y_k(t)-y_j(t))^Ty_j(t)
\\&=\sum_{j \in \mathcal{N}_B^{}}^{} \sum_{k \in  \mathcal{N}_j^{in}}^{} (a_k+\omega_j) (y_k(t)-y_j(t))^Ty_j(t)+\sum_{j \in  \mathcal{N}_H^{}}^{} \sum_{k \in  \mathcal{N}_j^{in}}^{} a_k (y_k(t)-y_j(t))^Ty_j(t)
\\&=\sum_{j \in \mathcal{N}_B^{}}^{} \sum_{k \in  \mathcal{N}_j^{in}}^{}\omega_j (y_k(t)-y_j(t))^Ty_j(t)+\sum_{j=1}^{N} \sum_{k \in  \mathcal{N}_j^{in}}^{} a_k (y_k(t)-y_j(t))^Ty_j(t)
\\&\leq\sum_{j \in \mathcal{N}_B^{}}^{} \sum_{k \in  \mathcal{N}_j^{in}}^{}\frac{\omega_j y^T_k(t)y_k(t)}{4} +\sum_{j=1}^{N} \sum_{k \in  \mathcal{N}_j^{in}}^{} a_k (y_k(t)-y_j(t))^Ty_j(t). 
\end{align*}
As a result, we have, 
\begin{align*}
&\dot{S}= \sum_{j=1}^{N} \dot{V}_j \leq -Y^TL^TY+Y^TL^{\prime T}E^\prime-\sum_{j=1}^{N}\rho_jy^T_j(t)y_j(t)\\&~~+\sum_{j \in \mathcal{N}_B^{}}^{} \sum_{k \in  \mathcal{N}_j^{in}}^{}\frac{\omega_j y^T_k(t)y_k(t)}{4}, 
\end{align*}
where,
\[E^\prime_{j,1} = \begin{cases}
e_j(t)  \pm \Delta_j   & \quad \text{if } G_j \in \mathcal{N}^{B} \\
e_j (t) & \quad \text{otherwise,}
\end{cases} \]
and $L^{\prime}$ is the Laplacian matrix of the new underlying communication graph consisting of $a^\prime_k$'s and is defined as,
\[[L^\prime]_{j,i} = \begin{cases}
\sum_{k \in  \mathcal{N}_j^{in}}^{} a_k^\prime  & \quad \text{if } j=i \\
-a_k^\prime  & \quad \text{if there is an arc from $G_i$ to $G_j$ with the gain $a_k^\prime$.}
\end{cases} \]
From Section \ref{sec:prob}, we remember,
\begin{align*}
\bar{Y}=\frac{1}{N}1_N^TY=\frac{1}{N}\sum_{i=1}^{N}y_i,
\end{align*}
and the measure for synchronization for the multi-agent system,
\begin{align*}
Y_\Delta=(y_1-\bar{Y},y_2-\bar{Y},...,y_N-\bar{Y})^T.
\end{align*}
We may follow the same steps as given in Theorem \ref{thm:synch}, and get to the following,
\begin{align*}
&\dot{S}= \sum_{j=1}^{N} \dot{V}_j \leq-Y^TL^TY+\sum_{j=1}^{N} (|\mathcal{N}_j^{in}| [\frac{1}{2 \alpha}+\frac{1}{2 \beta}] +|\mathcal{N}_{j_B}^{in}|\frac{\omega_j}{2 \alpha}) y_j^T(t)y_j(t)
\\&~~+\sum_{j=1}^{N} \sum_{k \in \mathcal{N}_j^{in}}^{} a^\prime_k [\frac{(\alpha+\beta)}{2}] e_j^{\prime^T}(t)e_j^\prime(t) -\sum_{j=1}^{N}\rho_jy^T_j(t)y_j(t)
\\&~~+\sum_{j \in \mathcal{N}_B^{}}^{} \sum_{k \in  \mathcal{N}_j^{in}}^{}(\frac{1}{4}+\frac{1}{2\beta})\omega_j y^T_k(t)y_k(t)
\\&~~ \leq -\lambda(G)Y^TY +\sum_{j=1}^{N} [(|\mathcal{N}_j^{in}| [\frac{1}{2 \alpha}+\frac{1}{2 \beta}] +|\mathcal{N}_{j_B}^{in}|\frac{\omega_j}{2 \alpha})-\rho_j]y^T_j(t)y_j(t)
\\&~~+\sum_{j=1}^{N} (|\mathcal{N}_j^{in}|+|\mathcal{N}_{j_B}^{in}|\omega_j) [\frac{(\alpha+\beta)}{2}] e_j^{\prime^T}(t)e_j^\prime(t)+\sum_{j \in \mathcal{N}_B^{}}^{} \sum_{k \in  \mathcal{N}_j^{in}}^{}(\frac{1}{4}+\frac{1}{2\beta})\omega_j y^T_k(t)y_k(t), \numberthis \label{relation3}
\end{align*}
where $\alpha$ and $\beta$ are the same parameters as given in Theorem \ref{thm:synch}. One can clearly quantify the negative effects, Byzantine nodes introduce to the entire framework by comparing (\ref{relation3}) and (\ref{LS}). The error term in (\ref{relation3}) can be expanded as well by utilizing $(e_j(t)\pm \Delta_j)^T(e_j(t)\pm \Delta_j)\leq 2 (e_j^{T}(t)e_j(t)+\Delta_j^2)$ , leading to the following,
\begin{align*} 
&  \dot{S}\leq -\lambda(G)Y^TY +\sum_{j=1}^{N} [(|\mathcal{N}_j^{in}| [\frac{1}{2 \alpha}+\frac{1}{2 \beta}] +|\mathcal{N}_{j_B}^{in}|\frac{\omega_j}{2 \alpha})-\rho_j]y^T_j(t)y_j(t)\\&~~+(\alpha+\beta) [\sum_{j=1}^{N} (|\mathcal{N}_j^{in}|+|\mathcal{N}_{j_B}^{in}|\omega_j)e_j^{T}(t)e_j(t) + \sum_{j \in \mathcal{N}_B^{}}^{} (|\mathcal{N}_j^{in}|+|\mathcal{N}_{j_B}^{in}|\omega_j) \Delta_j^2] \\&~~ 
+\sum_{j \in \mathcal{N}_B^{}}^{} \sum_{k \in  \mathcal{N}_j^{in}}^{}(\frac{1}{4}+\frac{1}{2\beta})\omega_j y^T_k(t)y_k(t).\numberthis \label{relation4}
\end{align*}
We introduce the same square diagonal matrix $\Theta \in R^{N\times N}$, where,
\[ [\Theta]_{j,i} = \begin{cases}
+\lambda(G)+\rho_j-|\mathcal{N}_j^{in}| [\frac{(\alpha+\beta) \delta_j}{2} + \frac{1}{2 \alpha}+\frac{1}{2 \beta}]     & \quad \text{if } j=i \\
0  & \quad \text{otherwise.}
\end{cases}\]
Given (\ref{Ydelta}), (\ref{relation4}) and $\Theta$, we have,
\begin{align*} 
&\dot{S}\leq - Y^T_\Delta \Theta Y_\Delta +\sum_{j \in \mathcal{N}_B^{}}^{} |\mathcal{N}_{j_B}^{in}|\frac{\omega_j}{2 \alpha}y^T_j(t)y_j(t)\\&~~+(\alpha+\beta) [\sum_{j=1}^{N} (\frac{|\mathcal{N}_j^{in}|}{2}+|\mathcal{N}_{j_B}^{in}|\omega_j)e_j^{T}(t)e_j(t) + \sum_{j \in \mathcal{N}_B^{}}^{} (|\mathcal{N}_j^{in}|+|\mathcal{N}_{j_B}^{in}|\omega_j) \Delta_j^2] \\&~~ 
+\sum_{j \in \mathcal{N}_B^{}}^{} \sum_{k \in  \mathcal{N}_j^{in}}^{}(\frac{1}{4}+\frac{1}{2\beta})\omega_j y^T_k(t)y_k(t). \numberthis \label{relation5}
\end{align*}
Given the assumption that the multi-agent system was initially designed according to Theorem \ref{thm:synch}, we have $\Theta>0$. After simplifying, and given $\dot{S} \rightarrow 0$ as $t\rightarrow \infty$, we have, 
\begin{align*} 
&0<-Y^T_\Delta \Theta Y_\Delta + \sum_{j \in \mathcal{N}_B^{}}^{} |\mathcal{N}_{j_B}^{in}|\frac{\omega_j}{2 \alpha}y^T_j(t)y_j(t)\\&+(\alpha+\beta) [\sum_{j=1}^{N} (\frac{|\mathcal{N}_j^{in}|}{2}+|\mathcal{N}_{j_B}^{in}|\omega_j)e_j^{T}(t)e_j(t) + \sum_{j \in \mathcal{N}_B^{}}^{} (|\mathcal{N}_j^{in}|+|\mathcal{N}_{j_B}^{in}|\omega_j) \Delta_j^2] \\&
+\sum_{j \in \mathcal{N}_B^{}}^{} \sum_{k \in  \mathcal{N}_j^{in}}^{}(\frac{1}{4}+\frac{1}{2\beta})\omega_j y^T_k(t)y_k(t). \numberthis \label{relation6}
\end{align*}
(\ref{relation6}) quantifies the effects of weight distortions and number of neighboring Byzantine agents on the convergence of the entire multi-agent system. More specifically, if one assumes the system is designed according to Theorem \ref{thm:synch} but initialized with the presence of Byzantine nodes then one can show an upper-bound for the effects of Byzantine agents on synchronization,
\begin{align*} 
&0< Y^T_\Delta \Theta Y_\Delta \leq \sum_{j \in \mathcal{N}_B^{}}^{} |\mathcal{N}_{j_B}^{in}|\frac{\omega_j}{2 \alpha}y^T_j(t)y_j(t)\\&+(\alpha+\beta) [\sum_{j=1}^{N} (\frac{|\mathcal{N}_j^{in}|}{2}+|\mathcal{N}_{j_B}^{in}|\omega_j)e_j^{T}(t)e_j(t) + \sum_{j \in \mathcal{N}_B^{}}^{} (|\mathcal{N}_j^{in}|+|\mathcal{N}_{j_B}^{in}|\omega_j) \Delta_j^2]  \\&
+\sum_{j \in \mathcal{N}_B^{}}^{} \sum_{k \in  \mathcal{N}_j^{in}}^{}(\frac{1}{4}+\frac{1}{2\beta})\omega_j y^T_k(t)y_k(t). \numberthis \label{relation7}
\end{align*}

It is obvious that in the presence of Byzantine agents, $Y_\Delta \neq 0$, as the honest agents will only be able to synchronize to a value that is based on the wrong data receiving from Byzantine agents ($y\pm\Delta$). In the best scenario, $y_H, y_B \rightarrow \bar{Y^\prime }=\frac{Y_H+Y_B}{N}$, where $Y_B$ represents the Byzantine outputs of all Byzantine agents sent to their neighbors and $Y_H$ represents the true outputs of honest agents. More specifically, if one assumes that the multi-agent system is designed according to Theorem \ref{thm:synch} but initialized with the presence of Byzantine agents then one can show a lower-bound and an upper-bound for the effects of Byzantine agents on synchronization and outputs of agents. The lower-bound happens when all agents synchronize to the wrong value $Y^\prime$. However, even the synchronization to this false value is not guaranteed anymore due to the positive upper-bound given in (\ref{relation7}). The positive upper-bound given in (\ref{relation7}) characterizes the worst possible outcome inflicted upon the multi-agent system by the Byzantine agents. For honest agents, (\ref{relation7}) demonstrates an upper-bound for the largest possible deviation caused by Byzantine agents between the honest agent's output and the correct synchronized value. Number of Byzantine nodes has a direct relation to this effect on synchronization. A larger number of Byzantine nodes can have a larger effect on deviating the multi-agent system from its true synchronized value. Additionally, there is the same direct effect between the Byzantine weights and synchronization, namely, the larger the Byzantine weights are, the larger deviations from the true synchronized value are. To sum up, there is a direct relationship between the upper-bound of deviations from the synchronized point (output over-shoot) and number of Byzantine neighbors and their associated Byzantine weights. The positive term involving $\Delta_j^2$ shows the relationship between the data falsification parameters and synchronization. The same direct relationship holds here as well. Falsifying the data by increasing the magnitudes of $\Delta_j$ directly weakens the synchronization of the entire multi-agent system. These results also show that an honest agent with a single Byzantine neighbor may never reach synchronization as the Byzantine neighbor by establishing the value of $\Delta_j$  and $\omega_j$ may consistently distract the honest node from reaching the synchronized value. Lastly, the positive term involving $\Delta_j^2$ guarantees a positive upper-bound for all agents, no matter if an agent has a Byzantine neighbor or not.
\begin{figure}[!t]
	\centering
	\includegraphics[scale = 0.8]{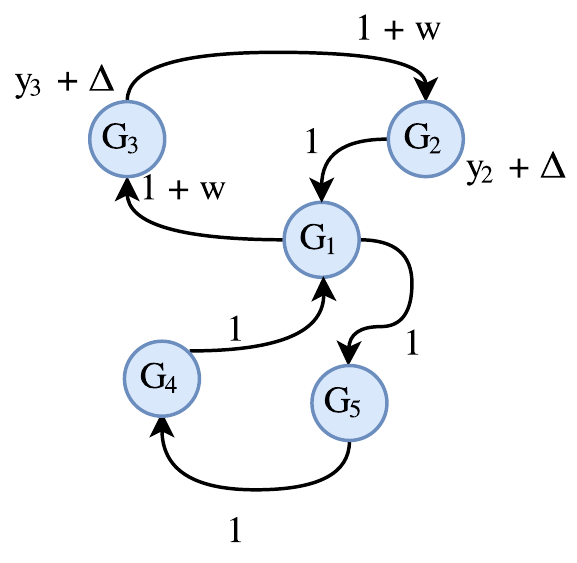}
	\caption{Graph - Byzantine Example.}
	\label{fig:GraphByz}
\end{figure}

As an example, let's look at the multi-agent system given in Fig. \ref{fig:GraphByz}. We assume that in the five agent system $G_1,G_4,G_5 \in \mathcal{N}_H$ and the rest are Byzantine agents. The honest agent $G_1$ has only one Byzantine neighbor $G_2$ with parameters $\Delta$ and $\omega$. We assume that the data falsifications for $G_2$ and  $G_3$ have a positive addition sign, $y_{B_2}(t)=y_2(t)+\Delta$ and $y_{B_3}(t)=y_3(t)+\Delta$. $y_B(t)$ is the output of Byzantine agent after data falsification. We have assumed that agent $G_2$ itself has another Byzantine neighbor $G_3$ ($|\mathcal{N}^{in}_{3_B}|=0$, $|\mathcal{N}_3^{in}|=1$ with parameters $\Delta$ and $\omega$). $G_4$ and $G_5$ are honest nodes. If $|\mathcal{N}^{in_B}_{1}|=1$, $|\mathcal{N}_1^{in}|=2$ for $G_1$, $\alpha=1$, $\beta=1$, $\lambda(G)=1$, $\rho_1=1.9$, $\delta_1=0.4$, $N=5$. Given that $|\mathcal{N}_{2_B}^{in}|=1$, $|\mathcal{N}_2^{in}|=1$ and with the assumption that $\rho_2=0.8$, $\delta_2=0.3$ for $G_2$ (Notice that $\delta_1$ and $\delta_2$ meet the synchronization conditions given in Theorem \ref{thm:synch}), for agents $G_1$ and $G_2$ at their respective triggering instances $t_k$ ($e_1^{^T}(t_k)e_1(t_k)=0$, $e_2^{^T}(t_k)e_2(t_k)=0$), we have,
\begin{align*} 
& 0 < 0.1 (y_1(t_k)-\bar{Y}(t_k))^2 \leq (2+\omega) \Delta^2+ 0.5\omega y_{2}^T(t_k)y_{2}(t_k)\\&~~+0.75\omega (y_{1}^T(t_k)y_{1}(t_k)+y_{3}^T(t_k)y_{3}(t_k)),\\&  \label{H1B2}
0 < 0.5 (y_2(t_k)-\bar{Y}(t_k))^2 \leq (2+\omega) \Delta^2 + 0.5\omega y_{2}^T(t_k)y_{2}(t_k)\\&~~+0.75\omega (y_{1}^T(t_k)y_{1}(t_k)+y_{3}^T(t_k)y_{3}(t_k)), 
\end{align*}
where $\bar{Y}(t_k)= \frac{y_1(t_k)+y_{B_2}(t_k)-\Delta+y_{B_3}(t_k)-\Delta+y_4(t_k)+y_5(t_k)}{5}$. This means that a single Byzantine agent can overtake an honest agent and control its behavior by determining the parameters $\omega$ and $\Delta$. Additionally, the Byzantine agents can affect $\bar{Y}$ through agent $G_1$.  This also means that if each honest agent in the network has only one Byzantine neighborhood and if all Byzantine neighborhoods work together by utilizing the same Byzantine parameters $\Delta$ and $\omega$ then the Byzantine agents can coerce the entire multi-agent system into following their desired behavior. However, for this attack to be meaningful, the values $\Delta$ and $\omega$ should be large enough to disrupt the overall performance of the multi-agent system. As we will see in the following sections, larger values of $\Delta$ and $\omega$ are easy to detect given our proposed detection framework. As a result a single Byzantine neighbor can be identified easily and its negative effects can be easily mitigated. In order for the Byzantine attack to be successful, each honest node should have more than only one Byzantine neighbor and $\mathcal{N}_j^{in_B}$ should be large enough. Our aim is to characterize the relationship between the number of Byzantine neighbors, detection performance, and synchronization. We will show the minimum number of Byzantine neighbors, an honest agent should have before the detection process is entirely blinded and the Byzantine agents become undetectable. We characterize the most efficient Byzantine attack. Lastly, we propose a more resilient algorithm for synchronization of the multi-agent systems.

\bf{Passivity and Effects of Byzantine Agents}:\normalfont ~Passivity to some extent can compensate and mitigate the negative effects of Byzantine agents. We later illustrate this through an example. One can determine from (\ref{relation7}), that for larger diagonal entries of $\Theta$, the effects of agents' output overshoot, or the size of the largest possible deviation from $\bar{Y}$ may be mitigated and decreased. From the definition of $\Theta$, one can see that the excess of passivity may lead to a larger entry for agent $G_j$ in $\Theta$ and a better worst case scenario in terms of deviations from the desired value $\bar{Y}$. For all agent $G_j$ where $j=1,...,N$, we can represent $\rho_j=\rho_j^\prime+\rho^\Delta_j>0$, where $\epsilon^\prime-\rho_j^\prime=\lambda(G)-|\mathcal{N}_j^{in}| [\frac{(\alpha+\beta) \delta_j}{2} + \frac{1}{2 \alpha}+\frac{1}{2 \beta}]$, with $\epsilon^\prime>0$. Then (\ref{relation7}) becomes,
\begin{align*} 
&0< Y_\Delta^T \Theta^\star Y_\Delta \leq \sum_{j \in \mathcal{N}_B^{}}^{} |\mathcal{N}_{j_B}^{in}|\frac{\omega_j}{2 \alpha}y^T_j(t)y_j(t) \\&+(\alpha+\beta) [\sum_{j=1}^{N} (\frac{|\mathcal{N}_j^{in}|}{2}+|\mathcal{N}_{j_B}^{in}|\omega_j)e_j^{T}(t)e_j(t) + \sum_{j \in \mathcal{N}_B^{}}^{} (|\mathcal{N}_j^{in}|+|\mathcal{N}_{j_B}^{in}|\omega_j) \Delta_j^2] \nonumber \\&
+\sum_{j \in \mathcal{N}_B^{}}^{} \sum_{k \in  \mathcal{N}_j^{in_B}}^{}(\frac{1}{4}+\frac{1}{2\beta})\omega_j y^T_k(t)y_k(t)-Y^T \Theta_{\rho_\Delta} Y.
\end{align*}
where 
\[ [ \Theta_{\rho_\Delta}]_{j,i} = \begin{cases}
\rho^\Delta_j   & \quad \text{if } j=i \\
0  & \quad \text{otherwise.}
\end{cases} \]
\[ [ \Theta^\star]_{j,i} = \begin{cases}
\epsilon^\prime  & \quad \text{if } j=i \\
0  & \quad \text{otherwise.}
\end{cases} \]
This means that an excess of passivity in agents can compensate for the negative effects of Byzantine weight and data manipulation. A design-based interpretation of this result tells us that one can design the triggering conditions (see \ref{eq:trig}) such that $\rho_j^\Delta>0$ for all agents, in order to increase the resilience of the entire event-triggered multi-agent system. However, this results in a decrease in the size of triggering intervals and a higher communication rate amongst agents. But the conclusion is that more passive multi-agent systems are also more resilient toward this type of Byzantine attack.
\subsection{Simulation Example} 
\begin{example}\label{exm:ex1}
\begin{figure}[!t]
	\centering
	\includegraphics[scale = 0.5]{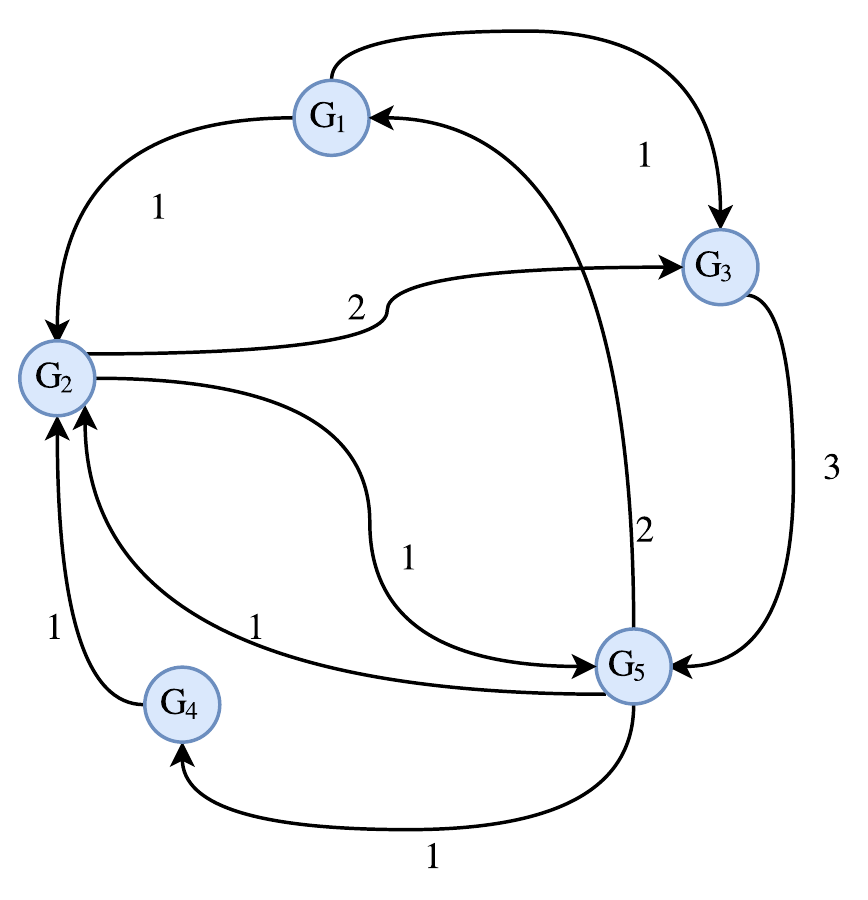}
	\caption{Graph - Example \ref{exm:ex1}.}
	\label{fig:Graph2}
\end{figure}
We consider a multi-agent event-triggered network system consisting of five agents ($i=1,...,5$) with the underlying balanced communication topology given in Fig. \ref{fig:Graph2}. We assume the following simple dynamics for sub-systems,  
\[ G_{i} = \begin{cases}
\dot{x}_i(t)= - c_i x_i(t) + u_i(t)   \\
y_i(t)=x_i(t), & \\
\end{cases} \]
with $c_1=1.2,~c_2=2.2,~c_3=2.4,~c_4=0.6,~c_5=4$. One can verify that all agents are dissipative with the storage function $V_i(x)=\frac{1}{2}x_i^T(t)x_i(t)$. This results in output passivity indices, $\rho_1=1.2,~\rho_2=2.2,~\rho_3=2.4,~\rho_4=0.6,~\rho_5=4$. The Laplacian matrix of the underlying communication graph amongst agents is,
\begin{align*}
L=\begin{bmatrix} 
2    &  -1  &  -1 & 0 &  0 \\
0    &  3   &  -2 & 0  &  -1\\
0    &  0   &  3  & 0  & -3 \\
0    &  -1  &  0  & 1  &  0 \\
-2   &  -1  &  0  & -1 &  4 
\end{bmatrix}.
\end{align*}
with the connectivity measure: $\lambda(G)=1.234$. Based on Theorem \ref{thm:synch}, one can design the following triggering conditions,
\begin{align*}
&||e_1(t)||_2^2>0.21||y_1(t)||_2^2,\\&
||e_2(t)||_2^2>0.14||y_2(t)||_2^2,\\&
||e_3(t)||_2^2>0.20||y_3(t)||_2^2,\\&
||e_4(t)||_2^2>0.60||y_4(t)||_2^2,\\&
||e_5(t)||_2^2>0.29||y_5(t)||_2^2,
\end{align*}
by selecting $\alpha_i=1,~\beta_i=1$ for $i=1,...,5$. For initial conditions, $y_1(0)=5,~y_2(0)=10,~y_3(0)=-5,~y_4(0)=1,~y_5(0)=-3$, Fig.\ref{fig:outputs} shows that the system synchronizes. Fig. \ref{fig:inner} shows the evolution of triggering condition for each agent and shows that Zeno-behavior does not happen and the event-triggered premise is met. 
\begin{figure}[!t]
	\centering
	\includegraphics[scale = 0.45]{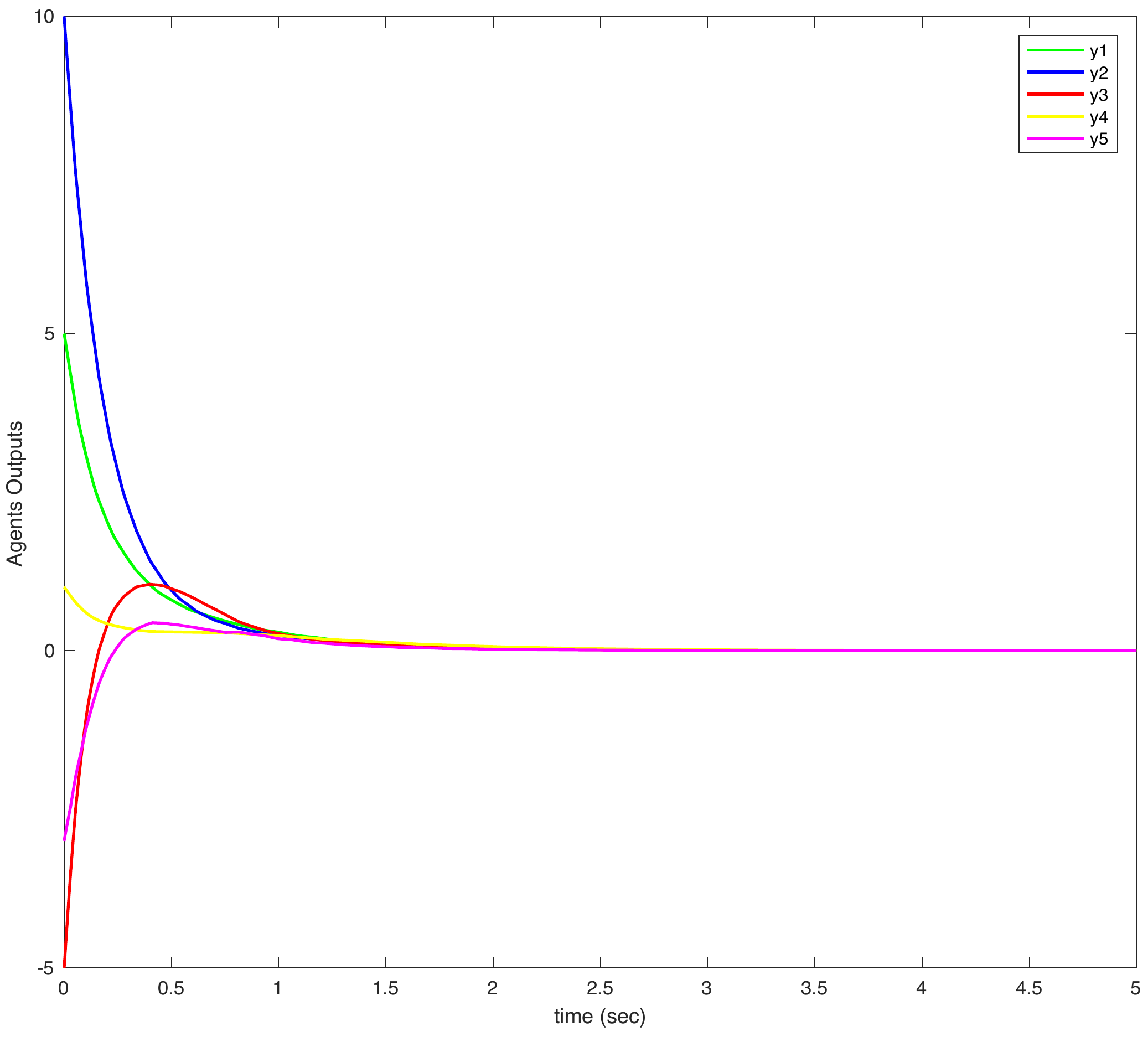}
	\caption{The Outputs of the Multi-Agent System.}
	\label{fig:outputs}
\end{figure}
\begin{figure}[!t]
	\centering
	\includegraphics[scale = 0.45]{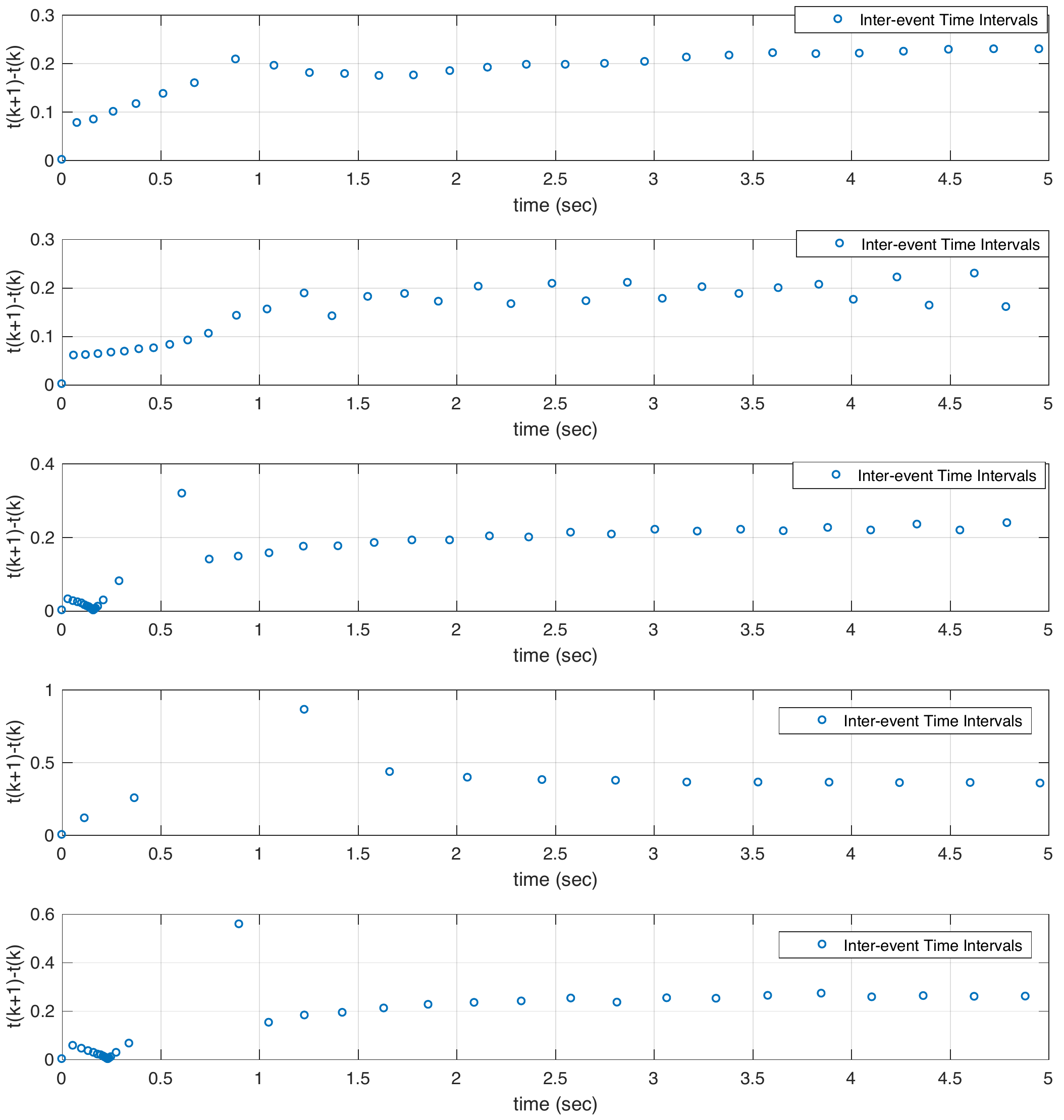}
	\caption{The Inner-Event Time Intervals of the Multi-Agent System.}
	\label{fig:inner}
\end{figure}

Now, we show the effects of a Byzantine attack on the multi-agent system to affirm our results in the previous sub-section. We consider the case where agent $G_1$ is compromised. First, we assume that $G_1$ has not manipulated its weight but that it only manipulates the information it is sending to other agents. $G_1$ sends $\tilde{Y}_1=Y_1+ \Delta$ where $\Delta=10$ instead of sending its true value. Fig \ref{fig:outputsByz1} shows the effects of the Byzantine agent on synchronization. As expected, the convergence deviates from the correct synchronized value by a positive magnitude which depends on $\Delta$, additionally, the error propagates through the network and affects other honest agents. Lastly, we consider the effects of weight manipulation and data falsification together. We assume that agent $G_5$ has changed its input weight $a_1=2$ to $a^\prime_1=8$ and also sends the same false data $\tilde{Y}$ to its neighbors. As shown in Fig. \ref{fig:outputsByz2}, a single Byzantine agent is able to deceive the entire multi-agent system into following its desired behavior by manipulating its weight and falsifying its data. Moreover, as expected, comparing Fig \ref{fig:outputsByz1} and Fig. \ref{fig:outputsByz2}, we see that by combining weight manipulation and data falsification, the upper-bound for all outputs of all agents has increased. 
\begin{figure}[!t]
	\centering
	\includegraphics[scale = 0.5]{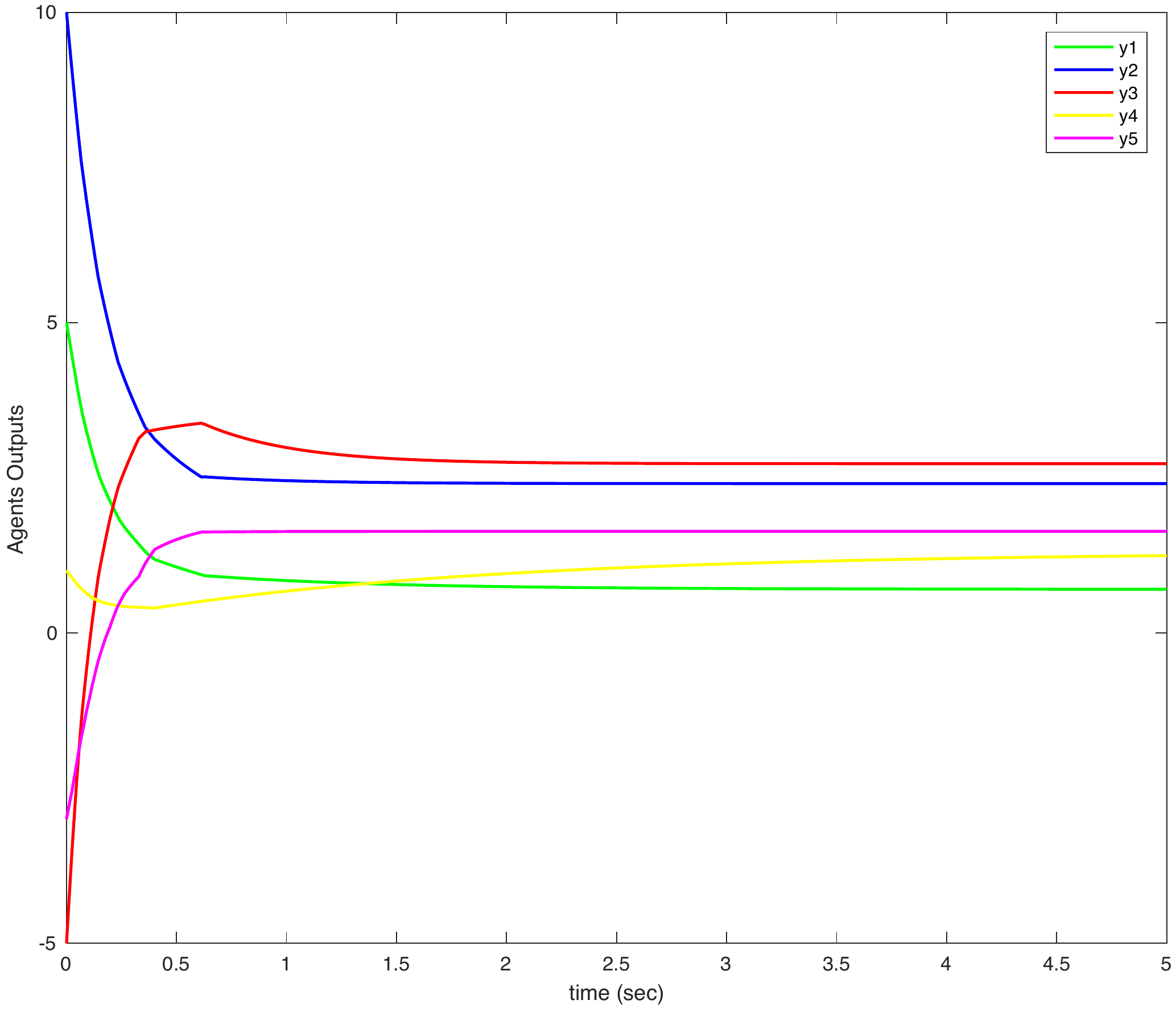}
	\caption{The Multi-Agent System's behavior under Data Falsification.}
	\label{fig:outputsByz1}
\end{figure}
\begin{figure}[!t]
	\centering
	\includegraphics[scale = 0.5]{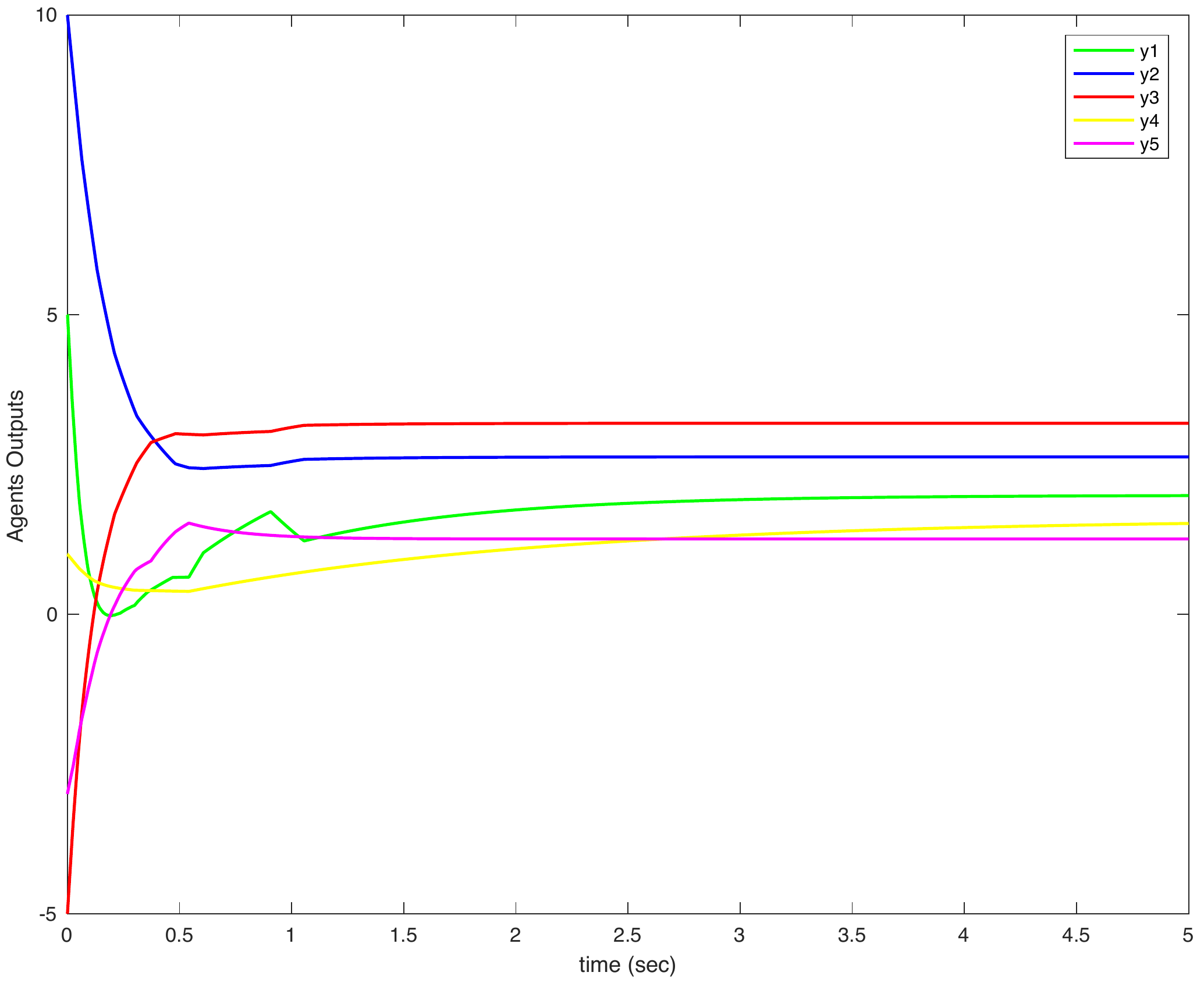}
	\caption{The Multi-Agent System's behavior under both Weight Manipulation and Data Falsification.}
	\label{fig:outputsByz2}
\end{figure}
\end{example}
It is important to note that in our analysis of attack parameters, we did not consider the attack probability ($P$). This was to characterize the worst possible effect the Byzantine attack may have on the entire system. In the following sections we will expand the above work to consider the probability of attack. Next, we will analyze the performance of the detection unit.
\section{An Analysis on the Performance of the Detection Framework}\label{sec:analdet}
\subsection{Transient Performance Analysis of the Detection Algorithm}\label{sub:transdet}
In this subsection, we analyze the transient performance of the detection framework. This analysis is based on characterization of the probability of correctly detecting that the multi-agent system is under attack and detecting the true signal (probability of detection) and the probability of false alarm; incorrectly determining that the multi-agent system has not reached synchronization when indeed it has (failing to detect the attack). As mentioned in Section \ref{sec:Byz}, Byzantine nodes attempt to replace the hypothesis $H_0$ with $H_1$ and vice versa. The probability of detection characterizes the ability of the detection unit to discover the Byzantine agents' attempt to "exploit" ($H_1 \rightarrow H_0$). The false alarm probability characterizes the probability that the detection unit does not detect the Byzantine agents' objective to "vandalize" ($H_0 \rightarrow H_1$) and falsely decides the non-existence of an attack. This means that even-though the system has reached synchronization, the detection unit follows the falsified information received from the Byzantine agents and decides against synchronization. The transient analysis of the detection unit is based on the hypothesis testing presented in (\ref{epsilon}). We remember that each agent $G_j$ independently calculates its local statistics with each neighbor $G_k$, according to  $T_k^j=\sum_{i=1}^{L}|y_k^i-y_j^i|^2$ over the time-interval of length $L$. The overall local statistics with all neighbors of $G_j$, which is utilized in the decision making process given in (\ref{epsilon}) at time-instance $t$, becomes $\wedge^t_j= (\sum_{k \in \mathcal{N}_j^{in}}^{}T_k^j)^t$. We also remember that according to the framework presented in Section \ref{sec:Byz}, the Byzantine agents are intelligent, know the true hypothesis, and will attempt to exploit and vandalize the synchronization process by confusing the detection framework and diverging the synchronization process by falsifying their data. In this subsection, we characterize the transient degradation of the detection performance in the presence of Byzantine neighbors. The local test statistic for the detection unit of an honest agents in the presence of honest and Byzantine neighbors at time-instance $t$ is $(\sum_{k \in \mathcal{N}_j^{in}}^{} T_k^j)^t=(\sum_{k \in\mathcal{N}_j^{in_H}}^{} T_k^j)^t + (\sum_{k \in \mathcal{N}_j^{in_B}}^{} \tilde{T}_k^j)^t$. For sufficiently large number of local test statistics of length $L$, the distribution of the test statistics with a Byzantine neighbor $G_k$, $\tilde{T}_k^j$, given the hypothesis $H_i$ $(i=0,1)$ is a Gaussian mixture of $\mathcal{N}((\mu_{i0})_k,(\sigma_{i0}^2)_k)$ with probability $(1-P_k)$ and $\mathcal{N}((\mu_{i1})_k,(\sigma_{i1}^2)_k)$ with probability $P_k$. The distribution of the test statistics from an honest neighbor $G_k$, $T_k^j$, given the hypothesis $H_i$ $(i=0,1)$ is a Gaussian distribution  $\mathcal{N}((\mu_{i0})_k,(\sigma_{i0}^2)_k)$, where,
\begin{align} 
&(\mu_{00})_k=L\sigma_k^2, ~(\mu_{01})_k=L\sigma_k^2+L\tilde{h}^2_k\Delta^2_k\sigma_k^2 \label{D1} \\& 
(\mu_{10})_k=(L+\eta_k)\sigma_k^2, ~(\mu_{11})_k=(L+\eta^\prime_k)\sigma_k^2 \label{D2} \\&
(\sigma_{00}^2)_k=2L\sigma_k^4,~(\sigma_{01}^2)_k=2(L+2L\tilde{h}^2_k\Delta^2_k)\sigma_k^4 \label{D3} \\&(\sigma_{10}^2)_k=2(L+2\eta_k)\sigma_k^4,~(\sigma_{11}^2)_k=2(L+2\eta^\prime_k)\sigma_k^4 \label{D4}
\end{align}
As a result the probability distribution of $\tilde{T}_k^j$ becomes,
\begin{align}
f_{PDF}(\tilde{T}_k|H_i)=(1-P_k)\phi((\mu_{i0})_k,(\sigma_{i0}^2)_k)+P_k\phi((\mu_{i1})_k,(\sigma_{i1}^2)_k), ~i=0,1.
\end{align}
$\phi(\mu,\sigma_{i0}^2)$ is the probability distribution function of $X \simeq \mathcal{N}(\mu,\sigma^2) $. In order to attain a closed form for the transient probability distribution of the detection center, first we start with a simple example and then expand the results to the general case. For the sake of convenience, we assume that $P_k=P$ for all $k \in \mathcal{N}_j^{in_B}$. If we assume that agent $G_j$ has 2 Byzantine neighbors $(G_1,G_2)$ and 2 Honest neighbors $(G_3,G_4)$ , at time-instance $t$, we have $\wedge_j^t=(\tilde{T}_1^j)^t+(\tilde{T}_2^j)^t+(T_3^j)^t+(T_4^j)^t$. $\wedge_j^t$ is the result of the summation of independent random variables. Consequently, the distribution of $\wedge_j^t$ is the result of the convolution of the distribution of these independent random variables,
\begin{align*}
&f_{PDF}(\wedge_j^t|H_k)=f_{PDF}((\tilde{T}_1^j)^t|H_k)\ast f_{PDF}((\tilde{T}_2^j)^t|H_k)\ast f_{PDF}((T_3^j)^t|H_k)\ast f_{PDF}((T_4^j)^t|H_k), ~k=0,1.
\end{align*}
Further we have,
\begin{align*}
&f_{PDF}(\wedge_j^t|H_k)=[(1-P_1)\phi((\mu_{k0})_1,(\sigma_{k0}^2)_1)+P_1\phi((\mu_{k1})_1,(\sigma_{k1}^2)_1)]\\&\ast[(1-P_2)\phi((\mu_{k0})_2,(\sigma_{k0}^2)_2)+P_2\phi((\mu_{k1})_2,(\sigma_{k1}^2)_2)]\\&\ast \phi(\mu_3,\sigma_3^2)\ast \phi(\mu_4,\sigma_4^2),\\&
=(1-P_1)(1-P_2)\phi((\mu_{k0})_1,(\sigma_{k0}^2)_1)\ast\phi((\mu_{k0})_2,(\sigma_{k0}^2)_2)\ast\phi(\mu_3,\sigma_3^2)\ast\phi(\mu_4,\sigma_4^2)\\&
+(1-P_1)P_2\phi((\mu_{k0})_1,(\sigma_{k0}^2)_1)\ast\phi((\mu_{k1})_2,(\sigma_{k1}^2)_2)\ast\phi(\mu_3,\sigma_3^2)\ast\phi(\mu_4,\sigma_4^2)\\&
+P_1(1-P_2)\phi((\mu_{k1})_1,(\sigma_{k1}^2)_1)\ast\phi((\mu_{k0})_2,(\sigma_{k0}^2)_2)\ast\phi(\mu_3,\sigma_3^2)\ast\phi(\mu_4,\sigma_4^2)\\&
+P_1P_2\phi((\mu_{k1})_1,(\sigma_{k1}^2)_1)\ast\phi((\mu_{k1})_2,(\sigma_{k1}^2)_2)\ast\phi(\mu_3,\sigma_3^2)\ast\phi(\mu_4,\sigma_4^2),  ~k=0,1.
\end{align*}
Given the fact that the convolution of two normal distributions is also a normal distribution with a mean and variance resulting from the summation of the means and variances of the initial normal distributions, and that $P_k=P$ for all $k \in \mathcal{N}_j^{in_B}$, we have,
\begin{align*}
&f_{PDF}(\wedge_j^t|H_k)=(1-P)(1-P)\phi((\mu_{k0})_1+(\mu_{k0})_2+\mu_3+\mu_4,(\sigma_{k0}^2)_1+(\sigma_{k0}^2)_2+\sigma_3^2+\sigma_4^2)\\&
+(1-P)P\phi((\mu_{k0})_1+(\mu_{k1})_2+\mu_3+\mu_4,(\sigma_{k0}^2)_1+(\sigma_{k1}^2)_2+\sigma_3^2+\sigma_4^2)\\&
+P(1-P)\phi((\mu_{k1})_1+(\mu_{k0})_2+\mu_3+\mu_4,(\sigma_{k1}^2)_1+(\sigma_{k0}^2)_2+\sigma_3^2+\sigma_4^2)\\&
+P^2\phi((\mu_{k1})_1+(\mu_{k1})_2+\mu_3+\mu_4,(\sigma_{k1}^2)_1+(\sigma_{k1}^2)_2+\sigma_3^2+\sigma_4^2),  ~k=0,1.
\end{align*}
Byzantine agents behave probabilistically in the sense that their states change from Byzantine to honest and vice versa with a probability that depends on $P$. We define the set $\mathcal{Z}_B$ as the combination of Byzantine states for Byzantine agents such that for this example, we have, $\mathcal{Z}_B=\{\{H_1,H_2\},\{B_1,H_2\},\{H_1,B_2\},\{B_1,B_2\}\}$, where the presence of $B_i$ or $H_i$ in the combinations of states indicates that the Byzantine agent $G_i$ is behaving as a Byzantine or honest agent, respectively. We denote $\mathcal{Z}^{ind}_B$ as indices of Byzantine agents in $\mathcal{Z}_B$ states, $\mathcal{Z}^{ind}_B=\{Z_1=\{\},Z_2=\{1\},Z_3=\{2\},Z_4=\{1,2\}\}$. As a result we have the complement set, $C(\mathcal{Z}^{ind}_B)=\{Z^c_1=\{1,2\},Z^c_2=\{2\},Z^c_3=\{1\},Z^c_4=\{\}\}$. Needless to say, we have the following cardinality relationship, $|\mathcal{Z}_i|+|\mathcal{N}_H|=N$.
\begin{lemma} \label{lem:prob}
	The probability distribution function of the local test statistic for agent $G_j$ with $N_B$ Byzantine neighbors and $N_H$ honest agents with the detection time-interval $L$, at time-instance $t$, given the hypothesis $H_i$ $(i=0,1)$, $\wedge_j^t$ is a Gaussian mixture determined by, 
	\begin{align*}
	&f_{PDF}(\wedge_j^t|H_k)=\sum_{Z_i \in \mathcal{Z}^{ind}_B}^{} P^{|\mathcal{Z}_i|}(1-P)^{N_B-|\mathcal{Z}_i|}\phi(\mu,\sigma^2),~\text{where},\\& \mu=\sum_{i \in \mathcal{Z}^c_i}^{}(\mu_{k0})_i+\sum_{i \in \mathcal{Z}_i }^{}(\mu_{k1})_i+\sum_{i \in \mathcal{N}_H}^{}(\mu_{k0})_i,\\& \sigma^2=\sum_{i \in \mathcal{Z}^c_i}^{}(\sigma^2_{k0})_i+\sum_{i \in \mathcal{Z}_i}^{}(\sigma^2_{k1})_i+\sum_{i \in \mathcal{N}_H}^{}(\sigma_{k0}^2)_i,~ \text{and} ~k=0,1.
	\end{align*}
\end{lemma}
Consequently, the transient performance of the detection unit of agent $G_j$ may be characterized by the probability of detection and false alarm as follows,
\begin{proposition}
	The probability of detection and false alarm of the detection unit for agent $G_j$ at time-instance $t$ may be characterized as,
	\begin{align*}
	&(P_D^j)^t=Pr(D=H_1|H_1)\\&=\sum_{Z_i \in \mathcal{Z}^{ind}_B}^{}P^{|\mathcal{Z}_i|}(1-P)^{N_B-|\mathcal{Z}_i|}Q(\frac{\gamma_j-\sum_{i \in \mathcal{Z}^c_i}^{}(\mu_{10})_i-\sum_{i \in \mathcal{Z}_i}^{}(\mu_{11})_i-\sum_{i \in \mathcal{N}_H}^{}(\mu_{10})_i}{\sqrt{\sum_{i \in \mathcal{N}}^{}(\sigma^2_{10})_i}}),\\&
	(P_{FA}^j)^t=Pr(D=H_1|H_0)\\&=\sum_{Z_i \in \mathcal{Z}^{ind}_B}^{}P^{|\mathcal{Z}_i|}(1-P)^{N_B-|\mathcal{Z}_i|}Q(\frac{\gamma_j-\sum_{i \in \mathcal{Z}^c_i}^{}(\mu_{00})_i-\sum_{i \in \mathcal{Z}_i}^{}(\mu_{01})_i-\sum_{i \in \mathcal{N}_H}^{}(\mu_{00})_i}{\sqrt{\sum_{i \in \mathcal{N}}^{}(\sigma^2_{00})_i}}).
	\end{align*}
\end{proposition}
\begin{proof}
	$(P_D^j)^t$ can be easily derived from calculating $Pr(D=H_1|H_1)$ from a combination of Gaussian distributions with the following means and variances, $(\mu_{10})_k=(L+\eta_k)\sigma_k^2, ~(\mu_{11})_k=(L+\eta^\prime_k)\sigma_k^2$ and $(\sigma_{10}^2)_k=2(L+2\eta_k)\sigma_k^4,~(\sigma_{11}^2)_k=2(L+2\eta^\prime_k)\sigma_k^4$ for all neighbors $k \in \mathcal{N}_{in}$ (honest and Byzantine agents $G_j$) according to Lemma \ref{lem:prob}. It is important to note that for a Gaussian distribution $Y$ with the mean $\mu$ and variance $\sigma^2$, $X=\frac{Y-\mu}{\sigma^2}$ is a standard normal distribution and $P(Y>y)=P(X>x)=Q(\frac{Y-\mu}{\sigma^2})=Q(x)$. $(P_{FA}^j)^t$ or $Pr(D=H_1|H_0)$ may be calculated in a similar manner given the means and variances, $(\mu_{00})_k=L\sigma_k^2, ~(\mu_{01})_k=L\sigma_k^2+L\tilde{h}^2_k\Delta^2_k\sigma_k^2$ and $(\sigma_{00}^2)_k=2L\sigma_k^4,~(\sigma_{01}^2)_k=2(L+2L\tilde{h}^2_k\Delta^2_k)\sigma_k^4$ for all $k \in \mathcal{N}_{in}$ (honest and Byzantine agents $G_j$). It is important to note that one may first establish a desired rate of false alarm by deciding $\gamma_j$ and then determine the detection performance. 
\end{proof}
It is important to note that the probability of detection indicates the probability that agent $G_j$ \it{detects} \normalfont the Byzantine attack which is trying to distort performance by persuading the detection unit that the entire network has reached synchronization or by forcing the honest agent to follow the falsified data, and consequently \it{decides} \normalfont against it based on the distribution of the true signal. Additionally, under (\ref{epsilon}), the probability of false alarm indicates the probability that Byzantine neighbors will succeed in coercing agent $G_j$ into mistakenly deciding that the entire multi-agent network has not reached synchronization when indeed it has, thereby fulfilling its adversarial objective to move the system from $H_0 \rightarrow H_1$. 
\begin{example}
	Consider agent $G_2$ in the event-triggered multi-agent system given in Example \ref{exm:ex1} with three neighbors. We consider the same underlying communication graph and dynamics for the entire event-triggered multi-agent system. We assume that $G_5$ is a Byzantine neighbor and $G_1$ and $G_4$ are honest neighbors for agent $G_2$. We will analyze the transient detection and false alarm probability distributions for the local test statistic $\wedge_2^t$ for agent $G_2$'s detection center and quantify the harmful effects of the Byzantine neighbor $G_5$ on the detection performance. We consider the same dynamics for the agents and the initial conditions, $y_1(0)=3.5,~y_2(0)=4,~y_3(0)=0.5,~y_4(0)=3,~y_5(0)=2$. The Byzantine agent manipulates its weight to $a_5^\prime$ where, $a_5^\prime=a_5+1$. We consider the channel gains $\tilde{h}_1=0.92,~\tilde{h}_4=0.95,~\tilde{h}_5=0.96$ and assume $\sigma^2_i=1$ for all the communication links, $i=1,..,5$. We plot the transient performance of the detection unit for a set of attack strengths $\Delta_1=0.8$, $\Delta_1=0.9$, $\Delta_1=1$, $\Delta_1=1.2$ and $\Delta_1=1.6$. In all cases, we assume the probability of attack $P_5=0.5$. Lastly, the detection interval for the detection unit is $L=15$ and $\lambda=15$ where, $\gamma_j=\sum_{k \in  \mathcal{N}_j^{in}}^{}L\sigma_k^2+\lambda$ is chosen based on the desired false alarm rate (see (\ref{epsilon})). Fig. \ref{fig:detection} depicts the transient detection performance for different attack strengths. As seen, the Byzantine neighbor can considerably harm the detection performance. Similarly, the false alarm rate for the same set-up for when the entire multi-agent has synchronized (under $H_0$) is shown in Fig. \ref{fig:false}. It is clear that the Byzantine neighbor can considerably increase the false alarm rate by appropriately selecting the attack parameters. This means that agent $G_2$ will mistakenly continue its communication with its neighbors based on the false belief that the multi-agent system has not reached synchronized. Similar to the steady-state analysis of the detection framework, this example shows that the Byzantine agents can considerably degrade the transient performance of the detection unit.
	\begin{figure}[!t]
		\centering
		\includegraphics[scale = 0.5]{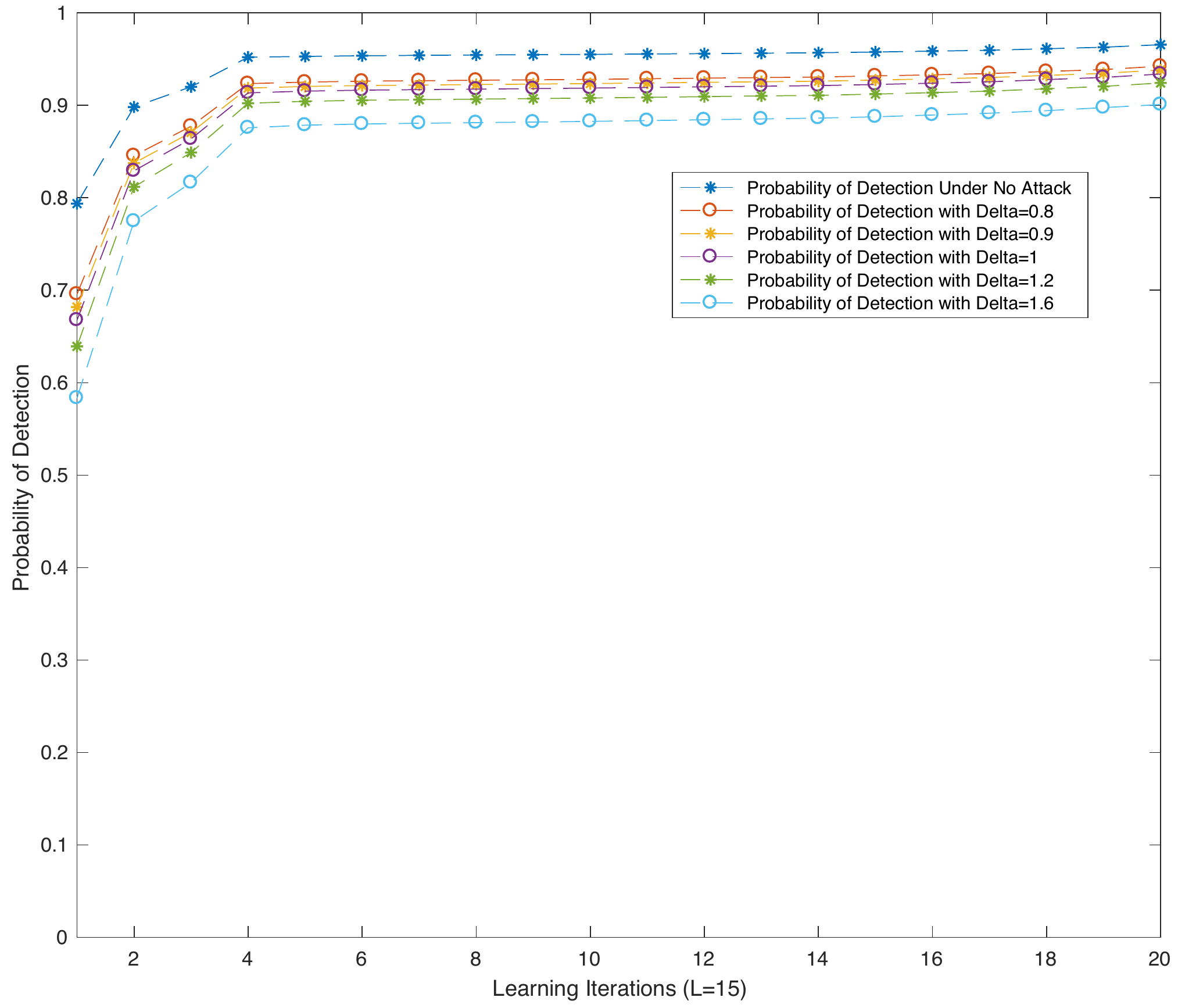}
		\caption{Probability of Detection  (Detection Interval $L=15$, Attack Parameters: $P_5=0.5$, $\Delta_5$ and $a_5^\prime=a_5+1$).}
		\label{fig:detection}
	\end{figure}
	\begin{figure}[!t]
		\centering
		\includegraphics[scale = 0.5]{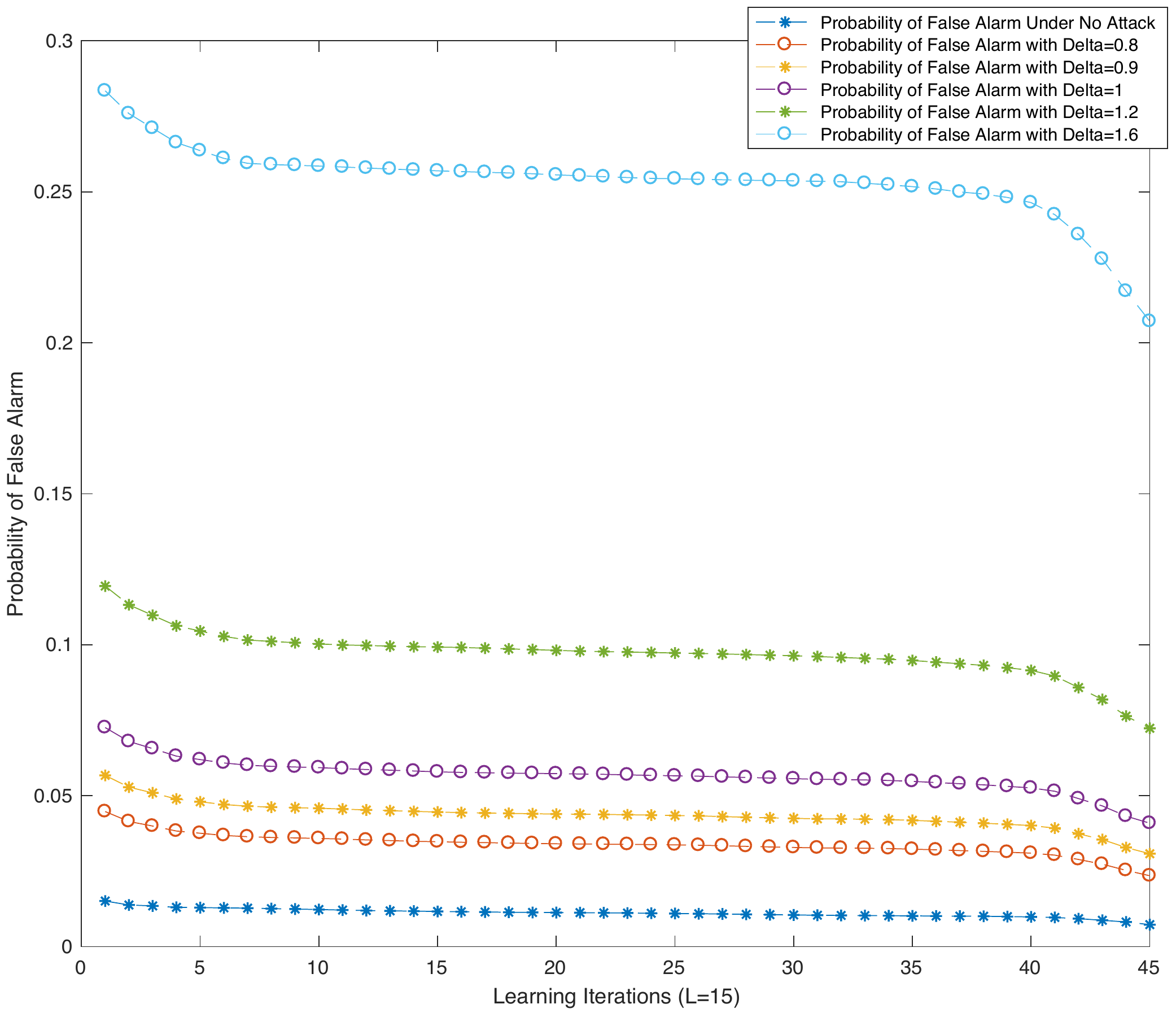}
		\caption{Probability of False Alarm (Detection Interval $L=15$, Attack Parameters: $P_5=0.5$, $\Delta_5$ and $a_5^\prime=a_5+1$).}
		\label{fig:false}
	\end{figure}
\end{example}
\subsection{Steady-State Performance Analysis of the Detection Algorithm}\label{sub:det}
Our detection platform is based on Neyman-Pearson theorem. A Neyman-Pearson based detector measures the signal-to-noise ratio for the unknown signal $y$ over a certain time-interval and detects the presence of the deterministic signal $s$ in $y$ at points, where the signal-to-noise ratio is maximized. It is known that the Neyman-Pearson theorem provides the optimal decision criterion, where the likelihood ratio is compared with a threshold $\gamma$, previously computed to minimize a given false alarm probability \cite{kay1998fundamentals}. The selection of the optimal global threshold $\gamma$ is beyond the scope of the work presented here, and we assume that $\gamma$ in (\ref{epsilon}) has already been selected based on some performance and application criteria. 

We characterize the steady-state performance of our proposed detection platform performance by examining its deflection coefficient against intelligent Byzantine attacks. The deflection coefficient of the test statistic is defined as,
\begin{align}
D(\wedge)=\frac{E[\wedge|H_1]-E[\wedge|H_0]}{E[(\wedge-E[\wedge|H_0])^2|H_0]},
\end{align} 
or the difference of the means (expectations) of the test statistic under the two independent and identically distributed hypothesis distributions, $H_0$ and $H_1$ (with the same variance), divided by the variance of the test statistic under $H_0$. The deflection coefficient formula given above can characterize the steady-state performance of the detection framework and analyze the limitations of the detection procedure by quantifying the distributions under both null and alternative hypotheses based on the number of Byzantine neighbors and attack parameters. This quantification can characterize the distance between the expectations of these two hypotheses to show a limit-case probability of correctly detecting the attacks. Since, the deflection coefficient is directly related to the area of overlapped regions between two distributions, it can efficiently characterize the decision performance in a binary hypothesis testing environment \cite{shen2010deflection,alonso2016adaptive, arka2013selective}. Moreover, the deflection coefficient can be obtained by only calculating the mean and variance from the observed data set without modeling the exact distributions. There is a direct relationship between the detection performance and positive values of the deflection coefficient. Given the event-triggered multi-agent network system design proposed in previous sections and designed according to Theorem \ref{thm:synch}, and the proposed detection framework in Section \ref{sec:Det}, here we characterize the minimum number of Byzantine neighboring agents required to make the detection coefficient for test statistics equal to zero. First, we characterize the relationship between the number of Byzantine neighbors of an honest agent and the performance of its detection unit. Second, we characterize the minimum number of Byzantine agents that can entirely blind the detection unit of a single honest agent.
\begin{theorem} \label{thm:blind}
	Consider an event-triggered multi-agent system designed according to Theorem \ref{thm:synch}. Consider that each agent $G_j$ $(j \in 1...N)$ is equipped with the detection unit proposed in Section \ref{sec:Det}. For an honest agent with $N_H$ honest and $N_B$ Byzantine neighbors, the condition for the detection unit to become entirely blinded or to make the deflection coefficient zero over the detection interval $L$ is,
	\begin{align*}
	\sum_{k=1}^{N_B} LP_k[2\tilde{h}_k\Delta_k(\mu_k-\mu_j)+\tilde{h}^2_k\Delta^2_k(\sigma_k^2-1)]=\sum_{k=1}^{N}\eta_k\sigma_k^2,
	\end{align*}
	where $\eta_k=\frac{\sum_{i=1}^{L}|\tilde{h}_k s^i_k-y_j(t_j^i)|^2}{\sigma_k^2}$, $\mu_j=\frac{1}{L}\sum_{i=1}^{L}y_j^i$ and $\mu_k=\frac{1}{L}\sum_{i=1}^{L}\tilde{h}_ky_k^i$.
\end{theorem}
\begin{proof}
	As mentioned in Section \ref{sec:Det}, each local test statistic $T_k$ over the detection time-interval $L$ between the honest agent $G_j$ and its neighbor $G_k$ may follow a central or non-central chi-square distribution. We define, $\eta_k=\frac{\sum_{i=1}^{L}|\tilde{h}_k y^i_k-y_j^i|^2}{\sigma_k^2}$ for an honest communication from agent $G_k$ to the host agent $G_j$ over the detection interval $L$. We define, $\eta^\prime_k=\frac{\sum_{i=1}^{L}|\tilde{h}_k (y^i_k-\Delta_k)-y_j^i|^2}{\sigma_k^2}$ for a Byzantine communication from agent $G_k$ to the host agent $G_j$ over the detection interval $L$. We can see that, $\eta^\prime_k=\eta_k+\frac{\sum_{i=1}^{L}(\tilde{h}^2_k \Delta^2_k+2\tilde{h}_k\Delta_ky_j^i-2\tilde{h}^2_k\Delta_ky_k^i)}{\sigma_k^2}$. Further, the  true mean $(\mu_{kj})$ and variance $(\sigma_{kj}^2)$ under the null hypothesis $H_0$ and alternative hypothesis $H_1$ for honest communications are as follows,
	\[\mu_{kj}=\begin{cases}
	L\sigma_k^2  & \quad \text{Under $H_0$ } \\
	(L+\eta_k)\sigma_k^2 & \quad \text{Under $H_1$,}
	\end{cases} \]
	\[\sigma_{kj}^2=\begin{cases}
	2L\sigma_k^4   & \quad \text{Under $H_0$ } \\
	2(L+2\eta_k)\sigma_k^4 & \quad \text{Under $H_1$.}
	\end{cases} \]
	Above, $\eta_k \simeq 0$ is implied under $H_0$ or the hypothesis that the two agents have synchronized according to (\ref{epsilon}). For the honest agent $G_j$ with $N$ neighbors where, $N_B$ of them are Byzantine and $N_H$ of them are honest, we may have,
	\begin{align}
	&E[\wedge|H_0]=\sum_{k=1}^{N_H}L\sigma_k^2+\sum_{k=1}^{N_B} [P_k(L+L\tilde{h}^2_k\Delta^2_k)\sigma_k^2+(1-P_k)L\sigma_k^2]\\&
	E[\wedge|H_1]=\sum_{k=1}^{N_H}(L+\eta_k)\sigma_k^2+\sum_{k=1}^{N_B} [P_k((L+\eta^\prime_k)\sigma_k^2)+(1-P_k)(L+\eta_k)\sigma_k^2],\\&
	E[(\wedge-E[\wedge|H_0])^2|H_0]=\sum_{k=1}^{N_H}2L\sigma_k^4+\sum_{k=1}^{N_B} [P_kL^2\tilde{h}^4_k\Delta_k^4\sigma_k^4-P_k^2L^2\tilde{h}^4_k\Delta_k^4\sigma_k^4+2L\sigma_k^4].
	\end{align}
	Utilizing the above definitions into $E[\wedge|H_1]-E[\wedge|H_0]$ and simplifying further we have,
	\begin{align*}
	E[\wedge|H_1]-E[\wedge|H_0]=\sum_{k=1}^{N_H}\eta_k\sigma_k^2+\sum_{k=1}^{N_B} [P_k[\frac{\sum_{i=1}^{L}(\tilde{h}^2_k\Delta^2_k+2\tilde{h}_k\Delta_ky_j^i-2\tilde{h}^2_k\Delta_ky_k^i)}{\sigma_k^2}-L\tilde{h}^2_k\Delta^2_k]\sigma_k^2+\eta_k\sigma_k^2]
	\end{align*}
	
	We denote the means of the output signals of agents $G_j$ and $G_k$ over the detection time-interval $L$ as $\mu_j=\frac{1}{L}\sum_{i=1}^{L}y_j^i$ and $\mu_k=\frac{1}{L}\sum_{i=1}^{L}\tilde{h}_ky_k^i$. For the Byzantine agents to be able to blind the detection unit ($D(\wedge)=0$), they need to enforce $E[\wedge|H_0]=E[\wedge|H_1]$. This means that,
	\begin{align}
	\sum_{k=1}^{N_B} LP_k[2\tilde{h}_k\Delta_k(\mu_k-\mu_j)+\tilde{h}^2_k\Delta^2_k(\sigma_k^2-1)]=\sum_{k=1}^{N}\eta_k\sigma_k^2,
	\end{align}
	where $\eta_k=\frac{\sum_{i=1}^{L}|\tilde{h}_k s^i_k-y_j(t_j^i)|^2}{\sigma_k^2}$. This quantifies the steady-state effects of the number of neighboring Byzantine agents, attack strengths and attack probabilities on the detection unit of an honest agent and also proves the theorem.
\end{proof}
If we assume that $\Delta_k=\Delta$, $P_k=P$, and $\tilde{h}_k=\tilde{h}$ for all $k=1,...N_B$ and $\eta_k= \eta$ and $\sigma_k=\sigma$ for all $k=1,...N$, and quantify the distances between the means of Byzantine agents' outputs and the honest agent's output, $\mu_k-\mu_j=d_k=D$ for $k=1,...N_B$, then the condition given in Theorem \ref{thm:blind} simplifies to $\frac{N_B}{N}=\frac{\eta\sigma^2}{LP[2\tilde{h}\Delta D + \tilde{h}^2\Delta^2(\sigma^2-1)]}$, where $N$ represents the number of neighbors for agent $G_j$. This relation shows that an intelligent Byzantine attack can blind the entire detection framework by an appropriate selection of $P$ and $\Delta$. This also means that blinding the detection framework is still possible even in cases that the Byzantine nodes are in the minority in the neighborhood of the honest agent $G_j$. Moreover, this reveals the trade-off that if the Byzantine agents are in the minority in the neighborhood of agent $G_j$, then they will need to select larger attack parameters ($P$ and $\Delta$), in order to blind the detection unit, this, however, in return makes the job of detection easier for the honest agents. For the honest node, this shows the importance of quick detection of rogue agents. For the Byzantine agents, this shows the importance of quickly occupying the neighborhood of $G_j$ in order to maintain their inconspicuous state and fulfill their adversarial objectives. Lastly, this relationship shows the importance of the distance between the means of outputs of Byzantine agents and the honest agent. In case of a Byzantine attack, the larger the distances between the means of outputs of the occupied agents and the mean of the output of the honest agent are, the easier it is for the Byzantine agents to degrade the performance of the multi-agent systems. In other words, the Byzantine agents will require to exert less effort (smaller values for $\Delta$ and $P$) to blind the detection unit. This also helps the Byzantine agents to stay hidden. However, the job of selecting attack parameters for Byzantine agents becomes more complicated as the multi-agent system gets closer to the synchronized state. In other words, this relation also reveals the trade-off between degrading the performance by the Byzantine attack and the desire to stay hidden from the detection unit.
\begin{example}\label{exm:ex2}
	Consider agent $G_2$ in the event-triggered multi-agent system given in Example \ref{exm:ex1} with three neighbors. We consider the same underlying communication graph and initial conditions for the entire event-triggered multi-agent system. We assume that $G_4$ is a Byzantine neighbor and $G_1$ and $G_5$ are honest neighbors of agent $G_2$. We consider the following channel gains for the communication links between agents $G_2$ and its neighbors: $\tilde{h}_1=0.8,~\tilde{h}_4=0.90,~\tilde{h}_5=0.72$ and assume $\sigma^2_i=1.2$ for all the noise in all the communication links i.e. $i=1,...,5$. The detection units rely on the detection time-interval $L=20$. The deflection coefficient for agent $G_2$'s detection unit is depicted in Fig. \ref{fig:DCoe}. The contour plot (Fig. \ref{fig:contour}) shows the underneath of the three-dimensional shaded surface in Fig. \ref{fig:DCoe} and clarifies the relationship amongst the attack strength parameter $\Delta_4$, the attack probability $P_4$ and the deflection coefficient of the detection unit located on agent $G_2$ i.e. $D_2(\wedge)$. The larger values of $\Delta_4$ in general, makes the job of detection easier, if $P_4$ is kept small enough. For the Byzantine agent to be able to blind the detection unit, a certain balance between $\Delta_4$ and $P_4$ is required. Lastly, Fig \ref{fig:DCoe} shows that it is possible for a single Byzantine agent to blind the agent $G_2$'s detection unit, even though the majority of $G_2$'s neighbors are honest nodes. Indeed, the fact that $\frac{1}{3}$ of the $G_2$'s neighborhood is occupied by a Byzantine agent has severely degraded the detection performance. 
	\begin{figure}[!t]
		\centering
		\includegraphics[scale = 0.5]{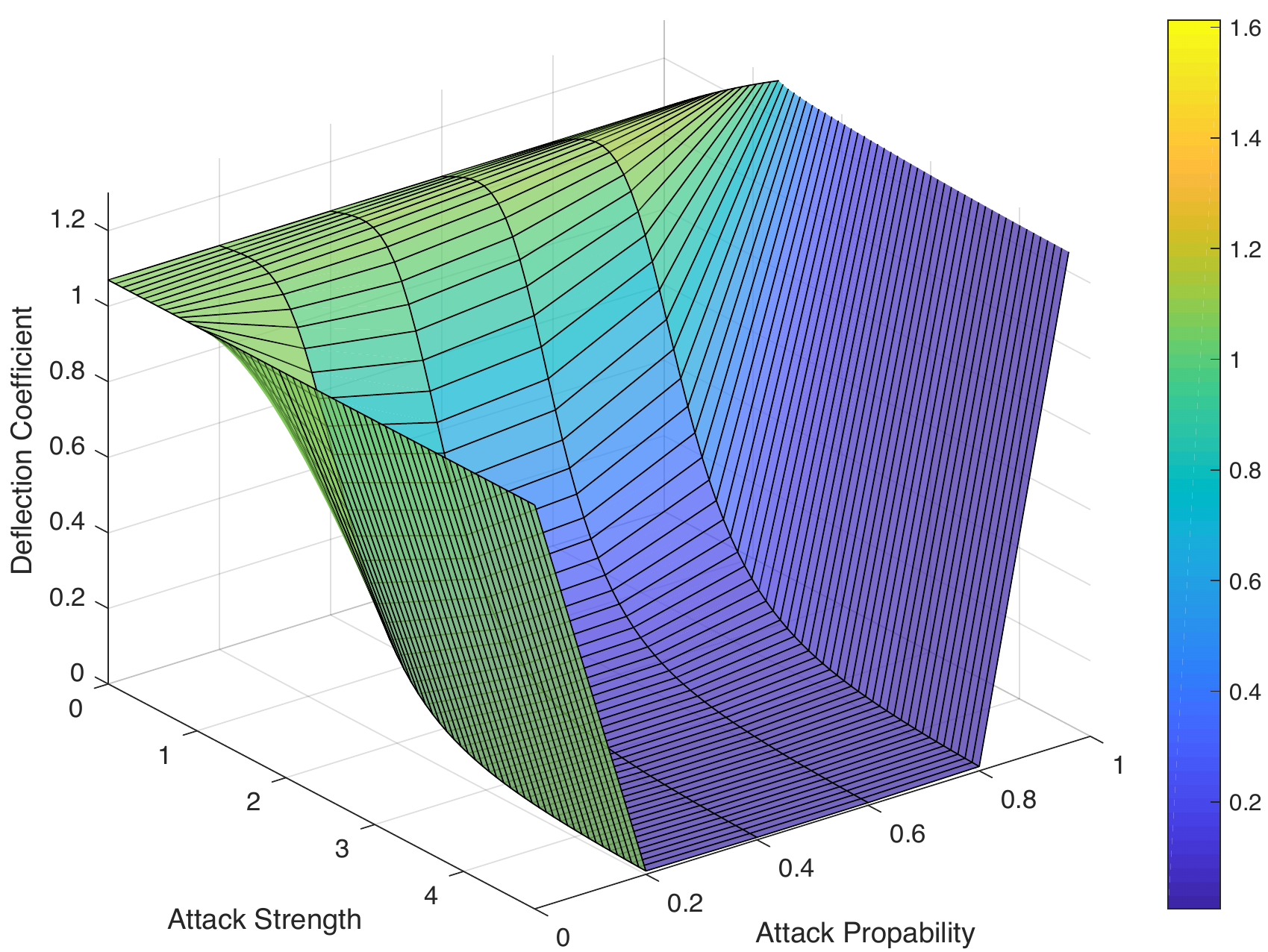}
		\caption{Deflection Coefficient for agent $G_2$ in Example \ref{exm:ex1} as a function of Attack Probability $P$ and Attack Strength $\Delta$.}
		\label{fig:DCoe}
	\end{figure}
	\begin{figure}[!t]
		\centering
		\includegraphics[scale = 0.5]{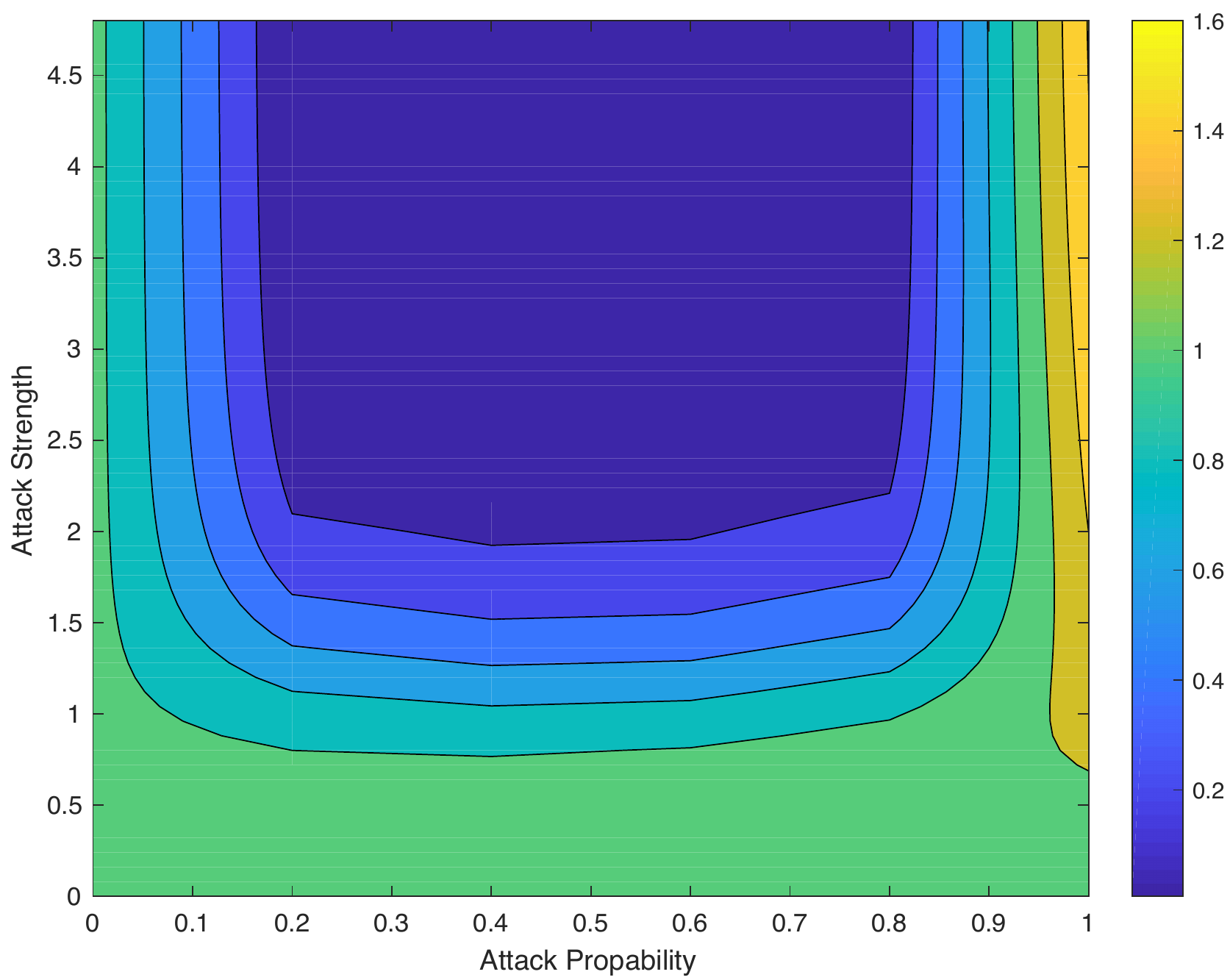}
		\caption{Contour plot of the Deflection Coefficient for agent $G_2$ as a function of Attack Probability $P$ and Attack Strength $\Delta$.}
		\label{fig:contour}
	\end{figure}
\end{example}
Next, we will propose two different learning-based control approaches to mitigate the negative effects mentioned above, imposed by Byzantine agents on the multi-agent system.
\section{A Learning-Based Control Method for Mitigating the Effects of the Byzantine Attack}\label{sec:det}
\subsection{Distributed Weight Assignments for Mitigating the Effects of Weight Manipulations}\label{sub:mitw}
In this subsection, we propose a robust distributed weight design that will achieve synchronization and in the case of an attack mitigate the adversarial effects of Byzantine agents. First, we will deal with the problem of weight manipulation. Common in the literature, it is assumed that the feedback control framework given in (\ref{inp}) is designed by the agent itself. This will leave the entire event-triggered multi-agent framework extremely vulnerable to adversarial attempts such as node capture or Byzantine attacks as the nodes themselves can independently have a great influence on the synchronization efforts (in our case, this great influence was quantified in Subsection \ref{sub:eff}). Here, we propose a synchronization algorithm in which the weights for the feedback control for agent $G_j$ are assigned by its neighbors $G_k$, where $k \in \mathcal{N}_j^{in}$. In other words, with sending its first information, the neighbor $G_k$ also sends its desired feedback weight which will be used to initiate the feedback control of agent $G_j$ in order to reach synchronization. Under this framework, we assume that each agent $G_j$ is aware of its $d_{in}(j)$ (defined in Section \ref{sec:gra}), similar to before the communication framework is balanced ($d_{in}(j)=d_{out}(j)$) and that each agent $G_j$ has only authority over designing its own triggering condition by selecting its design parameter $\delta_j$ \textemdash feedback weights are assigned by neighbors. As a result, a Byzantine agent, instead of being able to diverge the entire behavior of an overtaken agent in the event-triggered multi-agent network system and consequently mislead the entire multi-agent system, is only able to partly distract the proper behavior of its neighboring agents. An honest agent will only be taken over entirely, if the majority of its neighbors are Byzantine. This is highly unlikely in the presence of a detection framework. This will extremely lighten the burden of the mitigation process and improve the overall performance. 
\begin{theorem}\label{thm:synchAss}
Consider the event-triggered multi-agent network system described in Section \ref{sec:prob}, where each sub-system $G_j$ is output passive with the output passivity index $\rho_j$ and is controlled by the input given in (\ref{inp}). Consider that the feedback weights in (\ref{inp}) for each sub-system $G_j$  are assigned by its neighbors $G_k$, where $k \in \mathcal{N}_j^{in}$. If the underlying connected communication graph resulting from weight assignments is balanced, the communication time-delays and disturbances are negligible, and the communication attempts amongst all agents $G_j$ where $j=1,..., N$, are governed by the triggering conditions,
	\begin{align*}
	&||e_j(t)||_2^2 > \delta_j||y_j(t)||_2^2,
	\end{align*}
	where the design parameters $\delta_j$ are chosen such that,
	\begin{align*}
	0 < \delta_j \leq  \frac{\frac{2}{|\mathcal{N}_j^{in}|}(\lambda(\tilde{G})+\rho_j)-\frac{1}{ \alpha}-\frac{1}{ \beta}}{\alpha+\beta},
	\end{align*}
	where $\alpha>0$ and $\beta>0$ are design variables and $\lambda(\tilde{G})$ is the connectivity of the underlying communication graph, then the entire event-triggered multi-agent network system achieves output synchronization asymptotically.
\end{theorem}
\begin{proof}
Proof is similar to Theorem \ref{thm:synch}, as we assume that all agents are honest and initiate the neighbors' feedback control based on their respective $d_{in}$ according to the assumption that the resulting underlying communication graph is balanced. One can define a modified Laplacian matrix for the communication graph $\tilde{L}=\tilde{D}-\tilde{A}$, where $\tilde{D}$ is an $N \times N$ diagonal matrix with $\tilde{d}_{j,j} = d_{in}(j)$, representing the sum of assigned weights to agent $G_j$ and $\tilde{A}$ is the $N \times N$ adjacency matrix with $\tilde{a}_{i,j} \neq 0$ representing the gain assigned by agent $G_i$ to agent $G_j$, and $\tilde{a}_{i,j}=0$ when there is no communication link between two agents $G_i$ and $G_j$. Since the communication link is balanced ($d_{in}(j)=d_{out}(j)$ for $\forall j=1,...,N$), $\tilde{L}=L$, similarly, $\lambda(\tilde{G})=\lambda(G)$ then one can simply represent the new framework based on the previous one and show synchronization for the entire event-triggered multi-agent network system following the same steps presented in Theorem \ref{thm:synch}.
\end{proof}
Next, we analyze the effects of a Byzantine attack where a Byzantine agent $G_k$ will disturb the balanced communication graph through weight manipulation by assigning $a_k^B=a_k + \omega_k$ to its neighbors, where $\omega_k>0$. For analytical tractability, we do not consider the data falsification and will show that the new approach will mitigate the negative effects of weight manipulation. In the next subsection, we will discuss the mitigation process for the data falsification part of the Byzantine attack. Similar to the previous sections, we assume that amongst the $N$ agents, there are $N_B$ Byzantine nodes with the attack model described in Section \ref{sec:Byz} and $N_H$ honest nodes ($N_H+N_B=N$). $\mathcal{N}_H^{}$ and $\mathcal{N}_B^{}$ represent the set of honest and Byzantine agents, respectively. We represent the honest and Byzantine neighboring agents for $G_j$ by $\mathcal{N}_j^{in_H}$ and $\mathcal{N}_j^{in_B}$ ($\mathcal{N}_j^{in_H} \cap \mathcal{N}_j^{in_B}=\emptyset$, $\mathcal{N}_j^{in_H} \cup \mathcal{N}_j^{in_B}=\mathcal{N}_j^{in}$). $|\mathcal{N}_j^{in}|$ represents the same cardinality definition as given in Theorem \ref{thm:synch}. The set of all Byzantine agents is represented by $\mathcal{N}^{B}$ and the set of all honest agents is represented by $\mathcal{N}^{H}$. It is important to note that for the case where the feedback weights are assigned by the neighbors, the Byzantine neighbor $G_k^B$ assigns $a_k^B=a_k + \omega_k$ to the feedback control for agent $G_j$, otherwise, $a_k^H=a_k$ is assigned. The Lyapunov storage function for the entire event-triggered multi-agent network system becomes,
\begin{align*}
	&\dot{S}= \sum_{j=1}^{N} \dot{V}_j \leq \sum_{j=1}^{N} \sum_{k \in  \mathcal{N}_j^{in}}^{} a_k [(y_k(t)-y_j(t))-(e_k(t)-e_j(t))]^Ty_j(t)\\&+\sum_{j=1}^{N} \sum_{k \in  \mathcal{N}_j^{in_B}}^{} \omega_k [(y_k(t)-y_j(t))-(e_k(t)-e_j(t))]^Ty_j(t)-\sum_{j=1}^{N}\rho_jy^T_j(t)y_j(t).
\end{align*}
It is important to note that $\omega_k=0$ for honest neighbors. First, it can be shown that,
\begin{align*}
	&\sum_{j=1}^{N} \sum_{k \in  \mathcal{N}_j^{in_B}}^{} \omega_k (y_k(t)-y_j(t))^Ty_j(t)
	\leq\sum_{j=1}^{N} \sum_{k \in  \mathcal{N}_j^{in_B}}^{}\frac{\omega_k y^T_k(t)y_k(t)}{4}. \numberthis \label{R1}
\end{align*}
We follow the same approach as before and end up with, 
\begin{align*}
	&\dot{S}= \sum_{j=1}^{N} \dot{V}_j \leq -Y^TL^TY+Y^TL^{ T}E-\sum_{j=1}^{N}\rho_jy^T_j(t)y_j(t)\\&~~+\sum_{j=1}^{N} \sum_{k \in  \mathcal{N}_j^{in_B}}^{}\frac{\omega_k y^T_k(t)y_k(t)}{4}-\sum_{j=1}^{N} \sum_{k \in  \mathcal{N}_j^{in_B}}^{} \omega_k [(e_k(t)-e_j(t))]^Ty_j(t)\\& ~~=-Y^TL^TY+Y^TL^{\prime^ T}E-\sum_{j=1}^{N}\rho_jy^T_j(t)y_j(t)+\sum_{j=1}^{N} \sum_{k \in  \mathcal{N}_j^{in_B}}^{}\frac{\omega_k y^T_k(t)y_k(t)}{4}, \numberthis \label{R2}
\end{align*}
where $\alpha$ and $\beta$ are the same parameters as given in Theorem \ref{thm:synch}. We denote, $|\mathcal{W}_j|$ as the sum of the weight manipulations that were assigned to the agent $G_j$ from its Byzantine neighbors. $L^{\prime}$ is the Laplacian matrix of the new underlying communication graph consisting of $a^\prime_k$'s and is defined as,
\[[L^\prime]_{j,i} = \begin{cases}
\sum_{k \in  \mathcal{N}_j^{in}}^{} a_k^\prime  & \quad \text{if } j=i \\
-a_k^\prime  & \quad \text{if there is an arc from $G_i$ to $G_j$ with the gain $a_k^\prime$,}
\end{cases} \]
where $a^\prime_k$ are defined as before. We may follow the same steps as given in Theorem \ref{thm:synch}, and get to the following,
\begin{align*}
	&\dot{S}= \sum_{j=1}^{N} \dot{V}_j \leq -Y^TL^TY+\sum_{j=1}^{N} (|\mathcal{N}_j^{in}|+|\mathcal{W}_j|) [\frac{(\alpha+\beta) \delta_j}{2} + \frac{1}{2 \alpha}]y_j^T(t)y_j(t)\\&~~+\sum_{j=1}^{N}\sum_{k \in  \mathcal{N}_j^{in}}^{} (a_k+\omega_k) [\frac{y_k^T(t)y_k(t)}{2 \beta}]-\sum_{j=1}^{N}\rho_jy^T_j(t)y_j(t)+\sum_{j=1}^{N} \sum_{k \in  \mathcal{N}_j^{in_B}}^{}\frac{\omega_k y^T_k(t)y_k(t)}{4}\\&~~\leq -\lambda(G)Y^TY +\sum_{j=1}^{N} (|\mathcal{N}_j^{in}|+|\mathcal{W}_j|) [\frac{(\alpha+\beta) \delta_j}{2} + \frac{1}{2 \alpha}]y_j^T(t)y_j(t)\\&~~+\sum_{j=1}^{N}\sum_{k \in  \mathcal{N}_j^{in}}^{} (a_k+\omega_k) [\frac{y_k^T(t)y_k(t)}{2 \beta}]-\sum_{j=1}^{N}\rho_jy^T_j(t)y_j(t)+\sum_{j=1}^{N} \sum_{k \in  \mathcal{N}_j^{in_B}}^{}\frac{\omega_k y^T_k(t)y_k(t)}{4}, \numberthis \label{R3}
\end{align*}
We introduce the same square diagonal matrix $\Theta \in R^{N\times N}$, where,
\[ [\Theta]_{j,i} = \begin{cases}
+\lambda(G)+\rho_j-|\mathcal{N}_j^{in}| [\frac{(\alpha+\beta) \delta_j}{2} + \frac{1}{2 \alpha}+\frac{1}{2 \beta}]     & \quad \text{if } j=i \\
0  & \quad \text{otherwise.}
\end{cases}\]
Given (\ref{R3}) and $\Theta$, we have,
\begin{align*} 
	&\dot{S}\leq - Y^T_\Delta \Theta Y_\Delta +\sum_{j=1}^{N} |\mathcal{W}_j| [\frac{(\alpha+\beta) \delta_j}{2}+\frac{1}{2 \alpha}]y_j^T(t)y_j(t)
	+\sum_{j=1}^{N} \sum_{k \in  \mathcal{N}_j^{in_B}}^{}\omega_k(\frac{1}{4}+\frac{1}{2\beta}) y^T_k(t)y_k(t). \numberthis \label{R4}
\end{align*}
Given the assumption that the multi-agent system was initially designed according to Theorem \ref{thm:synch}, we have $\Theta>0$. After simplifying, and given $\dot{S} \rightarrow 0$ as $t\rightarrow \infty$, we have, 
\begin{align*} 
	&0< Y^T_\Delta \Theta Y_\Delta \leq \sum_{j=1}^{N} |\mathcal{W}_j| [\frac{(\alpha+\beta) \delta_j}{2}+\frac{1}{2 \alpha}]y_j^T(t)y_j(t)
	+\sum_{j=1}^{N} \sum_{k \in  \mathcal{N}_j^{in_B}}^{}\omega_k(\frac{1}{4}+\frac{1}{2\beta}) y^T_k(t)y_k(t). \numberthis \label{R4}
\end{align*}
Comparing (\ref{R4}) with (\ref{relation7}), one can see that assignments of the agents' weights by their neighbors can greatly decrease the magnitude of the upper-bound on the deviations from the synchronized state. Additionally, by not allowing the Byzantine agents design their own weights, we are diversifying the negative effects caused by the weight manipulations. We are dividing the weight manipulation attack into two parts. One part is still related to the Byzantine agents and their outputs and cannot be mitigated without a direct access to the corrupt agents (second part of the summation given in (\ref{R4})). The first part of the upper-bound shown in (\ref{R4}), however, may be mitigated by the honest agents given their passivity indices and their design of the triggering conditions. This greatly helps with the synchronization process and lowers the upper-bound of the deviations. Moreover, if a Byzantine agent only has honest neighbors, through this mitigation process, the output of the Byzantine agent will eventually reach synchronization as the information and weights received by the Byzantine agent from its honest neighbors will follow the requirements given in Theorem \ref{thm:synchAss}. Consequently, the honest agents will be able to entirely mitigate the negative effects of weight manipulations \textendash This will be illustrated in Example \ref{exm:ex4}. This combined with the detection framework presented in the next section for dealing with data falsifications, can completely eradicate the negative effects of the Byzantine attack. We will explain this in more details next.

\bf{Mitigating the effects of Weight Manipulation by utilizing the Passivity Properties of Agents}:\normalfont ~ As it was characterized before passivity can ameliorate the effects of a Byzantine attack. For all agents $G_j$ where $j=1,...,N$, we can represent the passivity indices with, $\rho_j=\rho_j^\prime+\rho^\Delta_j>0$, where $\rho_j^\prime>0$ and $\rho^\Delta_j>0$. We assume that the triggering conditions are designed according to Theorem \ref{thm:synchAss}, where,  
\begin{align*}
0 < \delta_j \leq  \frac{\frac{2}{|\mathcal{N}_j^{in}|}(\lambda(\tilde{G})+\rho_j)-\frac{1}{ \alpha}-\frac{1}{ \beta}}{\alpha+\beta},
\end{align*}
Simplifying this relation based on $\rho_j$, and annotating $\rho_j^\prime=\lambda(G)-|\mathcal{N}_j^{in}| [\frac{(\alpha+\beta) \delta_j}{2} + \frac{1}{2 \alpha}+\frac{1}{2 \beta}]$, we may have, $\rho_j-\rho^\Delta_j=\rho_j^\prime$. As a result, if the triggering conditions are designed such that,
\begin{align*}
\rho^\Delta_j> \sum_{k \in  \mathcal{N}_j^{in_B}}^{} \omega_k [\frac{(\alpha+\beta) \delta_j}{2}+\frac{1}{2 \alpha}],
\end{align*}
where $\omega_k$'s are weight manipulations caused by the Byzantine neighbors, then the effects of weight manipulations committed by the Byzantine agents are assuaged. To extrapolate this result to the entire event-triggered multi-agent system, we assume that the triggering conditions are designed such that the above relation holds for each agent, we introduce the positive definite diagonal $N \times N$ matrix  $\Theta^\Delta \in R^{N\times N}$, where,
\[ [\Theta^\Delta]_{j,i} = \begin{cases}
\rho^\Delta_j-\sum_{k \in  \mathcal{N}_j^{in_B}}^{} \omega_k [\frac{(\alpha+\beta) \delta_j}{2}+\frac{1}{2 \alpha}]>0  & \quad \text{if } j=i \\
0  & \quad \text{otherwise.}
\end{cases}\]
As a result, (\ref{R4}) becomes,
\begin{align*} 
&0< Y^T_\Delta \Theta Y_\Delta \leq -Y^T \Theta^{\Delta} Y 
+\sum_{j=1}^{N} \sum_{k \in  \mathcal{N}_j^{in_B}}^{}\omega_k(\frac{1}{4}+\frac{1}{2\beta}) y^T_k(t)y_k(t). \numberthis \label{R5}
\end{align*}
Defining the positive definite matrix $\Theta^\prime=\Theta + \Theta^{\Delta}$, and simplifying further we have, 
\begin{align*} 
&0< Y^T_\Delta \Theta^\prime Y_\Delta \leq 
\sum_{j=1}^{N} \sum_{k \in  \mathcal{N}_j^{in_B}}^{}\omega_k(\frac{1}{4}+\frac{1}{2\beta}) y^T_k(t)y_k(t). \numberthis \label{R6}
\end{align*}
By comparing (\ref{R6}) with (\ref{R4}) and consequently (\ref{relation6}), one can see that the new method of distributed weight assignments in conjunction with the utilization of the passivity qualities of sub-systems can greatly mitigate the effects of a Byzantine attack's weight manipulations. It is important to note that by giving the authority to the honest agents to be able to adjust their triggering conditions according to their assessment of the magnitude of weight manipulations committed by their neighboring Byzantine agents, one can entirely mitigate this part of the Byzantine attack. This is done based on the fact that the honest agent $G_j$ is aware of its $d_{in}(j)$ and can estimate the weight manipulations by observing the difference between its $d_{in}(j)$ and the actual weights assigned to $G_j$ by its neighbors. Moreover, by decreasing the magnitude of the design variable $\delta_j$ (shortening the triggering intervals - increasing the communication rate), $G_j$ may increase $\rho^\Delta_j$ and compensate for the negative effects of weight manipulations. This will entirely eradicate the negative effects of weight manipulations committed by isolated Byzantine agents (with no Byzantine neighbors). However, in order to mitigate the negative effects of weight manipulations in cases where the Byzantine agents have Byzantine neighbors will require further attention. The mitigation method offered in the next section combined with the detection framework will attempt to further mitigate these negative effects. Lastly, one can initiate the design of the event-triggered multi-agent system by selecting smaller values for $\delta_j$, $j=1,...,N$  (a more conservative event-triggered design). This will generally result in a more resilient event-triggered multi-agent network system against Byzantine attacks.
\begin{example}\label{exm:ex4}
We consider an event-triggered multi-agent network system consisting of four agents ($i=1,...,4$) with the underlying balanced communication topology given in Fig. \ref{fig:ex4}. We assume the following dynamics for agents, 
	\[ G_{i} = \begin{cases}
	\dot{x}_i(t)= - c_i x_i(t) + u_i(t)   \\
	y_i(t)=x_i(t), & \\
	\end{cases} \]
with $c_1=1.2, c_2=1.8, c_3=2.6, c_4=0.80$ and initial conditions, $y_1(0)=2,~y_2(0)=-10,~y_3(0)=1$, and $y_4(0)=-2$. One can verify that all agents are dissipative with the storage function $V_i(x)=\frac{1}{2}x_i^T(t)x_i(t)$. This results in output passivity indices $\rho_1=1.2,~\rho_2=1.8,~\rho_3=2.6,~\rho_4=0.8$ for the agents. The Laplacian matrix of the underlying communication graph amongst agents before Byzantine attack is balances as follows,
\begin{align*}
	L=\begin{bmatrix} 
			1 & -1  &  0 &  0 \\
			0 &  2  & -2 &  0 \\
			0 &  0  &  2 & -2 \\
		   -1 & -1  &  0 &  2
	\end{bmatrix}.
\end{align*}
with the connectivity measure, $\lambda(G)=2$. We assume that the event-triggered multi-agent network system is designed based on Theorem \ref{thm:synchAss} with the following triggering conditions,
	\begin{align*}
		&||e_1(t)||_2^2>0.80||y_1(t)||_2^2,\\&
		||e_2(t)||_2^2>0.64||y_2(t)||_2^2,\\&
		||e_3(t)||_2^2>0.60||y_3(t)||_2^2,\\&
		||e_4(t)||_2^2>0.35||y_4(t)||_2^2,
	\end{align*}
by selecting $\alpha_i=1,~\beta_i=1$ for $i=1,...,4$. We assume $G_2$ and $G_4$ are Byzantine agents and instead of assigning correct weights $a_2=1$ and $a_4=2$ to agents $G_1$ and $G_3$, they assign $a_2+\omega_2$ and  $a_4+\omega_4$ to agents $G_1$ and $G_3$ where $\omega_2=2.5$ and $\omega_4=2$. At time $t=0.4s$, in order to mitigate the attack, honest agents $G_1$ and $G_3$ increase $\rho^\Delta_1$ and $\rho^\Delta_3$, and their communication rate by shortening their triggering intervals. They alter their triggering parameters from $\delta_1=0.80$ and $\delta_3=0.60$ to $\delta_1=0.40$ and $\delta_3=0.15$. Fig.\ref{fig:ex4out} shows that the system synchronizes as a consequence of this mitigation attempt. One can see clearly in Fig.\ref{fig:ex4out} that due to the Byzantine attack the agents diverge at first and it is only after the honest agents mitigate the attack by following the steps given in Sub-Section \ref{sub:mitw} that the multi-agent system takes some corrective steps and eventually synchronizes. 
\begin{figure}[!t]
	\centering
	\includegraphics[scale = 0.7]{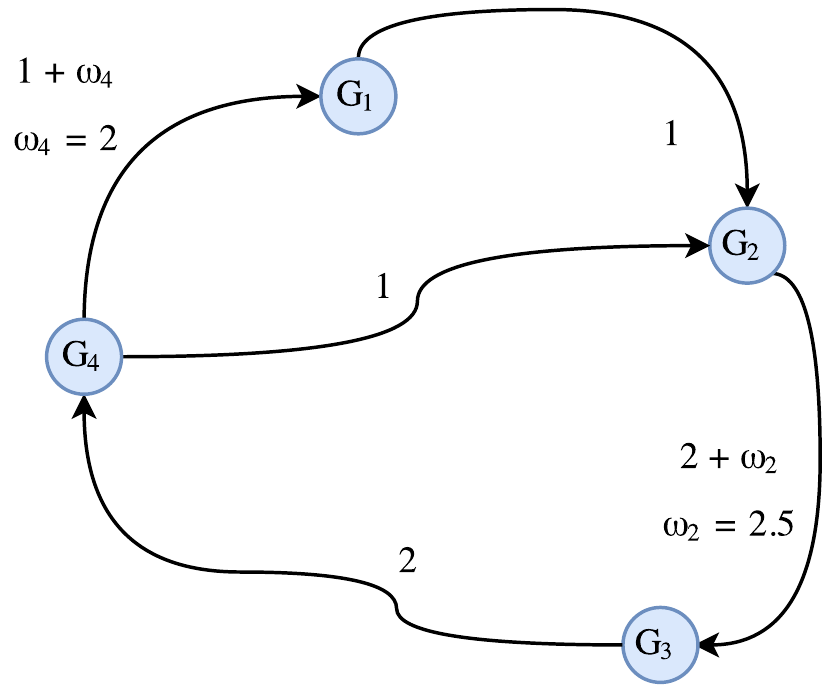}
	\caption{The Underlying Communication Graph for the Multi-Agent System Presented in Example 4 (Attack Parameters: $\omega_2=2.5$ and $\omega_4=2$).}
	\label{fig:ex4}
\end{figure}
\begin{figure}[!t]
	\centering
	\includegraphics[scale = 0.5]{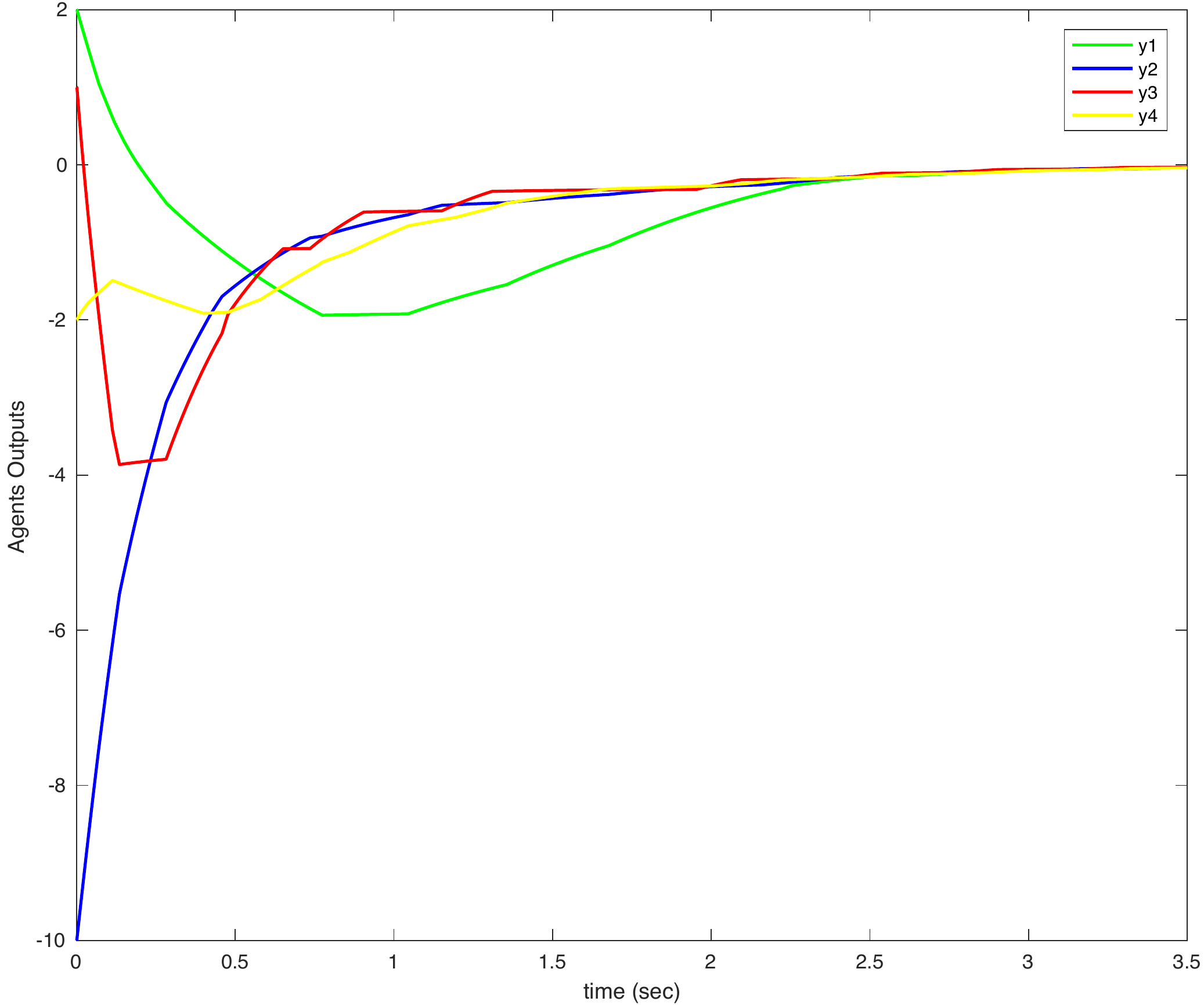}
	\caption{The Outputs of the Multi-Agent System in the Presence of Weight-Manipulation Byzantine Attack (Attack Parameters: $\omega_2=2.5$ and $\omega_4=2$) after the Mitigation Process.}
	\label{fig:ex4out}
\end{figure}
\end{example}
\subsection{A Learning-Based Distributed Algorithm for Mitigating the Effects of Data Falsification}\label{sub:learn}
\bf{Identifying the Byzantine Agents}:\normalfont ~ Here, we propose an algorithm based on which each agent $G_j$ is able to identify each of its neighbors $G_k$, where $k \in \mathcal{N}_j^{in}$ as an honest or Byzantine neighbor. The identification of each neighboring agent is done in order to realize whether the receiving information is trustworthy or not. This is necessary for the next step of the mitigation process. We already know that based on a sufficiently large number of detection data-points of length $L$ coming from the neighbor agent $G_k$, one can postulate that if the neighbor agent $G_k$ is honest, then the data points will follow a normal distribution conditioned on the hypothesis $H_i (i=0,1)$, namely $f_{PDF}(T_k|H_i)=\mathcal{N}((\mu_{i0})_k,(\sigma_{i0}^2)_k)$ where $i=0,1$ and $k \in \mathcal{N}_j^{in}$. The exact form of $(\mu_{i0})_k$ and $(\sigma_{i0}^2)_k$ are given in (\ref{D1}), (\ref{D2}), (\ref{D3}) and (\ref{D4}), and the test statistics $T_k$ is defined in Sub-section \ref{sub:transdet}. $f_{PDF}(T_k|H_i)$ is the probability density function (PDF) of test statistics under each hypothesis $H_i$, $i=0,1$. Similarly, if the neighbor agent $G_k$ is Byzantine, then the data coming from this neighbor is a Gaussian mixture from $\mathcal{N}((\mu_{i0})_k,(\sigma_{i0}^2)_k)$ with the probability of $1-P_k$ and $\mathcal{N}((\mu_{i1})_k,(\sigma_{i1}^2)_k)$ with the probability of $P_k$, where $i=0,1$ and $k \in \mathcal{N}_j^{in}$. As a result, the decision making process will be based on a hypothesis testing framework, where it is decided that the neighboring agent $G_k$ is honest under the hypothesis $H_i$, i.e. $(Dec^i_k)_0$, if the receiving data can be justified by the Gaussian distribution expected under the hypothesis $H_i$. Otherwise, if a Gaussian mixture justifies the data points better, $(Dec^i_k)_1$ under the hypothesis $i$ is decided. This means that it is decided that $G_k$ is Byzantine, if the receiving data from $G_k$ follows the distribution of the expected Gaussian mixture under the hypothesis $H_i$. This decision may be made using the maximum likelihood decision rule \cite{bickel2015mathematical},
\begin{align} \label{DecMak}
\frac{f_{PDF}(T_k|(Dec^i_k)_0)}{f_{PDF}(T_k|(Dec^i_k)_1)} \gtrless^H_B 1.
\end{align}
However, these distributions are unknown to the honest agents and the detection unit should learn these distributions' respective parameters. Next, we will cover the estimation process based on the proposed framework.

\bf{Learning the Distributions' Parameters}:\normalfont ~ The parameters in (\ref{DecMak}) are unknown and should be estimated. The formulation given in (\ref{DecMak}) is parametric and consequently, the framework may be looked at as a parametric statistical estimation problem \cite{bickel2015mathematical}. Accordingly, under the hypothesis $H_i (i=0,1)$, for the honest neighboring agents $G_k$ where $k \in \mathcal{N}_j^{in_H}$, the parameters to be estimated are $\theta_0=((\mu_{i0})_k,(\sigma_{i0}^2)_k)$ and for the Byzantine neighboring agents $G_k$ where $k \in \mathcal{N}_j^{in_B}$, the parameters to be estimated are the ones in the set $\theta_0$ and  $\theta_1=((\mu_{i1})_k,(\sigma_{i1}^2)_k,P_k)$. Here, we propose a learning-based algorithm for estimating the parameter sets $\theta_0$ and $\theta_1$. We annotate the estimation of the means and variances as, $(\tilde{\mu}_{i0})_k, (\tilde{\mu}_{i1})_k,(\tilde{\sigma}_{i0}^2)_k$, $(\tilde{\sigma}_{i1}^2)_k$, and $\tilde{P}_k$ for $H_i (i=0,1)$ and $k \in \mathcal{N}_j^{in_H}$ or $k \in \mathcal{N}_j^{in_B}$. Let us assume a detection time-interval of length $L$ where, the time-interval includes $L$ triggering instances that may be utilized as the sample points for the estimation \cite{gubner2006probability}. For example, between neighboring agents $G_j$ and $G_k$ at the discrete time-instance $i$, the test statistics attained by agent $G_j$ in regard to the neighbor $G_k$ becomes $t_k^i=\sum_{n=i-L}^{i}|y_k^n-y_j^n|^2$, and the $L$ number of samples $y_k^n$ and $y_j^n$ depend on the most recent outputs of the respective event-detectors at the time-instance $n$ during the detection interval $L$. 

In order to estimate the parameters in the set $\theta_0$ for the case of honest agents under the hypotheses $H_0$ and $H_1$, we can simply utilize the method of moments \cite{bickel2015mathematical}. For the honest neighbors, we know that the data should preferably follow a normal distribution with the means and variences given in (\ref{D1}), (\ref{D2}), (\ref{D3}) and (\ref{D4}). We assume the learning iterations of length $L_p$ (each learning iteration consists of $L_p$ data points). Each learning iteration may consist of one or two sets of data belonging to the hypothesis $H_0$ and $H_1$, respectively. At the learning iteration $l$, we may have $T_k^{(l)}=[(t_k^1)^0,(t_k^2)^0,...,(t_k^{L_0})^0,(t_k^1)^1,(t_k^2)^1,...,(t_k^{L_1})^1]$, where $L_0+L_1=L_p$. Given $L_i$ ($i=0,1$) data points, for a normal distribution and for the learning iteration $l$ with $L_p$ total number of data points, the first and second moments theoretically may be represented as,
\begin{align*}
(m_1)^i_k=\frac{1}{L_i}\sum_{j=1}^{L_i}t_k^j,
\end{align*} and,
\begin{align*}
(m_2)^i_k=\frac{1}{L_i}\sum_{j=1}^{L_i}{t_k^j}^2,
\end{align*}
for $H_i (i=0,1)$ and $k \in \mathcal{N}_j^{in_H}$. Consequently, $(m_1)^i_k= (\tilde{\mu}_{i0})_k$ is the estimator for the sample mean for $H_i (i=0,1)$ and $k \in \mathcal{N}_j^{in_H}$ at the learning iteration $l$. And for the variances, we have, $(\tilde{\sigma}_{i0}^2)_k+(\tilde{\mu}_{i0})_k^2=(m_2)^i_k$ for $H_i (i=0,1)$ with $k \in \mathcal{N}_j^{in_H}$. As a result, we may have,
\begin{align*}
(\tilde{\sigma}_{i0}^2)_k=\frac{1}{L_i}\sum_{j=1}^{L_i}{t_k^j}^2-(\frac{1}{L_i}\sum_{j=1}^{L_i}t_k^j)^2=\frac{1}{L_i}\sum_{j=1}^{L_i} (t_k^j-(\tilde{\mu}_{i0})_k)^2,
\end{align*}
for $H_i (i=0,1)$ and $k \in \mathcal{N}_j^{in_H}$. To sum up, the learned parameter set $\theta_0$ at the learning iteration $l$ for an honest communication between two honest agents becomes,
\begin{align}
\tilde{\theta}_0=((\tilde{\mu}_{i0})_k,(\tilde{\sigma}_{i0}^2)_k)=(\frac{1}{L_i}\sum_{j=1}^{L_i}t_k^j,\frac{1}{L_i}\sum_{j=1}^{L_i} (t_k^j-(\tilde{\mu}_{i0})_k)^2),
\end{align} 
for $H_i (i=0,1)$ and $k \in \mathcal{N}_j^{in_H}$.

Next, we define the the complete learning process for an honest neighboring agent based on the above estimators. Each learning phase consists of $L_p$ detection time-intervals of length $L$ (where the parameters are estimated as above) or $L_p\times L$ data points. One can recall from Sub-Section \ref{sub:decmakste} that each agent is consistently deciding whether the system has reached synchronization ($Decision_{syn}$), i.e. $H_0$ is the correct hypothesis, or otherwise, i.e. $H_1$ is the correct hypothesis. As a result, each data point $l_i$ after the detection time-interval of length $L$ under the set of the $L_p$ data points comes with an index indicating if the estimation is happening under the hypothesis $H_0$ or $H_1$. We annotated the number of these estimation as $L_i$ given the hypothesis $H_i$ ($i=0,1$). Needless to say $L_0+L_1=L_p$. As an application related side-note, one can set a required lower-bound for the number of data points under hypothesis $H_i$ ($i=0,1$) before the estimation (learning process) for the parameters under the hypothesis starts. For instance, one can require, $L_1\geq \tau_1$, before the learning process starts for the parameters under $H_1$. For the learning process that starts at time $l^i+1$ respectively for ($i=0,1$) and when the next $L_p$ data points are available, we already have, $(\tilde{\mu}_0)_k=[(\tilde{\mu}_{00})_k^0,...,(\tilde{\mu}_{00})_k^{l^0}]$, $(\tilde{\sigma}_0^2)_k=[(\tilde{\sigma}_{00}^2)_k^0,...,(\tilde{\sigma}_{00}^2)_k^{l^0}]$ and similarly, $(\tilde{\mu}_1)_k=[(\tilde{\mu}_{10})_k^0,...,(\tilde{\mu}_{10})_k^{l^1}]$, $(\tilde{\sigma}_1^2)_k=[(\tilde{\sigma}_{10}^2)_k^0,...,(\tilde{\sigma}_{10}^2)_k^{l^1}]$. We can define our so-far estimates as $(\tilde{\mu}_{00})_k(l^0)=\frac{\sum_{j=1}^{l^0}(\tilde{\mu}_{00})_k^j}{l^0}$, $(\tilde{\mu}_{10})_k(l^1)=\frac{\sum_{j=1}^{l^1}(\tilde{\mu}_{10})_k^j}{l^1}$ as the (current) estimate for the means under each hypothesis and $(\tilde{\sigma}_{00}^2)_k(l^0)=\frac{\sum_{j=1}^{l^0}(\tilde{\sigma}_{00}^2)_k^j}{l^0}$ and $(\tilde{\sigma}_{10}^2)_k(l^1)=\frac{\sum_{j=1}^{l^1}(\tilde{\sigma}_{10}^2)_k^j}{l^1}$ as the current estimates for the variances under each hypothesis. These values also play the rule of the initial points for the next learning iteration. As a result, in a recursive manner, the next learned values at $l^i+1$ for the means may be determined as follows,

\begin{align}
&(\tilde{\mu}_{00})_k(l^0+1)=\frac{l^0}{l^0+1}(\tilde{\mu}_{00})_k(l^0)+\frac{1}{(l^0+1)L_0}\sum_{i=1}^{L_0}(t_k^i)^0,\\&
(\tilde{\mu}_{10})_k(l^1+1)=\frac{l^1}{l^1+1}(\tilde{\mu}_{10})_k(l^1)+\frac{1}{(l^1+1)L_1}\sum_{i=1}^{L_1}(t_k^i)^1,
\end{align}
where $(t_k^i)^0$ and $(t_k^i)^1$ respectively represent the next set of test statistic data points received under the hypothesis $H_0$ or $H_1$ for the next learning interval $l=1,...,L_p$ from agent $G_k$. This recursive algorithm will require the estimation framework to only record the true values of the last $l^0$ or $l^1$ estimates respectively in a queue, and by calculating the new estimate at $l^i+1$, the first element of the respective queue may be discarded and the new estimate may be added to the queue. This means that the learning process will make use of $l^i \times L_p \times L$ data points while only storing $l^i$ data points for each hypothesis $H_i$ ($i=0,1$), respectively. Similarly, the process for learning the variances becomes,
\begin{align}
&(\tilde{\sigma}_{00}^2)_k(l^0+1)=\frac{l^0}{l^0+1}[(\tilde{\sigma}_{00}^2)_k(l^0)]\nonumber+\frac{1}{(l^0+1)L_0}\sum_{i=1}^{L_0}((t_k^i)^0-(\tilde{\mu}_{00})_k(l^0+1))^2,\numberthis\\&
(\tilde{\sigma}_{10}^2)_k(l^1+1)=\frac{l^1}{l^1+1}[(\tilde{\sigma}_{10}^2)_k(l^1)]\nonumber+\frac{1}{(l^1+1)L_1}\sum_{i=1}^{L_1}((t_k^i)^1-(\tilde{\mu}_{10})_k(l^1+1))^2,\numberthis
\end{align}
where $(t_k^i)^0$ and $(t_k^i)^1$ respectively represent the next set of test statistic data points received under the hypothesis $H_0$ or $H_1$ for the next learning interval $l=1,...,L_p$ from agent $G_k$. Lastly, as mentioned, as a design matter, one can put performance criteria such as $L_0>\tau_1$ and $L_1>\tau_2$ as quantities to be met first before the learning process for each of parameter sets under each hypothesis starts in order to make sure that the learning data-set is large enough for a more precise estimation. This two-level estimation process, will achieve a very good learning-based estimates while maintaining low memory requirements, as only the values of the last $l^i$ estimates (containing the information for $L \times L_p$ data points) and their respective hypothesis keys are required to be memorized in a feedback, recursive manner.

For Byzantine agents, we take another common approach, previously utilized in control literature \cite{david1999maximum,bresler1986exact}, called maximum liklihood method (MLE) of parameter estimation. This is due to the fact that additional to the means and variances of the Byzantine data, one needs to also estimate the latent variable $P_k$ ($k \in \mathcal{N}_j^{in_B}$) or the probability of the attack.  MLE, developed by Fisher \cite{fisher1997absolute}, has many desirable theoretical properties, such as consistency, efficiency and unbiasedness under certain conditions \cite{bickel2015mathematical}. The likelihood is the joint probability of a set of observations, conditioned on
a choice for the parameters $\tilde{\theta}_1$, $Lik(\tilde{\theta}_1,y)=P(y|\tilde{\theta}_1)$, where $y$ represents the data sample points, $\tilde{\theta}_1$ is the set of parameters to be estimated, and P is the probability distribution. According to this relation, the parameter set ($\tilde{\theta}_1^{MLE}$) that
maximizes the likelihood of the observed data gives the best estimator. This value is called the maximum likelihood estimate (MLE),
\begin{align*}
\tilde{\theta}_1^{MLE}= arg max_{\tilde{\theta}_1}~ Lik(\tilde{\theta}_1,y).
\end{align*}
Each learning phase $t$ consists of $L_p$ data-points. We denote the test statistic between the agent $G_j$ and the Byzantine neighbor $G_k$ during the learning phase of length $L_p$ as, $T_k^{(t)}=[(t_k^1)^0,(t_k^2)^0,...,(t_k^{L_0})^0,(t_k^1)^1,(t_k^2)^1,...,(t_k^{L_1})^1]$. Similar to before $L_0+L_1=L_p$. Additionally, similar to before, one can start the learning process for each set of parameters under each hypothesis $H_i$ ($i=0,1$), once $L_0>\tau_1$ and $L_1>\tau_2$. We utilize the estimates resulting from each learning phase as initial values for the next round of estimations. Additionally, we utilize the Expectation-Maximization (EM) algorithm for the learning process. The Expectation-Maximization (EM) algorithm is an iterative learning-based method for finding $\tilde{\theta}_1^{MLE}$ \cite{bickel2015mathematical}. The EM algorithm alternates between two steps, an expectation step which calculates the expectation of the log-likelihood given the current estimates for the parameters, and a maximization step which computes the parameters that maximize the expected log-likelihood found in the first step \textemdash This step involves derivations with respect to unknown parameters (means and variances) and substitutions for the latent parameter set $Z$. These new values then initialize the next expectation step. We annotate the latent parameters as $Z=[z_0,z_1]$  (in our case, the latent variables represent the attack probabilities for the Byzantine agent $G_k$, i.e. $\tilde{P}_k$ and $1-\tilde{P}_k$). For Byzantine neighboring agent $G_k$, we have,
\begin{align*}
Lik(\tilde{\theta}_1,y)=Lik((\tilde{\mu}_{ij})_k,(\tilde{\sigma}_{ij}^2)_k;T_k,Z)=p(T_k,Z|(\tilde{\mu}_{ij})_k,(\tilde{\sigma}_{ij}^2)_k),
\end{align*}
where $T_k$ represents our data points, $p(.)$ is the joint PDF of the data points and latent variables conditioned on the parameters and $j=0,1$. Further, we can describe the above based on a marginal and a conditional distribution, this gives us,
\begin{align} \label{mullog0}
p(T_k,Z|(\tilde{\mu}_{ij})_k,(\tilde{\sigma}_{ij}^2)_k)=p(z_j|T_k,(\tilde{\mu}_{ij})_k,(\tilde{\sigma}_{ij}^2)_k)p(T_k|(\tilde{\mu}_{ij})_k,(\tilde{\sigma}_{ij}^2)_k)=\pi_k^jp(T_k|(\tilde{\mu}_{ij})_k,(\tilde{\sigma}_{ij}^2)_k),
\end{align}
where $0\leq\pi_k^j\leq1$ represents the distribution for the latent variable which is the probability of attack for the Byzantine agent $G_k$. Also, $j=0,1$ and $\pi_k^0+\pi_k^1=1$. To sum all this up, for estimating the Byzantine parameters of agent $G_k$, and in order to describe the above relationship based on single data points, we expand (\ref{mullog0}) to have,
\begin{align} \label{mullog}
Lik((\tilde{\mu}_{ij})_k,(\tilde{\sigma}_{ij}^2)_k;T_k,Z)=p(T_k,Z|(\tilde{\mu}_{ij})_k,(\tilde{\sigma}_{ij}^2)_k)=\prod_{n=1}^{L_p}\prod_{j=0}^{1}\pi_k^jp(t^n_k)^i|(\tilde{\mu}_{ij})_k,(\tilde{\sigma}_{ij}^2)_k).
\end{align}

Under the EM algorithm, for the first step, the expectation is calculated based on the log-likelihood function of the distributions. Given (\ref{mullog}), the expectation step ($Q$-function) based on the current estimate set $\tilde{\theta}_1^{(l)}$, for agent $G_k$ under the hypotheses $H_i$ $(i=0,1)$, becomes, 
\begin{align}\label{E}
&Q(\tilde{\theta}_1|\tilde{\theta}_1^{(l)})=E_{z|T_k;\tilde{\theta}_1^{(l)}}[\log Lik((\tilde{\mu}_{ij})_k,(\tilde{\sigma}_{ij}^2)_k;T_k,Z)]\nonumber\\&=E_{z|T_k;\tilde{\theta}_1^{(l)}}[\log p(T_k,Z|\tilde{\theta}_1)|T_k,\tilde{\theta}_1^{(l)}]\nonumber\\&=\sum_{n=1}^{L_0}\log[\sum_{j=0}^{1}\pi_k^jp((t_k^n)^0|(\tilde{\mu}_{0j})_k,(\tilde{\sigma}_{0j}^2)_k)p(z_j|(t_k^n)^0,(\tilde{\mu}_{0j})_k^{(l)},(\tilde{\sigma}_{0j}^2)_k^{(l)}]\nonumber\\&+\sum_{n=1}^{L_1}\log[\sum_{j=0}^{1}\pi_k^jp((t_k^n)^1|(\tilde{\mu}_{1j})_k,(\tilde{\sigma}_{1j}^2)_k)p(z_j|(t_k^n)^1,(\tilde{\mu}_{1j})_k^{(l)},(\tilde{\sigma}_{1j}^2)_k^{(l)}],
\end{align}
where $\tilde{\theta}_1^{(l)}=((\tilde{\mu}_{0j})_k,(\tilde{\sigma}_{0j}^2)_k, (\tilde{\mu}_{1j})_k,(\tilde{\sigma}_{1j}^2)_k,\pi_k^j)$ for $j=0,1$, are the current estimates for the neighboring Byzantine neighbor $G_k$. It is also well-known in the literature that given the current estimate $\theta_1^{(l)}$, the conditional distribution of $Z$, i.e. $p(z_j|(t_k^n)^r,(\tilde{\mu}_{rj})_k^{(l)},(\tilde{\sigma}_{rj}^2)_k^{(l)})$, under hypothesis $H_r$, $(r=0,1)$ respectively, for each summation in (\ref{E}). is determined by Bayes' Theorem as,
\begin{align}
p(z_j|(t_k^n)^r,(\tilde{\mu}_{rj})_k^{(l)},(\tilde{\sigma}_{rj}^2)_k^{(l)})=\frac{(\pi_k^j)^{(l)}p((t_k^n)^r|(\tilde{\mu}_{rj})_k^{(l)},(\tilde{\sigma}_{rj}^2)_k^{(l)})}{\sum_{s=0}^{1}(\pi_k^s)^{(l)}p((t_k^n)^r|(\tilde{\mu}_{rs})_k^{(l)},(\tilde{\sigma}_{rs}^2)_k^{(l)})}.
\end{align}

For the maximization step, we should maximize $Q(\tilde{\theta}_1|\tilde{\theta}_1^{(l)})$, by taking derivatives with respect to the parameters, i.e. $\tilde{\theta}_1^{(l+1)} =  arg~ max_{\tilde{\theta}_1}~Q(\tilde{\theta}_1|\tilde{\theta}_1^{(l)})$. Simplifying further, and utilizing Jensen's inequality and given the fact that log likelihood is a concave function \cite{bickel2015mathematical}, we have,
\begin{align}
&\tilde{\theta}_1^{(l+1)} =  arg~ max_{\tilde{\theta}_1}~Q(\tilde{\theta}_1|\tilde{\theta}_1^{(l)}) \nonumber\\& \equiv
arg~ max_{\tilde{\theta}_1}[\sum_{n=1}^{L_0}\sum_{j=0}^{1} [p(z_j|(t_k^n)^0,(\tilde{\mu}_{0j})_k^{(l)},(\tilde{\sigma}_{0j}^2)_k^{(l)})(\log \pi_k^j-\frac{((t_k^n)^0-(\tilde{\mu}_{0j})_k)^2}{2(\tilde{\sigma}_{0j}^2)_k}-\frac{\log (\tilde{\sigma}_{0j}^2)_k}{2})]\nonumber\\&+\sum_{n=1}^{L_1}\sum_{j=0}^{1}[p(z_j|(t_k^n)^1,(\tilde{\mu}_{1j})_k^{(l)},(\tilde{\sigma}_{1j}^2)_k^{(l)})(\log \pi_k^j-\frac{((t_k^n)^1-(\tilde{\mu}_{1j})_k)^2}{2(\tilde{\sigma}_{1j}^2)_k}-\frac{\log (\tilde{\sigma}_{1j}^2)_k}{2})]].
\end{align}

This should be done subject to the constraint that $\sum_{j=0}^{1} \pi_k^j = 1$ for the Byzantine agent $G_k$. Similar to common approaches in literature in regard to EM-based estimation of Gaussian mixtures \cite{kung2005biometric}, we utilize a Lagrangian multiplier for maximization, hence we have,
\begin{align}
&\max \mathcal{J} = Q(\tilde{\theta}_1|\tilde{\theta}_1^{(l)}) + \lambda (\sum_{j=0}^{1} \pi_k^j - 1)\nonumber\\&\equiv~
\max [\sum_{n=1}^{L_0}\sum_{j=0}^{1} [p(z_j|(t_k^n)^0,(\tilde{\mu}_{0j})_k^{(l)},(\tilde{\sigma}_{0j}^2)_k^{(l)})(\log \pi_k^j-\frac{((t_k^n)^0-(\tilde{\mu}_{0j})_k)^2}{2(\tilde{\sigma}_{0j}^2)_k}-\frac{\log (\tilde{\sigma}_{0j}^2)_k}{2})]\nonumber\\&+\sum_{n=1}^{L_1}\sum_{j=0}^{1}[ p(z_j|(t_k^n)^1,(\tilde{\mu}_{1j})_k^{(l)},(\tilde{\sigma}_{1j}^2)_k^{(l)})(\log \pi_k^j-\frac{((t_k^n)^1-(\tilde{\mu}_{1j})_k)^2}{2(\tilde{\sigma}_{1j}^2)_k}-\frac{\log (\tilde{\sigma}_{1j}^2)_k}{2})]\nonumber\\&+ \lambda (\sum_{j=0}^{1} \pi_k^j - 1)].
\end{align}

In order to maximize the above, one should solve for the equations resulting from the derivative of each parameter by equating them with zero. As an example, we have (for $j=0,1$), 
\begin{align}
\frac{d}{d\pi_k^j} \mathcal{J}=\lambda+\frac{\sum_{n=1}^{L_0}p(z_j|(t_k^n)^0,(\tilde{\mu}_{0j})_k^{(l)},(\tilde{\sigma}_{0j}^2)_k^{(l)})}{\pi_k^j}+\frac{\sum_{n=1}^{L_1}p(z_j|(t_k^n)^1,(\tilde{\mu}_{1j})_k^{(l)},(\tilde{\sigma}_{1j}^2)_k^{(l)})}{\pi_k^j}=0,
\end{align}
which gives us, 
\begin{align}
-\pi_k^j\lambda=\sum_{n=1}^{L_0}p(z_j|(t_k^n)^0,(\tilde{\mu}_{0j})_k^{(l)},(\tilde{\sigma}_{0j}^2)_k^{(l)})+\sum_{n=1}^{L_1}p(z_j|(t_k^n)^1,(\tilde{\mu}_{1j})_k^{(l)},(\tilde{\sigma}_{1j}^2)_k^{(l)}),
\end{align}
or $-\pi_k^j\lambda=(L_0 + L_1)\pi_k^j$. And we have, $\lambda=-L_p$. In a similar manner, we can take derivatives and simplify further to find the following recursive estimations for the Byzantine parameters for agent $G_k$, 
\begin{align}
	&(\pi_k^j)^{(l+1)}=\frac{\sum_{n=1}^{L_0}p(z_j|(t_k^n)^0,(\tilde{\mu}_{0j})_k^{(l)},(\tilde{\sigma}_{0j}^2)_k^{(l)})+\sum_{l=1}^{L_1}p(z_j|(t_k^n)^1,(\tilde{\mu}_{1j})_k^{(l)},(\tilde{\sigma}_{1j}^2)_k^{(l)})}{L_p},\\&
	(\tilde{\mu}_{0j})_k^{(l+1)}=\frac{\sum_{n=1}^{L_0}p(z_j|(t_k^n)^0,(\tilde{\mu}_{0j})_k^{(l)},(\tilde{\sigma}_{0j}^2)_k^{(l)})(t_k^n)^0}{\sum_{n=1}^{L_0}p(z_j|(t_k^n)^0,(\tilde{\mu}_{0j})_k^{(l)},(\tilde{\sigma}_{0j}^2)_k^{(l)})},\\&
	(\tilde{\sigma}_{0j}^2)_k^{(l+1)}=\frac{\sum_{n=1}^{L_0}p(z_j|(t_k^n)^0,(\tilde{\mu}_{0j})_k^{(l)},(\tilde{\sigma}_{0j}^2)_k^{(l)})((t_k^n)^0-(\tilde{\mu}_{0j})_k^{(l+1)})^2}{\sum_{n=1}^{L_0}p(z_j|(t_k^n)^0,(\tilde{\mu}_{0j})_k^{(l)},(\tilde{\sigma}_{0j}^2)_k^{(l)})},\\&
	(\tilde{\mu}_{1j})_k^{(l+1)}=\frac{\sum_{n=1}^{L_1}p(z_j|(t_k^n)^1,(\tilde{\mu}_{1j})_k^{(l)},(\tilde{\sigma}_{1j}^2)_k^{(l)})(t_k^n)^1}{\sum_{n=1}^{L_1}p(z_j|(t_k^n)^1,(\tilde{\mu}_{1j})_k^{(l)},(\tilde{\sigma}_{1j}^2)_k^{(l)})},\\&
	(\tilde{\sigma}_{1j}^2)_k^{(l+1)}=\frac{\sum_{n=1}^{L_1}p(z_j|(t_k^n)^1,(\tilde{\mu}_{1j})_k^{(l)},(\tilde{\sigma}_{1j}^2)_k^{(l)})((t_k^n)^1-(\tilde{\mu}_{1j})_k^{(l+1)})^2}{\sum_{n=1}^{L_1}p(z_j|(t_k^n)^1,(\tilde{\mu}_{1j})_k^{(l)},(\tilde{\sigma}_{1j}^2)_k^{(l)})}.
\end{align}

Similar to the previous algorithm, one can design independent performance criteria such as $L_0>\tau_1$ and $L_1>\tau_2$ as quantities to be met first before the learning process for each of the parameter sets starts. This is to make sure that the learning data-set is large enough for a precise estimation. The algorithm may be solved recursively. At discrete time instance $l$ when enough information is received the recursive learning algorithm starts and the estimates at the end of each expectation-maximization run may be used as initial values for the next learning iteration $l+1$ that uses a new set of $L_p$ data points.  
After the learning process has ended, the honest agent may classify its neighbor $G_k$ as a Byzantine or honest agent following the likelihood-based hypothesis testing,
\begin{align*} 
\frac{\tilde{f}_{PDF}(T_k|(Dec^i_k)_0)}{\tilde{f}_{PDF}(T_k|(Dec^i_k)_1)} \gtrless^H_B 1,
\end{align*}
where $\tilde{f}_{PDF}(.)$ represents the probability distribution function attained based on the best estimates of the parameters available to the agent. 

\bf{Mitigating the Effects of Data Falsification}:\normalfont ~ Once the Byzantine agents are identified based on the above algorithm. One can utilize this information to mitigate the effects of the attack. Unlike most approaches in the literature that rely on excluding the Byzantine agents, we utilize the Byzantine information against rogue agents in order to benefit the entire event-triggered multi-agent system. As the first step, for the decision making step, we define a new local summary statistic based on the information received from only the honest agents, i.e. $(T^\star_j)^H= \sum_{k \in  \mathcal{N}_j^{in_H}}^{}T_k^j$. Similar to before, each agent will make its own decision on the synchronization hypothesis using the predefined threshold $\gamma_j^H$,
\[ Decision_{syn}=\begin{cases}
H_0   & \quad \text{if $(T_j^*)^H<\gamma_j^H$ } \\
H_1  & \quad \text{otherwise.}\end{cases}\]

Under the hypothesis $H_0$, at the learning iteration $l+1$, the honest agent would estimate the means ($(\tilde{\mu}_{00})_k^{(l+1)}$ and $(\tilde{\mu}_{01})_k^{(l+1)}$) based on the received information from the Byzantine neighbor $G_k$ or estimate only $(\tilde{\mu}_{00})_k^{(l+1)}$ based on the received information from the honest neighbor $G_k$ according to the algorithm given in the previous sub-section. These estimates follow the form given in (\ref{D1}). As a result, the honest agent may closely estimate the attack parameter $\tilde{\Delta}_k^{l+1}$ after each learning iteration for the Byzantine neighbor $G_k$ as follows,
\begin{align} \label{delta1}
&\tilde{\Delta}_k^{l+1} \approx \sqrt{ \frac{(\tilde{\mu}_{01})_k^{(l+1)}-(\tilde{\mu}_{00})_k^{(l+1)}}{L\tilde{\sigma}_k^2 \tilde{h}_k^2}}.
\end{align}

In a similar manner, under the hypothesis $H_1$, an analogous estimation process may be undertaken by the honest agent. Here, it is important to note that,
\begin{align} \label{md}
(\tilde{\mu}_{11})_k^{(l+1)}-(\tilde{\mu}_{10})_k^{(l+1)}=L\tilde{h}^2_k \Delta^2_k+2L\tilde{h}_k\Delta_k(\mu_j-\tilde{\mu}_k),
\end{align}
where $\mu_j=\frac{1}{L}\sum_{i=1}^{L}y_j^i$. Since $\mu_k$ is not available to us, it is prudent to use the estimation $\tilde{\mu}_k=\frac{1}{L}\sum_{i=1}^{L}\mu^i$ where $\mu^i=\frac{1}{|\mathcal{N}^{in}|}\sum_{k \in \mathcal{N}^{in}}^{}y_k^i$ at time-instance $i$ during the detection interval of length $L$. $\tilde{\mu}_k$ can provide us with a good initial value as it represent the general state the entire multi-agent system is at. Later, at the learning iteration $l+1$, one can replace $\mu_k$ with $\frac{1}{L}\sum_{i=1}^{L}(\tilde{y}_k^i-\tilde{\Delta}_k^{l})$. Based on (\ref{md}), we also have, 
\begin{align} \label{md2}
\tilde{h}_k \Delta^2_k+2\Delta_k(\mu_j-\tilde{\mu}_k)=(\sqrt{\tilde{h}_k}\Delta_k+\frac{(\mu_j-\tilde{\mu}_k)}{\sqrt{\tilde{h}_k}})^2-\frac{(\mu_j-\tilde{\mu}_k)^2}{\tilde{h}_k}.
\end{align}
Finally, by utilizing (\ref{md}) and (\ref{md2}), we have,
\begin{align} \label{delta2}
\tilde{\Delta}_k^{l+1} \approx\frac{1}{\sqrt{\tilde{h}_k}}(\sqrt{ \frac{(\tilde{\mu}_{11})_k^{(l+1)}-(\tilde{\mu}_{10})_k^{(l+1)}}{L\tilde{h}_k}+\frac{(\mu_j-\tilde{\mu}_k)^2}{\tilde{h}_k}}-\frac{(\mu_j-\tilde{\mu}_k)}{\sqrt{\tilde{h}_k}}).
\end{align}

In the above relations, we have assumed that $\tilde{\sigma}_k$ and $\tilde{h}_k$ for the communication links are available to the agents. As mentioned before, these assumptions are justified by the fact that each detection unit can perform simple noise power estimation and channel gain estimation (by averaging the signal-to-noise ratio over a certain time interval) between consecutive sensing intervals to accurately obtain these values. Finally, instead of excluding the Byzantine agent $G_k$ in the process of mitigating the attack, one can utilize the false information after the estimation and mitigate the negative adversarial effects after the learning iteration $l+1$ by replacing the Byzantine agent's output information with $\tilde{y}_k=y_k \mp \tilde{\Delta}^{l+1}_k$ under the hypothesis $H_i$ $(i=0,1)$, respectively. Next, we will demonstrate our approach with an example.
\begin{example}\label{exm:ex5}
We consider a multi-agent event-triggered network system consisting of four agents ($i=1,...,4$) with the underlying balanced communication topology given in Fig. \ref{fig:ex5}. We assume the following dynamics for agents, 
	\[ G_{i} = \begin{cases}
	\dot{x}_i(t)= - c_i x_i(t) + u_i(t)   \\
	y_i(t)=x_i(t), & \\
	\end{cases} \]
	with $c_1=1.2, c_2=2.2, c_3=2.4, c_4=0.60$ and initial conditions, $y_1(0)=5, y_2(0)=10, y_3(0)=4, y_4(0)=1$. One can verify that all agents are dissipative with the storage function $V_i(x)=\frac{1}{2}x_i^T(t)x_i(t)$. This results in output passivity indices $\rho_1=1.2, \rho_2=2.2, \rho_3=2.4, \rho_4=0.6$. The Laplacian matrix of the underlying communication graph amongst agents before Byzantine attack is balances as follows,
	\begin{align*}
		L=\begin{bmatrix} 
			1 & 0  &  -1 &  0 \\
			0 &  1  & -1 &  0 \\
		   -1 &  0  &  2 & -1 \\
		    0 & -1  &  0 &  1
		\end{bmatrix}.
	\end{align*}
	with the connectivity measure, $\lambda(G)=2$. We assume that the multi-agent network system is designed based on Theorem \ref{thm:synchAss} with the following triggering conditions,
	\begin{align*}
		&||e_1(t)||_2^2>0.21||y_1(t)||_2^2,\\&
		||e_2(t)||_2^2>0.14||y_2(t)||_2^2,\\&
		||e_3(t)||_2^2>0.20||y_3(t)||_2^2,\\&
		||e_4(t)||_2^2>0.45||y_4(t)||_2^2,
	\end{align*}
An additive Gaussian noise with zero mean and variance $\sigma_k^2=1.22$ $(\mathcal{N}(0,\sigma_k^2))$ is assumed in the communication links. The channel gains are $\tilde{h}_k=1$ for $k=1,...,4$. We assume $G_1$ is a Byzantine agent. Under the hypothesis $H_0$, at time $t=3s$, $G_1$ attacks the network with the attack parameters $P_1=0.70$ and $\Delta_1=8$. One can see that the behavior of the multi-agent system drastically deviates from its desired synchronized behavior as a result of the attack (Fig. \ref{fig:ex5out} and Fig. \ref{fig:ex5dev}).
The honest agent $G_3$ detects the Byzantine agent and starts the process of learning the Byzantine agent's behavior using the proposed mitigation algorithm. The learning parameters are $L=12$ with $20$ learning iterations ($l=20$) of length $20$ ($L_p=20$) which takes advantage of an overall of $400$ data points. At time $t=6s$, the honest agent $G_3$ estimates the attack parameters as $\tilde{P}_1=0.68$ (Fig. \ref{fig:ex5prob}) and $\tilde{\Delta}_1=6.35$ using the proposed algorithm given in the previous sub-section and the relation given below,
\begin{align*} 
	\tilde{\Delta}_1^{l+1} \approx\frac{1}{\sqrt{\tilde{h}_1}}(\sqrt{ \frac{(\tilde{\mu}_{11})_1^{(l+1)}-(\tilde{\mu}_{10})_1^{(l+1)}}{L\tilde{h}_1}+\frac{(\mu_3-\tilde{\mu}_1)^2}{\tilde{h}_1}}-\frac{(\mu_3-\tilde{\mu}_1)}{\sqrt{\tilde{h}_1}}).
\end{align*}
The estimations at each learning iteration are given in Fig \ref{fig:ex5prob}, Fig \ref{fig:ex5m1} and Fig \ref{fig:ex5m2}. The mitigation process starts at time $t=8$ where the information received from agent $G_1$ by agent $G_3$ is replaced with $\tilde{y}_1=y_1-\tilde{\Delta}_1^{l+1}$. One can see the positive effects of this mitigation process in Fig. \ref{fig:ex5out} and Fig. \ref{fig:ex5dev} toward the end of the experiment where the multi-agent system enhances its performance and reaches synchronization again.
\begin{figure}[!t]
	\centering
	\includegraphics[scale = 0.7]{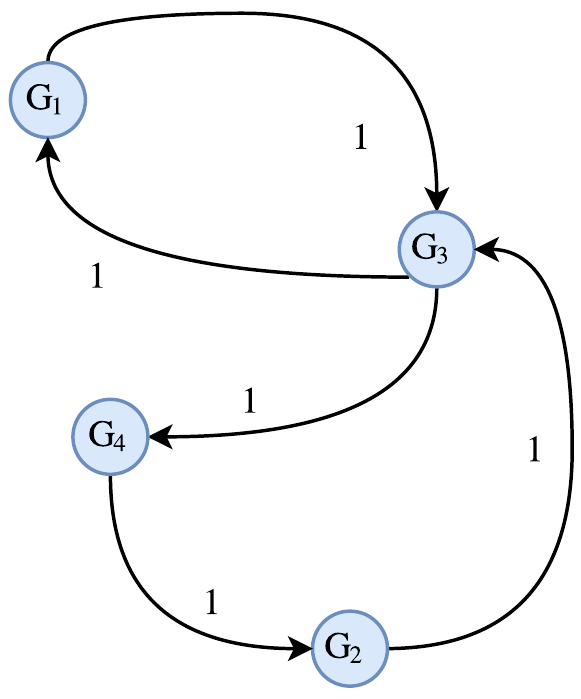}
	\caption{The Underlying Communication Graph for the Multi-Agent System Presented in Example 5.}
	\label{fig:ex5}
\end{figure}
\begin{figure}[!t]
	\centering
	\includegraphics[scale = 0.5]{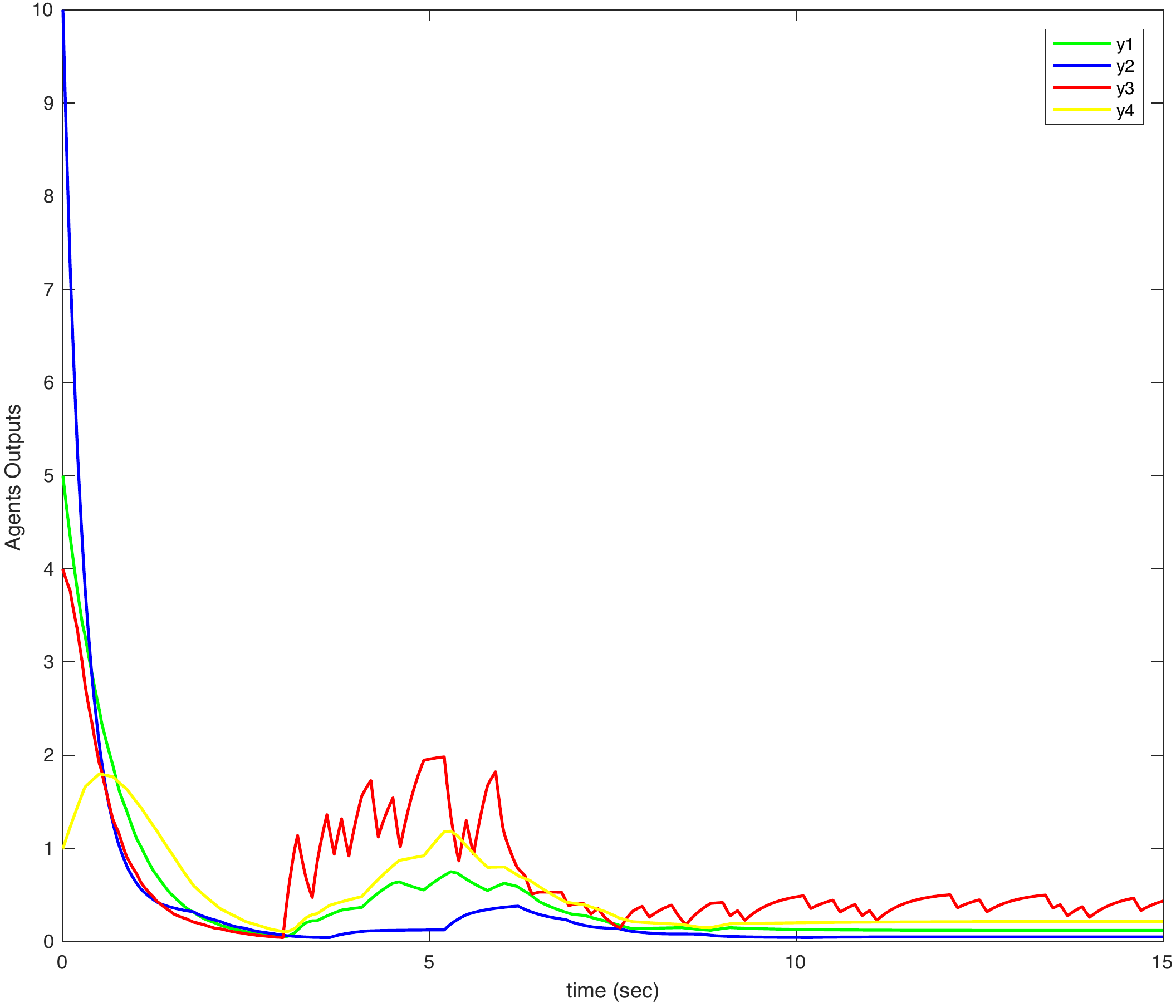}
	\caption{The Outputs of the Multi-Agent System in the Presence of the Byzantine Attack (Attack Parameters: $P_1=0.70$ and $\Delta_1=8$).}
	\label{fig:ex5out}
\end{figure}
\begin{figure}[!t]
	\centering
	\includegraphics[scale = 0.5]{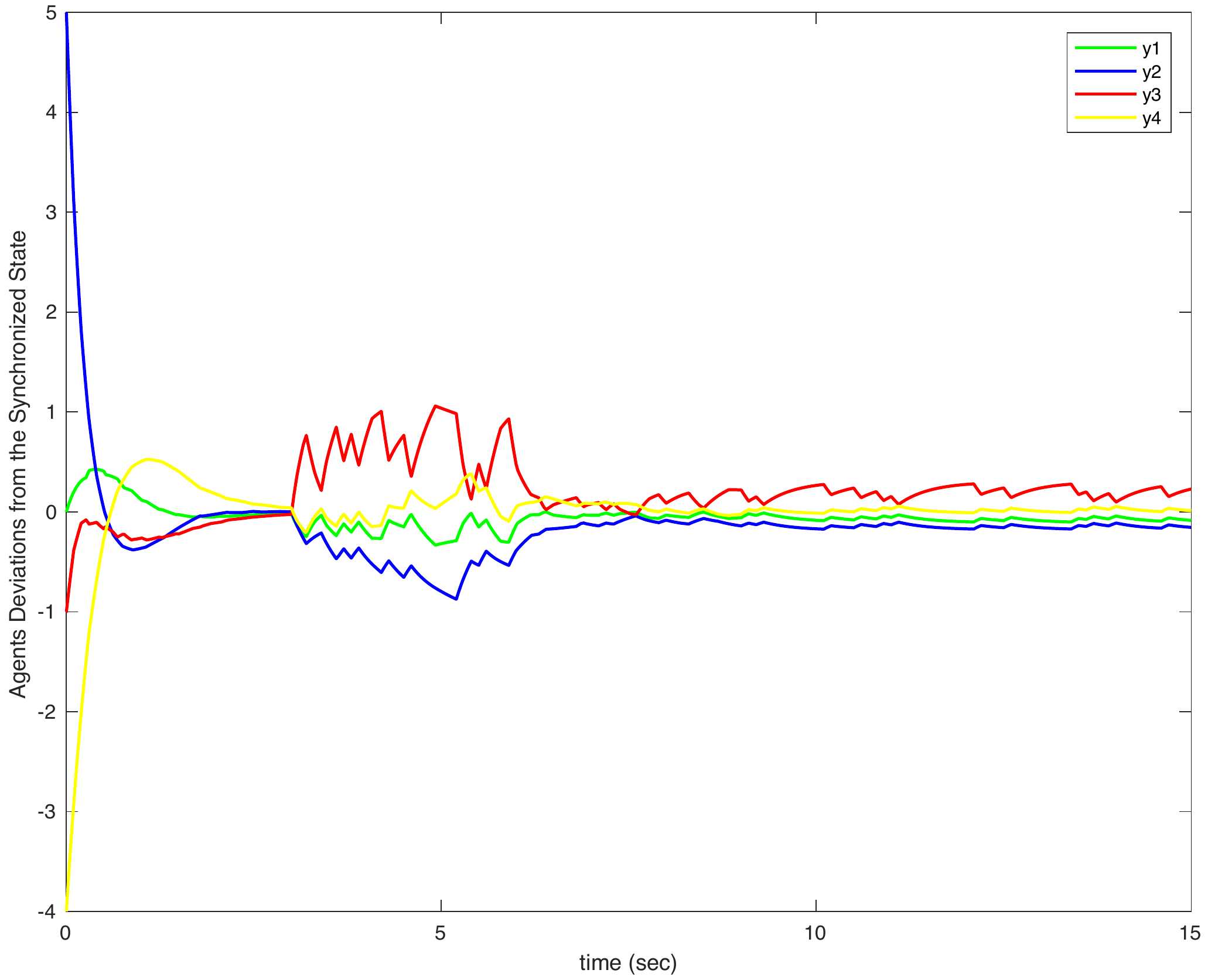}
	\caption{The Deviations of the Outputs from the Synchronized State of the Multi-Agent System in the Presence of the Byzantine Attack (Attack Parameters: $P_1=0.70$ and $\Delta_1=8$).}
	\label{fig:ex5dev}
\end{figure}
\begin{figure}[!t]
	\centering
	\includegraphics[scale = 0.5]{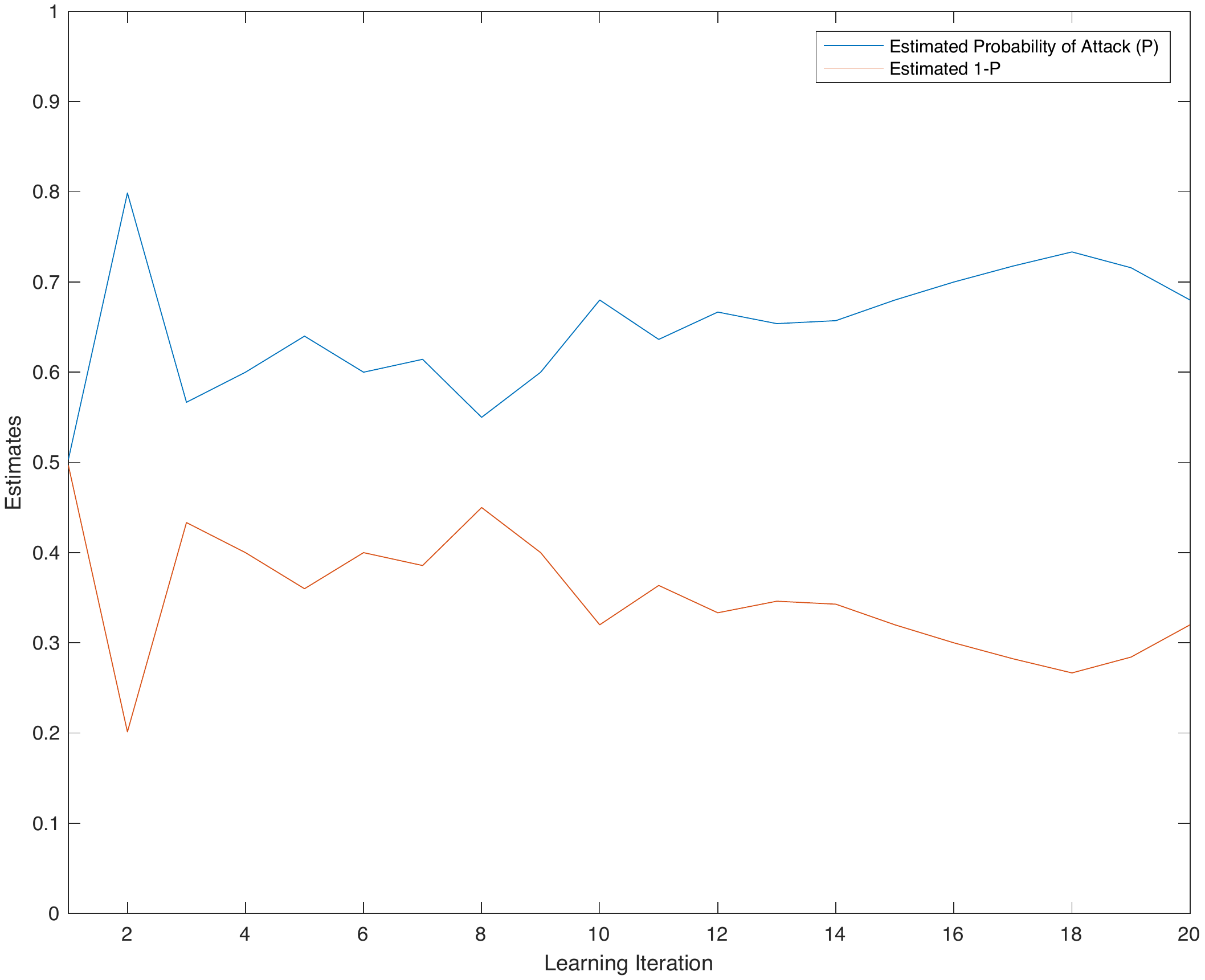}
	\caption{The Estimated Probability of Attack Using the Proposed Algorithm.}
	\label{fig:ex5prob}
\end{figure}
\begin{figure}[!t]
	\centering
	\includegraphics[scale = 0.5]{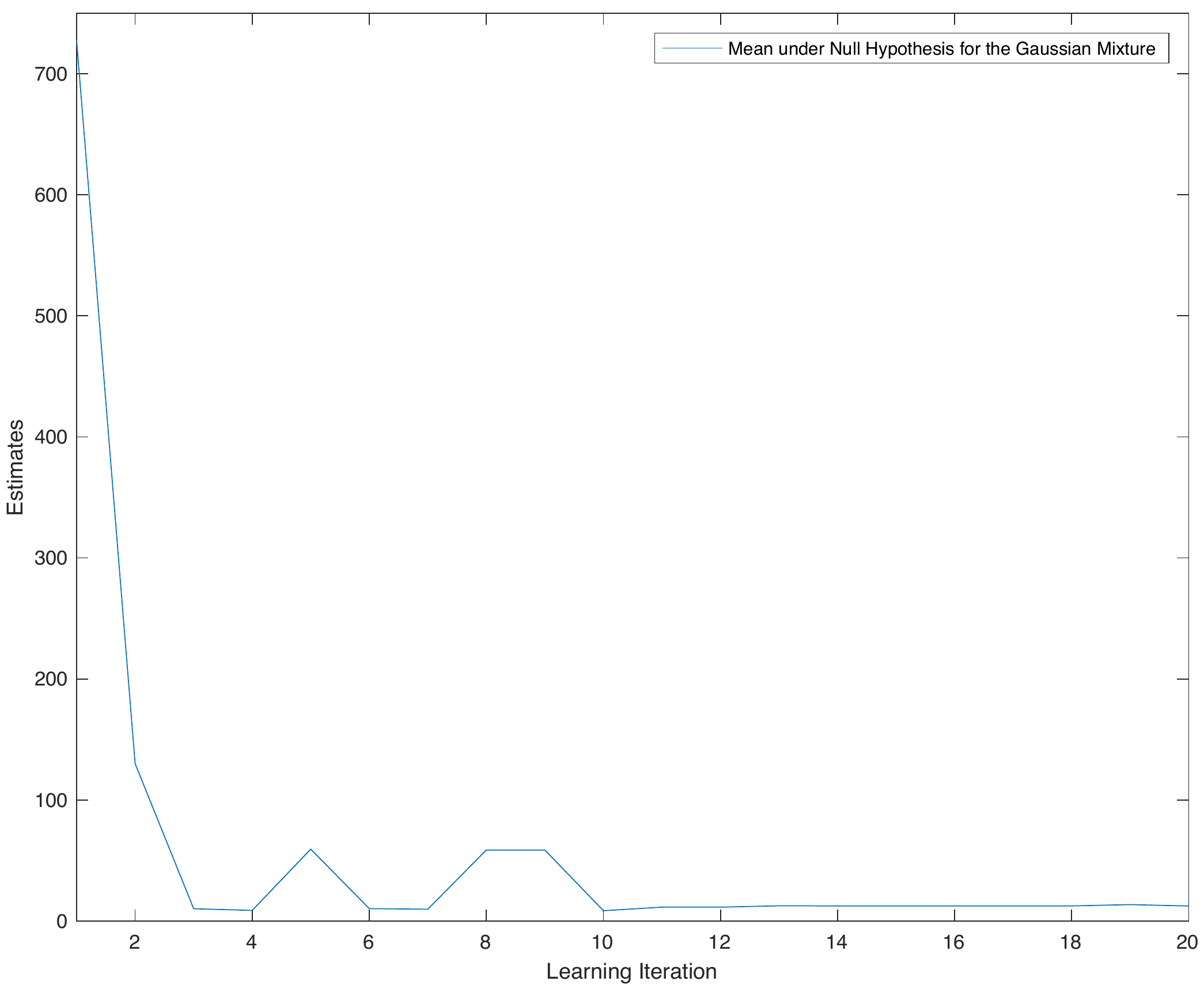}
	\caption{The Estimated Mean $(\tilde{\mu}_{00})_1^{(t+1)}$, under $H_0$ Using the Proposed Algorithm.}
	\label{fig:ex5m1}
\end{figure}
\begin{figure}[!t]
	\centering
	\includegraphics[scale = 0.5]{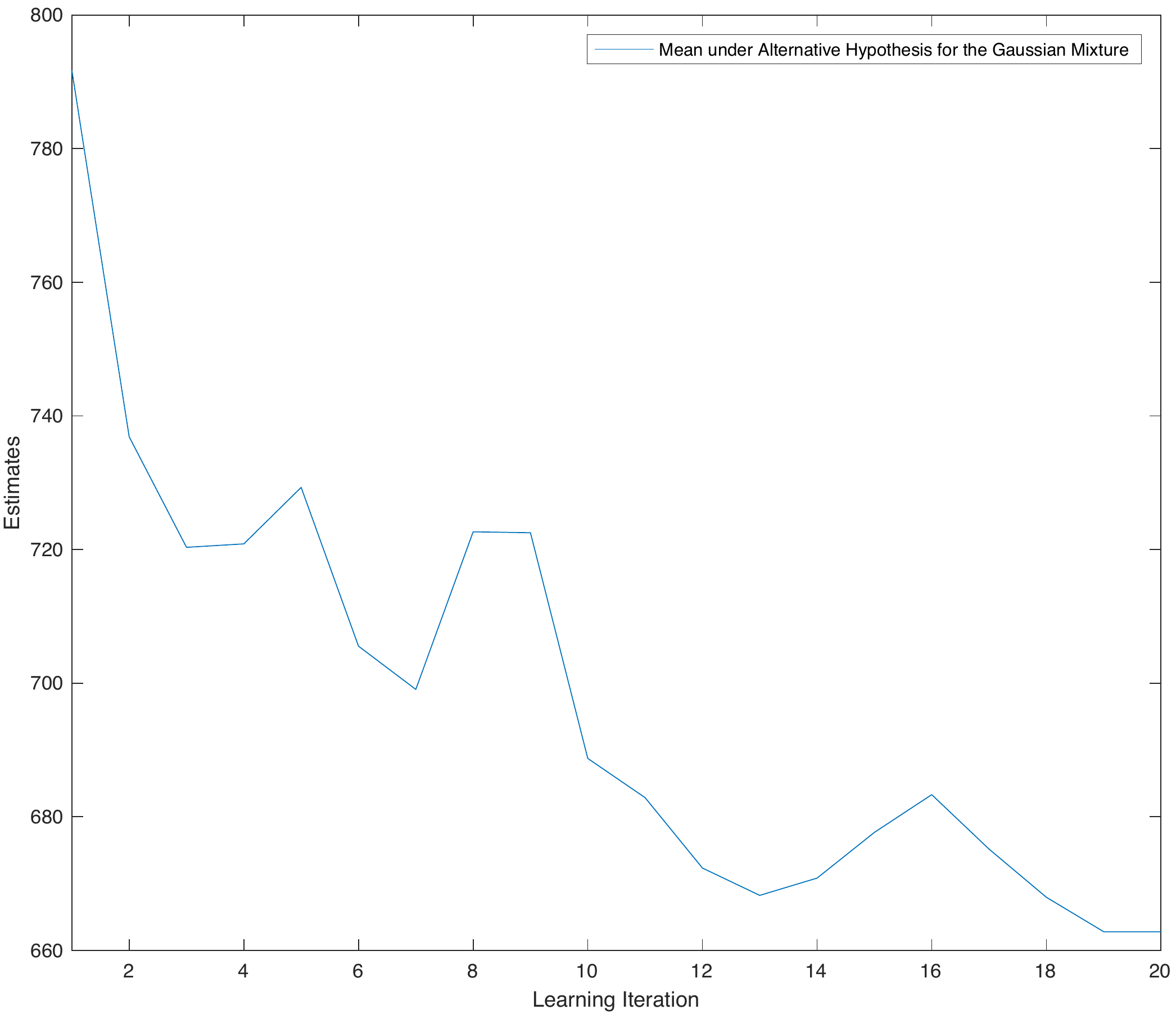}
	\caption{The Estimated Mean $(\tilde{\mu}_{01})_1^{(t+1)}$, under $H_1$ Using the Proposed Algorithm..}
	\label{fig:ex5m2}
\end{figure}
\end{example}
\clearpage
\section{Concluding Remarks}\label{sec:con}
The work presented in this paper may be divided into two parts. The first part consists of a comprehensive event-triggered control design proposal that can guarantee synchronization for a network of multi-agent systems based on their passivity properties. This proposed control design is capable of reducing the communication load amongst sub-agents while maintaining synchronization and desired performance criteria. Under this part of the work, we also show the lack of Zeno behavior for the event-triggered conditions. The second part of our work concerns security. Under this section, we introduced a general powerful model for Byzantine attacks containing both data falsification and weight manipulation. Additionally, we introduced a detection framework, through which, the honest agents will attempt to detect and mitigate the effects of the attack. We gave a full performance analysis of the detection unit based on both transient and steady-state characteristics of the framework. Lastly, we presented two learning-based mitigation processes. The first one was based on the passivity properties of the agents and intended to mitigate the negative effects of weight manipulation. The second proposed learning-based control framework dealt with the problem of data falsification. Under this framework, the honest agents attempt to estimate their neighbor's states and consequently learn the attack parameters for Byzantine neighbors. After learning these parameters then the honest agents utilize this information to eradicate the negative effects of adversarial attempts and enhance the performance and synchronization of the entire event-triggered multi-agent network system.

\newpage
\bibliographystyle{IEEEtran}
\bibliography{IEEEabrv,bibfile}
\end{document}